\documentclass[dvipsnames,11pt]{article}
\usepackage{amsmath,amssymb,amsthm,color,bm,mathrsfs,extarrows,tikz,graphicx,mathtools,enumitem,fancybox,makecell,pifont}
\usepackage[square,sort,comma,numbers]{natbib}
\usepackage{subcaption}
\usepackage[ruled,algo2e]{algorithm2e}
\usepackage[margin=1.in]{geometry}
\usepackage{xcolor,bbm}
\usepackage[pagebackref]{hyperref}
\allowdisplaybreaks

\setlist[itemize]{itemsep=0pt}
\setlist[enumerate]{itemsep=0pt}
\hypersetup{colorlinks=true,urlcolor=blue,linkcolor=blue,citecolor=[rgb]{.42,.56,.14},}
\usetikzlibrary{decorations.pathreplacing,decorations.markings,decorations.pathmorphing,decorations.shapes,arrows.meta,positioning}

\SetKwProg{Procedure}{Procedure}{:}{end}
\SetKwFunction{SolutionSampling}{SolutionSampling}
\SetKwFunction{RejectionSampling}{RejectionSampling}
\SetKwFunction{MarginSample}{MarginSample}
\SetKwFunction{MarginOverflow}{MarginOverflow}
\SetKwFunction{TruncMarginSample}{TruncMarginSample}
\SetKwFunction{TruncMarginOverflow}{TruncMarginOverflow}
\SetKwFunction{ConstructSep}{ConstructSep}
\SetKwFunction{BernoulliFactory}{BernoulliFactory}

\usepackage[capitalise,nameinlink]{cleveref}
\Crefname{lemma}{Lemma}{Lemmas}
\Crefname{fact}{Fact}{Facts}
\Crefname{theorem}{Theorem}{Theorems}
\Crefname{corollary}{Corollary}{Corollaries}
\Crefname{claim}{Claim}{Claims}
\Crefname{example}{Example}{Examples}
\Crefname{problem}{Problem}{Problems}
\Crefname{definition}{Definition}{Definitions}
\Crefname{notation}{Notation}{Notations}
\Crefname{assumption}{Assumption}{Assumptions}
\Crefname{subsection}{Subsection}{Subsections}
\Crefname{section}{Section}{Sections}
\Crefformat{equation}{(#2#1#3)}

\newtheorem{theorem}{Theorem}[section]
\newtheorem*{theorem*}{Theorem}

\newtheorem{proposition}[theorem]{Proposition}
\newtheorem*{proposition*}{Proposition}
\newtheorem{lemma}[theorem]{Lemma}
\newtheorem*{lemma*}{Lemma}
\newtheorem{corollary}[theorem]{Corollary}
\newtheorem*{corollary*}{Corollary}
\newtheorem*{conjecture*}{Conjecture}
\newtheorem{fact}[theorem]{Fact}
\newtheorem*{fact*}{Fact}

\newtheorem*{exercise*}{Exercise}

\newtheorem*{hypothesis*}{Hypothesis}

\theoremstyle{definition}
\newtheorem{definition}[theorem]{Definition}

\newtheorem{assumption}[theorem]{Assumption}
\newtheorem{exercise-easy}[theorem]{Exercise}
\newtheorem{exercise-med}[theorem]{Exercise}
\newtheorem{exercise-hard}[theorem]{Exercise$^\star$}
\newtheorem{claim}[theorem]{Claim}
\newtheorem*{claim*}{Claim}

\newtheorem{remark}[theorem]{Remark}
\newtheorem*{remark*}{Remark}

\newtheorem*{observation*}{Observation}

\DeclareMathOperator*{\E}{\mathbb E}

\renewcommand{\Pr}{\operatorname*{\mathbf{Pr}}}

\newcommand{\eps}{\varepsilon}
\newcommand{\abs}[1]{\left| #1 \right|}

\newcommand{\pbra}[1]{\left( #1 \right)}
\newcommand{\sbra}[1]{\left[ #1 \right]}
\newcommand{\cbra}[1]{\left\{ #1 \right\}}
\newcommand{\ceilbra}[1]{\left\lceil #1 \right\rceil}
\newcommand{\floorbra}[1]{\left\lfloor #1 \right\rfloor}
\renewcommand{\mid}{\,\middle\vert\,}
\newcommand{\bin}{\{0,1\}}

\newcommand{\True}{\mathsf{True}}
\newcommand{\False}{\mathsf{False}}
\newcommand{\TVdist}{d_{\mathsf{TV}}}

\newcommand{\poly}{\mathsf{poly}}

\newcommand{\vbl}{\mathsf{vbl}}

\newcommand{\indicator}{\mathsf{1}}
\newcommand{\Poi}{\mathsf{Poi}}
\newcommand{\Qmark}{\text{\ding{72}}}
\newcommand{\EQmark}{\text{\ding{73}}}
\newcommand{\depth}{\mathsf{depth}}

\newcommand{\Ccalsep}{\mathcal{C}_\mathsf{sep}}
\newcommand{\Vcalsep}{\mathcal{V}_\mathsf{sep}}
\newcommand{\CcalQmark}{\mathcal{C}_\Qmark}

\newcommand{\Ccalfrozen}{\mathcal{C}_\mathsf{frozen}}

\newcommand{\Ccalbad}{\mathcal{C}_\mathsf{bad}}

\newcommand{\Ccalcon}{\mathcal{C}_\mathsf{con}}
\newcommand{\Vcalcon}{\mathcal{V}_\mathsf{con}}
\newcommand{\barCcalint}{\overline{\mathcal{C}}_\mathsf{int}}
\newcommand{\Ccalint}{\mathcal{C}_\mathsf{int}}

\newcommand{\Vcalalive}{\mathcal{V}_\mathsf{alive}}
\newcommand{\barCcalQmarkint}{\overline{\mathcal{C}}_{\Qmark\text{-}\mathsf{int}}}

\newcommand{\CcalQmarkint}{\mathcal{C}_{\Qmark\text{-}\mathsf{int}}}

\newcommand{\HD}{\mathsf{HD}}
\newcommand{\NextVar}{\mathsf{NextVar}}
\newcommand{\Tcalsim}{\Tcal_\mathsf{sim}}
\newcommand{\Ncalrec}{\mathcal{N}_\mathsf{rec}}
\newcommand{\Ncaltrunc}{\mathcal{N}_\mathsf{trunc}}
\newcommand{\Ncalsamptrunc}{\mathcal{N}_{\mathsf{samp}\text{-}\mathsf{trunc}}}
\newcommand{\Ncalrectrunc}{\mathcal{N}_{\mathsf{rec}\text{-}\mathsf{trunc}}}
\newcommand{\drec}{{d_\mathsf{rec}}}

\newcommand{\Naturale}{\mathbf{e}}

\newcommand{\Acal}{\mathcal{A}}

\newcommand{\Ccal}{\mathcal{C}}
\newcommand{\Dcal}{\mathcal{D}}
\newcommand{\Ecal}{\mathcal{E}}

\newcommand{\Pcal}{\mathcal{P}}
\newcommand{\Qcal}{\mathcal{Q}}
\newcommand{\Rcal}{\mathcal{R}}
\newcommand{\Scal}{\mathcal{S}}
\newcommand{\Tcal}{\mathcal{T}}

\newcommand{\Vcal}{\mathcal{V}}
\newcommand{\Wcal}{\mathcal{W}}
\newcommand{\Xcal}{\mathcal{X}}

\renewcommand{\tilde}{\widetilde}
\renewcommand{\bar}{\overline}
\renewcommand{\hat}{\widehat}

\title{Improved Bounds for Sampling Solutions of Random CNF Formulas}
\author{
Kun He\thanks{
The Key Lab of Data Engineering and Knowledge Engineering, MOE, Renmin University of China. 
Email: \texttt{hekun.threebody@foxmail.com}.
The research of K.\ He is supported by the NSFC grant No.\ 62002231.}
\and
Kewen Wu\thanks{University of California at Berkeley. Email: \texttt{shlw\_kevin@hotmail.com}}
\and
Kuan Yang\thanks{John Hopcroft Center for Computer Science, Shanghai Jiao Tong University. Email: \texttt{kuan.yang@sjtu.edu.cn}. 
The research of K.\ Yang is supported by the NSFC grant No.\ 62102253.}
}
\date{}

\begin{document}
\maketitle

\begin{abstract}
Let $\Phi$ be a random $k$-CNF formula on $n$ variables and $m$ clauses, where each clause is a disjunction of $k$ literals chosen independently and uniformly.
Our goal is to sample an approximately uniform solution of $\Phi$ (or equivalently, approximate the partition function of $\Phi$).

Let $\alpha=m/n$ be the density.
The previous best algorithm runs in time $n^{\mathsf{poly}(k,\alpha)}$ for any $\alpha\lesssim2^{k/300}$ [Galanis, Goldberg, Guo, and Yang, SIAM J. Comput.'21].
Our result significantly improves both bounds by providing an almost-linear time sampler for any $\alpha\lesssim2^{k/3}$.

The density $\alpha$ captures the \emph{average degree} in the random formula.
In the worst-case model with bounded \emph{maximum degree}, current best efficient sampler works up to degree bound $2^{k/5}$ [He, Wang, and Yin, FOCS'22 and SODA'23], which is, for the first time, superseded by its average-case counterpart due to our $2^{k/3}$ bound.
Our result is the first progress towards establishing the intuition that the solvability of the average-case model (random $k$-CNF formula with bounded average degree) is better than the worst-case model (standard $k$-CNF formula with bounded maximal degree) in terms of sampling solutions.
\end{abstract}


\section{Introduction}\label{sec:introduction}
A random $k$-CNF formula $\Phi=\Phi(k,n,m)$ is a formula on $n$ Boolean variables and $m$ clauses, where each clause is a disjunction of $k$ literals sampled from all $2n$ possible literals uniformly and independently.
Let $\alpha = m/n$ be the \emph{density} of the formula, which captures the \emph{average degree} for variables in $\Phi$.

The random $k$-CNF model exhibits a fascinating phenomenon of a sharp phase transition in satisfiability. 
Based on numerical simulations and non-rigorous arguments in physics~\cite{MPZ2002,MMZ2005}, it was conjectured that there exists a critical value $\alpha_{\star} = \alpha_{\star}(k)$ such that for all $\eps > 0$, it holds that
$$
\lim_{n\to\infty}\Pr\sbra{\Phi(k, n, m) \text{ is satisfiable}} = \begin{cases}
    1 & \text{ if $\alpha = \alpha_{\star} - \eps$},\\
    0 & \text{ if $\alpha = \alpha_{\star} + \eps$}.
\end{cases}
$$

It has been a well-known challenge to prove the conjecture and determine the critical value $\alpha_{\star}$. 
Following a line of work, \cite{KKKS1998,FR1999, AM2002,AP2003,CP2016}, this conjecture is proved by Ding, Sly, and Sun \cite{DSS2022} for sufficiently large $k$, where the exact value of $\alpha_{\star}$ is also established. 
Roughly speaking, we have $\alpha_{\star}=2^k\ln(2)-(1+\ln(2))/2+o_k(1)$ as $k\to+\infty$.

However, the method for showing the sharp lower bound of $\alpha_{\star}$ is not constructive, and thus does not provide efficient algorithms to find solutions.
The current best polynomial-time algorithm for searching solutions is the \textsf{FIX} algorithm given by Coja-Oghlan~\cite{Coja2010}, which succeeds with high probability if $\alpha\lesssim2^k\ln(k)/k$.\footnote{We use $\lesssim$ to informally and flexibly hide low-order terms to simplify expressions.}
This is conjectured to be the search threshold, i.e., finding a solution is conjectured computationally hard if $\alpha$ goes beyond $2^k\ln(k)/k$. 
It is known~\cite{AC2008} that the solution space of random formulas has long-range correlations beyond density bound ${2^k} \ln(k)/k$, which suggests that local search algorithms are unlikely to succeed in polynomial time. 
Later, some particular algorithms have been ruled out (See e.g., \cite{Het2016,CHH2017}). 
To date, the strongest negative result is given by Bresler and Huang~\cite{BH2021}, which proves that a class of \emph{low-degree polynomial algorithms} (including \textsf{FIX}) cannot efficiently solve random $k$-CNF formulas beyond density $4.911\cdot2^k \ln(k)/k$. 
This gives a strong evidence that $2^k\ln(k)/k$ is the correct algorithmic phase transition.

Beyond decision and search, it is a natural next step to sample a satisfying assignment uniformly from the solution space. 
This is closely related to approximating the number of solutions of the formula $\Phi$, denoted by $Z(\Phi)$, and falls under the algorithmic study of partition functions in statistical physics.
Montanari and Shah~\cite{MS2007} presented the first efficient algorithm to approximately compute the partition function $\log(Z_\beta(\Phi))/n$ for a weighted model of random $k$-CNF, where the weight of an assignment $\sigma$ is $\Naturale^{-\beta\cdot H(\sigma)}$ and $H(\sigma)$ is the number of unsatisfied clauses under $\sigma$. 
The number of satisfying assignments $Z(\Phi)$ then corresponds to $\lim_{\beta\to+\infty}Z_\beta(\Phi)$. However, their algorithm is based on the correlation decay method and only works within the \emph{uniqueness regime} of the Gibbs distribution of the random $k$-CNF model. This uniqueness regime is $\alpha\lesssim2\ln(k)/k$, exponentially lower than the satisfiability and search thresholds.
The first significant improvement was given by Galanis, Goldberg, Guo, and Yang \cite{DBLP:journals/siamcomp/GalanisGGY21}, who designed a fully polynomial-time approximation scheme for $Z(\Phi)$ with runtime $n^{\poly(k,\alpha)}$ assuming $\alpha\lesssim2^{k/300}$.

\paragraph*{Comparison with the Worst-Case Model.}
Since the density $\alpha$ is defined to be the ratio between the number of clauses and variables, it is easy to see that $k\cdot\alpha$ equals the \emph{average degree} of variables in the random $k$-CNF model.
Here we compare this average-case model (i.e., random $k$-CNF formulas with average degree $k\alpha$) with its worst-case counterpart (i.e., standard $k$-CNF formulas with maximum degree $d$).
Since randomness kills structures in the worst-case examples, intuitively the average-case model should have advantages over the worst-case model in terms of solvability under the same (average/maximum) degree assumption.
This brings out the following intriguing question:
\begin{center}
    \it Is it true that the average-case model is easier to solve than the worst-case model?
\end{center}
This question has been answered affirmatively for satisfiability and search:
\begin{itemize}
\item The satisfiability threshold of the average-case model is $k\alpha\approx k2^k\ln(2)$ \cite{DSS2022}, whereas it shrinks to $d\approx2^{k+1}/(\Naturale k)$ in the worst-case model by the lopsided Lov\'asz local lemma \cite{erdos1991lopsided,gebauer2016local}.
\item The search threshold for the average-case model is (at least) $k\alpha\approx2^k\ln(k)$ \cite{Coja2010}, which is still beyond the above $d\approx2^{k+1}/(\Naturale k)$ satisfiability threshold of the worst-case model.\footnote{In fact, the search threshold for the worst-case model here is indeed this bound \cite{gebauer2016local}.}
\end{itemize}
Given these, it is reasonable to speculate that the task of sampling solutions is also easier in the average-case model than the worst-case model, which, however, is less clear before our work.

Moitra \cite{Moi19} designed the the first sampling algorithm for the worst-case model, which works whenever $d\lesssim2^{k/60}$ and runs in time $n^{\poly(k,d)}$. 
Since then, both the degree bound and the runtime have been significantly improved.
After \cite{FGYZ20,feng2021sampling,Vishesh21sampling,HSW21}, the state-of-the-art bound is $d\lesssim2^{k/5}$ and $n\cdot\poly(k,d,\log(n))$ runtime by He, Wang, Yin \cite{he2022sampling,HWY22Deterministic}.
In terms of the computational hardness, Bez{\'{a}}kov{\'{a}}, Galanis, Goldberg, Guo, and {\v{S}}tefankovi{\v{c}} \cite{BGGGS19} showed that the sampling task becomes intractable if $d$ can go beyond $2^{k/2}$ assuming $\textsf{NP}\neq\textsf{RP}$.

In contrast, for the average-case model, there is no improvement after \cite{DBLP:journals/siamcomp/GalanisGGY21}.
The best bound is still $\alpha\lesssim2^{k/300}$ and $n^{\poly(k,\alpha)}$ runtime, which falls short of the solvability intuition.
Indeed, \cite{DBLP:journals/siamcomp/GalanisGGY21} builds upon the techniques of \cite{Moi19}, and thus has the similar runtime bound; on the other hand, the existence of high-degree variables in the random setting poses significant challenges in carrying over the previous analysis, which results in the even worse $2^{k/300}\ll2^{k/60}$ degree bound.
Moreover, the ideas leading to subsequent improvements \cite{FGYZ20,feng2021sampling,Vishesh21sampling,HSW21} over \cite{Moi19} do not seem to extend here. We will elaborate in more detail in \Cref{sec:proof_overview}.

Therefore, it remains an intriguing open problem whether the ``average-case easier than worst-case'' conjecture is also true for sampling thresholds.
Our result is the first evidence towards this direction: Our algorithms works up to $\alpha\lesssim2^{k/3}$ and runs in time $n^{1+o_k(1)}\cdot\poly(k,\alpha,\log(n))$.
This not only drastically improves both degree and runtime bounds in \cite{DBLP:journals/siamcomp/GalanisGGY21}, but outperforms the current best $2^{k/5}$ degree bound in the worst-case model \cite{he2022sampling,HWY22Deterministic} as predicted by the intuition above.\footnote{We do \emph{not} claim that our result validates the intuition. On the one hand, it is very possible that our bounds can be further improved. On the other hand, the bounds for the worst-case model may also be far from the truth considering the hardness results \cite{BGGGS19}.}

\paragraph*{Independent Works.}
Independent of our work, there are two recent works on sampling solutions of random $k$-CNF formulas \cite{DBLP:journals/corr/abs-2206-15308,chen2023algorithms} improving \cite{DBLP:journals/siamcomp/GalanisGGY21}.
The algorithm from \cite{DBLP:journals/corr/abs-2206-15308} works when $\alpha\lesssim2^{0.039k}$ and runs in almost-linear time; and the algorithm from \cite{chen2023algorithms} requires $\alpha\lesssim2^{0.0134k}$ and runs in $n^{\poly(k,\alpha)}$ time.
In terms of results, our density bound $\alpha\lesssim2^{k/3}$ and almost-linear runtime subsume both of them.

Both \cite{DBLP:journals/corr/abs-2206-15308} and \cite{chen2023algorithms} use Markov-chain-based algorithms in line with \cite{FGYZ20,feng2021sampling,Vishesh21sampling,HSW21}, while our algorithm follows the recursive sampling approach recently developed in \cite{anand2022perfect,he2022sampling,HWY22Deterministic}.
Therefore both the analysis and bounds of the papers are very different.

\subsection{Our Results and Future Directions}\label{sec:our_results}

Our main result is a Monte Carlo algorithm with almost-linear runtime for sampling solutions of a random CNF formula with large density.

\Cref{thm:informal_alpha} characterizes the extreme case where $\alpha$ is close to $2^{k/3}$ up to $\poly(k)$ factors.

\begin{theorem}[Informal]\label{thm:informal_alpha}
Assume $\alpha\approx2^{k/3}$ and $k,n$ sufficiently large. Then with high probability, we can sample an approximate uniform solution of $\Phi$ in time $n^{1+1/k}\cdot\poly(\alpha,\log(n))$.
\end{theorem}

The runtime of our algorithm improves as the gap between the density and $2^{k/3}$ becomes larger.
\Cref{thm:informal_time} obtains extremely efficient runtime with a slight exponential sacrifice on the density.

\begin{theorem}[Informal]\label{thm:informal_time}
Assume $\alpha\approx2^{0.33\cdot k}$ and $k,n$ sufficiently large. Then with high probability, we can sample an approximate uniform solution of $\Phi$ in time $n^{1+2^{-0.001\cdot k}}\cdot\poly(\alpha,\log(n))$.
\end{theorem}

Both \Cref{thm:informal_alpha} and \Cref{thm:informal_time} are the informal and special cases of the following \Cref{thm:main_algorithm},\footnote{In the statement of \Cref{thm:main_algorithm}, we only hide absolute constants in $\Omega(\cdot),o(\cdot)$ and fixed polynomial in $\poly(\cdot)$. These do not depend on any parameter we introduce.} which achieves a smooth interpolation between the slack $\xi$ on the density and the efficiency on the runtime.
It also makes the ``approximate uniform'' precise by an explicit total variation distance measure $\eps$.

\begin{theorem}\label{thm:main_algorithm}
There exists a Monte Carlo algorithm $\Acal=\Acal(\eps,k,\alpha,n,\Phi)$ for $\eps\in(0,1)$, $k\ge2^{20}$, and $n\ge2^{\Omega(k)}$ such that the following holds: 
If 
$$
\alpha\le\frac{2^{k/3}}{k^{50}}\cdot\xi
\quad
\text{where }2^{-k/8}\le\xi\le1,
$$
then $\Acal$ runs in time 
$$
(n/\eps)^{1+\xi/k}/\xi\cdot\poly(k,\alpha,\log(n/\eps)).
$$
Moreover, let $\mu'$ be the output distribution of $\Acal$ and let $\mu$ be a uniform solution of $\Phi$.
Then
$$
\Pr_\Phi\sbra{\Phi\text{ is not satisfiable}~\lor~\TVdist\pbra{\mu,\mu'}\le\eps}\ge1-o(1/n),
$$
where $\TVdist(\cdot,\cdot)$ is the total variation distance and $\mu$ is a uniform random solution of $\Phi$.
\end{theorem}

The $o(1/n)$ factor in \Cref{thm:main_algorithm} can be improved to any $n^{-\Omega(1)}$ by slightly changing constants in our analysis for the structural properties in \Cref{sec:properties_of_random_cnf_formulas}.
Similarly, the denominator $k^{50}$ in the density bound or the $1/k$ on the exponent of the runtime bound can be polynomially improved by more refined calculation.

Our sampling algorithm can be turned into an efficient approximate counting algorithm.
This can be achieved by executing the algorithm multiple times to get approximations for marginal probabilities of variables in partial assignments, then applying well-known reductions between marginals and total number of solutions.
We refer interested readers to \cite[Section 9]{DBLP:journals/siamcomp/GalanisGGY21} for detail.

Curiously, our result holds in a stronger sense that we allow adversaries to change the signs of the literals in the clauses, i.e., an adversary can add or remove negations arbitrarily. 
Indeed, we identify the good formula purely based on the structural properties of the underlying hypergraphs on variables, regardless of the negations.
This feature may be of independent interests.

\paragraph*{Future Directions.}

We highlight some interesting future directions regarding sampling solutions of random formulas:
\begin{itemize}
\item \textbf{Better Density Bounds.}
Our sampling algorithm is efficient for density up to $2^{k/3}$.
In contrast, the satisfiability and search thresholds are roughly $2^k$.
We believe that there exist better sampling algorithms that goes beyond $2^{k/3}$.
A milestone will be to get around $2^{k/2}$, which, if true, would match the hardness in the worst case setting \cite{BGGGS19}.
In fact, it is speculative that the sampling threshold is also near $2^k$, since the random $k$-CNF formula is locally sparse and, in the bounded-degree model, solutions of $k$-CNF formulas on linear hypergraphs admits efficient sampling for variable degree up to $2^k$ \cite{DBLP:journals/rsa/HermonSZ19,QWZ22}.
\item \textbf{Random Monotone Formulas.}
As mentioned above, our algorithm works even when the signs of the literals are chosen adversarially. This is partially due to our use of Lov\'asz local lemma which is oblivious to the signs.
It is possible that better algorithms arises from better understanding on the patterns of negations.
Towards this direction, we ask if better density bounds are obtainable for random monotone $k$-CNF formulas, which should be the easiest due to its trivial satisfiability.
For its bounded-degree counterpart, it is indeed known that the sampling threshold is $2^{k/2}$ \cite{BGGGS19,DBLP:journals/rsa/HermonSZ19,QWZ22}, much larger than the $2^{k/3}$ bound obtained here.
\item \textbf{Better Error Bounds.}
The $o(1/n)$ error bound in \Cref{thm:main_algorithm} can be easily improved to any $n^{-\Omega(1)}$. 
It is even imaginable to obtain a bound scales with $k$, say, $n^{-\sqrt k}$.
However, it is not clear how to go beyond $n^{-\Omega(k)}$. This is because our analysis crucially replies on Lov\'asz local lemma which in turn needs an $\Omega(k)$ lower bound on the clause width, i.e., the number of distinct literals in a clause. Whereas, once the error bound becomes smaller than $n^{-\Omega(k)}$, we may get many clauses of very small width.
\item \textbf{Small Input Regimes.}
Our result holds for large inputs that has both large clause width $k\ge2^{20}$ and large amount of variables $n\ge2^{\Omega(k)}$. It is an intriguing question whether we can weaken these assumptions.
The former large-$k$ assumption appears commonly in the study of satisfiability and search thresholds (See e.g., \cite{DSS2022,Coja2010}), and there are non-rigorous arguments and experimental evidence \cite{ardelius2008exhaustive} showing the difficulty and distinction for small $k$'s, which may carry over to the sampling task as well.
The second large-$n$ assumption comes from our pursuit for \emph{highly efficient} algorithms. Indeed, if we are satisfied with arbitrary overhead on $k$ in the runtime, say, $2^{2^{O(k)}}\cdot n^{1+\xi/k}$, it can be removed as we can trivially go over all possible $2^n=2^{2^{O(k)}}$ assignments when $n\le2^{O(k)}$.
But it is not clear how to do it if we want $n^{1+o(1)}\cdot\poly(k,\alpha,\log(n))$ or even $\poly(n,k,\alpha)$ runtime.
\end{itemize}

\subsection{Proof Overview}\label{sec:proof_overview}

Our algorithm is inspired by a recursive sampling scheme recently developed in \cite{anand2022perfect,he2022sampling,HWY22Deterministic}.
We first identify the technical difficulties in applying the techniques from \cite{FGYZ20,feng2021sampling,Vishesh21sampling,HSW21}, which have proved successful in the worst-case model.
Then we show how \cite{DBLP:journals/siamcomp/GalanisGGY21} circumvent some of the issues using techniques from \cite{Moi19} and what makes their bound much worse than \cite{Moi19}.
Finally we discuss our approach and technique novelties leading to near $2^{k/3}$ density and almost-linear runtime.

\paragraph*{Bottlenecks in Previous Algorithms.}
The algorithms in \cite{FGYZ20,feng2021sampling,Vishesh21sampling,HSW21} are based on Markov chains.
Recall that in the worst-case model, the variables have a worst-case degree bound $d$.
Their algorithm can be summarized as follows: (1) Classify the variables as marked and unmarked ones. (2) Construct a Markov chain on the marked variables where each time we update a (random) marked variable based on its marginal distribution conditioned on the partial assignment at that point. (3) When the Markov chain on the marked variables mixes, we sample the unmarked variables to obtain a solution.

The core of their analysis is the \emph{local uniformity} for the marked variables:
Once we guarantee that every clause has enough unmarked variables, by Lov\'asz local lemma \cite{erdHos1973problems,DBLP:journals/jacm/HaeuplerSS11}, the marginal distribution of a marked variable is close to an unbiased coin, assuming the unmarked variables are untouched and regardless of the value of the other marked variables.
We also need to guarantee that every clause has enough marked variables to ensure that the update in Step (2) and the sampling in Step (3) are efficient.
In addition, this marking needs to be provided in advance of the Markov chain, which makes the mark-vs-unmark trade-off static and thus restrict their final degree bounds.

The most challenging part is to establish bounds for the mixing time for Step (3).
To this end, \cite{FGYZ20,feng2021sampling} rely on path coupling arguments, i.e., showing large contraction for one-step update of neighboring Markov chain states; and \cite{Vishesh21sampling,HSW21} uses information percolation arguments, i.e., bounding the probability of long-range uncoupling in the time series.
Both these arguments face severe obstacles in the average-case model due to the existence of high-degree variables which appear with high probability and do not have the local uniformity property.
As a consequence, the contraction in path coupling arguments could be vanishing, and the long-range uncoupling could actually appear.

Aside from the proof strategies, there is some evidence that this kind of one-step-update Markov chain relying on the local uniformity property may be slow mixing.
Consider a star graph of degree $D$, i.e., a node $v$ connecting to nodes $u_1,u_2,\ldots,u_D$.
The node $v$ models a high-degree variable or a component consisting of mostly high-degree variables, and nodes $u_1,u_2,\ldots,u_D$ are the surrounding low-degree neighbors.
Then it is likely that this structure appears in the underlying hypergraph of a random $k$-CNF formula for $D=\omega_{k,\alpha}(1)$ or even $D=\poly(k,\alpha)\cdot\log(n)$.
Let $\sigma$ and $\sigma'$ be two distinct assignments that do not touch $v$.
Now the Markov chain will ignore $v$ and only update $u_i$'s due to the local uniformity constraint.
For each $u_i$, even if its current value is the same in $\sigma$ and $\sigma'$, the one-step-update may make it differ.
The probability of this uncoupling is a small constant (independent of $n$) provided by the local uniformity, which means the estimate of the mixing time is $O(1)^D\gtrsim\poly(n)$.

Note that the recent independent works \cite{DBLP:journals/corr/abs-2206-15308,chen2023algorithms} bypass this issue by making the Markov chain update more (actually, constant fraction of) variables a time.
To argue the mixing time, they leverage recently developed \emph{spectral independence} techniques.
We refer interested readers to their paper for detail.
Unfortunately, their bounds still suffer from the static mark-vs-unmark trade-off and are thus much weaker than our result.

\paragraph*{How \cite{DBLP:journals/siamcomp/GalanisGGY21} Circumvents the Barrier.}
The algorithm in \cite{DBLP:journals/siamcomp/GalanisGGY21} does not involve Markov chains, and it samples an assignment by fixing a variable once at a time according to its (approximate) marginal distribution conditioned on the previous assignment.

To obtain the marginal distribution of a variable $v$, they adapt the linear programming framework from \cite{Moi19}.
Intuitively, starting from $v$, they gradually expand the possible values of its neighboring variables in a tree fashion.
Using this tree, they formulate a system of linear inequalities regarding the marginal probabilities provided by the local uniformity property,\footnote{In fact, the linear inequalities are about the ratio of the marginal probabilities. But since this is not important for us, we do not expand here.} where the marginals of the leaf nodes can be directly computed.
Then it is shown that any feasible solution to the linear system is a good approximation of the actually marginals, and in addition, it suffices to expand the tree up to logarithmic depth.
Therefore, a good approximate of the marginal of $v$ can be obtained by solving the linear programming.

This approach can be carried out in the average-case model. In particular, the above star graph example is no longer an issue if the formulated linear system includes all the $2^D$ partial assignments on $u_1,\ldots,u_D$.
Since $D$ is also upper bounded by $\poly(k,\alpha)\cdot\log(n)$ with high probability, we just expand the tree to this depth.
This explains their runtime being $n^{\poly(k,\alpha)}$: They need to solve a linear system of size $2^{\poly(k,\alpha)\cdot\log(n)}$.

For the density bound, the analysis in \cite{Moi19} already suffers from the loss in the static marking scheme required for local uniformity and to control the error of the linear system.
In the average-case model, \cite{DBLP:journals/siamcomp/GalanisGGY21} needs to first separate the high-degree variables, and then impose a stronger local uniformity assumption on the rest to make sure the error analysis goes through.
As a consequence, the bound in \cite{DBLP:journals/siamcomp/GalanisGGY21} is even worse than the one in \cite{Moi19}.

\paragraph*{How We Improve \cite{DBLP:journals/siamcomp/GalanisGGY21}.}
Our sampling algorithm follows the outline in \cite{Moi19,DBLP:journals/siamcomp/GalanisGGY21} by gradually fixing the variables towards a full assignment.
However, we replace the linear programming framework with the recursive sampling framework recently developed in \cite{anand2022perfect,he2022sampling,HWY22Deterministic}, which can be seen as a dynamic marking scheme as opposed to the static one above.
The benefit is two-fold: The runtime is significantly improved since we no longer need to solve giant linear systems, and the density bound is much better since the recursive sampling approach allows us to weaken the local uniformity assumption.

The first step of our algorithm is to start with high-degree variables and include all the bad variables which are influenced by them and do not possess local uniformity properties.
This part is similar to \cite{DBLP:journals/siamcomp/GalanisGGY21,DBLP:journals/siamcomp/Coja-OghlanF14} but we tighten their analysis in the study of structure properties of the random formula.
In particular, each remaining clause, after removing these bad variables, still has width $(1-o(1))\cdot k$.

To give a quantitative sense on the local uniformity property, we introduce $\theta\in(0,1)$ as the parameter for maximum possible ``marked'' variables in a clause. Note that our algorithm does not compute a static marking, and thus $\theta k$ is only used to upper bound the number of fixed variables in a clause at any point (or equivalently, $(1-\theta)k$ lower bounds the number of untouched variables in a clause).
Then by Lov\'asz local lemma \cite{erdHos1973problems,DBLP:journals/jacm/HaeuplerSS11}, the local uniformity parameter is 
\begin{equation}\label{eq:overview_1}
\delta\approx\alpha\cdot2^{-(1-\theta)k},
\end{equation}
which means the correct marginal distribution $\mu_v$ conditioned on the previous assignment for any remaining good variable $v$ is $\delta$-close to an unbiased coin.

Now we sample $\mu_v$ sequentially for good variables $v$ as \cite{DBLP:journals/siamcomp/GalanisGGY21}.
By local uniformity, we can already fix its value $\sigma(v)$ to $0$/$1$ with probability $(1-\delta)/2$ each, and set $\sigma(v)=\Qmark$ for the remaining $\delta$ uncertainty.
We denote this distribution as $\tau$.
Since ultimately we need to complete the $\Qmark$ to $0$/$1$ to obtain a sample from $\mu_v$, we will need to sample from $\tau_v\propto\mu_v-\tau$.
This part is similar to \cite{he2022sampling}.
With the \emph{Bernoulli factory} technique \cite{nacu2005fast,huber2016nearly,dughmi2021bernoulli}, samples from $\tau_v$ can be obtained efficiently provided samples from $\mu_v$.
This alone is merely a self-referencing: Sampling from $\mu_v$ circles back to samples from $\mu_v$.
But the trick here is to postpone sampling from $\tau_v$ and perform more sampling from $\tau$.

Let $v_1$ be a different variable with local uniformity property conditioned on $v=\Qmark$.
We can tentatively sample its value $\sigma(v_1)\sim\tau$, and, if $\sigma(v_1)=\Qmark$, update it by $\sigma(v_1)\sim\tau_{v_1|v=\Qmark}$.
Then we turn to the next variable $v_2$, sample $\sigma(v_2)\sim\tau$, and update $\sigma(v_2)\sim\tau_{v_2|v=\Qmark,v_1=\sigma(v_1)}$ if necessary.
Iteratively doing so gives us $\sigma(v_1),\sigma(v_2),\ldots,\sigma(v_t)$.
Now if we update $\sigma(v)\sim\tau_{v|v_1=\sigma(v_1),v_2=\sigma(v_2),\ldots,v_t=\sigma(v_t)}$, it follows the correct distribution $\tau_v$ in general by the law of conditional probability.
The hope here is that, after fixing $\sigma(v_1),\ldots,\sigma(v_t)$, the CNF formula decomposes into components and $v$ belongs to a small one, which allows us to efficiently obtain samples from $\mu_{v|v_1=\sigma(v_1),v_2=\sigma(v_2),\ldots,v_t=\sigma(v_t)}$ using rejection sampling for Bernoulli factory.
We remark that this algorithm incurs many recursions as, for example, sampling $\sigma(v_1)\sim\tau_{v_1|v=\Qmark}$ will also be postponed and implemented by the same recursive sampling idea.

The correctness of the above marginal sampling algorithm is evident from the description and can be proved rigorously by induction.
The difficulty lies in the efficiency analysis.
Indeed, we face two issues regarding the runtime: (1) The recursion may dive too deep such that branches into too many possibilities, and (2) the final Bernoulli factory may still require exponential time.
To address them, we keep track of the component $\Ccalcon^\sigma$ containing variables and clauses we visited during the recursion and relate its size $|\Ccalcon^\sigma|$ to the depth of the recursion and the efficiency of the final Bernoulli factory.
By a similar analysis as \cite{he2022sampling}, we show that a deep recursion produces a large component $\Ccalcon^\sigma$.
Therefore, to address both (1) and (2), it suffices to truncate the program once $|\Ccalcon^\sigma|$ exceeds certain size.

Then the issue comes back to the correctness: Is the output of the algorithm close to a uniform solution?
Observe that the difference between the new algorithm and the original one only lies in the place where truncation happens.
Therefore, it suffices to bound the probability that a large component $\Ccalcon^\sigma$ appears in the original algorithm.
To this end, we will construct a succinct witness $\Wcal^\sigma$ that enjoys the following properties: (a) Each large $\Ccalcon^\sigma$ gives rise to a witness $\Wcal^\sigma$, (b) each fixed $\Wcal$ appears as a witness of some $\sigma$ with small probability during the algorithm, and (c) there are not many possible $\Wcal$.
The construction of $\Wcal^\sigma$ is the place where we significantly deviate from (and simplify) the previous analysis and leverage the structural properties of random formulas.

Our witness $\Wcal$ consists of two sets of clauses $\Ccalint$ and $\CcalQmarkint$.
\begin{enumerate}[label=(\roman*)]
\item\label{itm:overview_1} $\Ccalint$ contains some unsatisfied clauses.

This is helpful for Property (b).
When we execute the algorithm and are about to fix a variable appearing in some clause $C\in\Ccalint$, the variable cannot be fixed to the bit that satisfies $C$.
Thus intuitively, the probability that the algorithm proceeds in the direction consistent with $\Wcal$ halves in this step.
\item\label{itm:overview_2} $\CcalQmarkint$ contains some clauses containing $\Qmark$'s.

Then similar to the $\Ccalint$ case, this intuitively requires the algorithm to go into the direction that assigns $\Qmark$ from $\tau$ whenever we encounter a variable indicated as $\Qmark$ in $\CcalQmarkint$.
The proper transition probability in this step is governed by the local uniformity $\tau(\Qmark)=\delta$.
\item\label{itm:overview_3} $\CcalQmarkint$ connects $\Ccalint$ in the underlying hypergraph.

This is helpful for Property (c). Using structural properties of the random formula, it can be shown that 
\begin{equation}\label{eq:overview_4}
\text{\# possible $\Wcal$'s}\lesssim\poly(n)\cdot\alpha^{|\Wcal|}.
\end{equation}
\end{enumerate}

Assume $\Wcal$ is the witness for $\Ccalcon^\sigma$, i.e., $\Wcal=\Wcal^\sigma$ from Property (a).
\Cref{itm:overview_1} tells us to include more visited variables in $\sigma$, since each one of them represents a probability decay for Property (b).
Recall that $\theta\in(0,1)$ controls the fraction of variables we can visit for each clause during the algorithm.
Then we have a trivial bound: The number of visited variables in $\Ccalint$ is at most $\theta k\cdot|\Ccalint|$.
Perhaps surprisingly, by the locally sparse properties of the random formula, we can almost achieve this bound!
More precisely, we show that one can carefully select a subset of $\Ccalcon^\sigma$ to form $\Ccalint$ such that the number of visited variables in $\Ccalint$ is at least $(\theta-o(1))k\cdot|\Ccalint|$, which means the accumulated probability drop from \Cref{itm:overview_1} is roughly
\begin{equation}\label{eq:overview_2}
2^{-\theta k\cdot|\Ccalint|}.
\end{equation}
\Cref{itm:overview_2} also requires us to include more $\Qmark$'s in $\sigma$ for Property (b). For this, we investigate the connectivity $\Ccalint$ inside $\Ccalcon^\sigma$ and show that one can make it connected by only inserting clauses that contain $\Qmark$'s.
Thus, by including the minimum amount of such clauses in a spanning tree fashion, we can additionally guarantee that the number of visited $\Qmark$'s in $\CcalQmarkint$ is at least $|\CcalQmarkint|$.
This means the accumulated probability drop from \Cref{itm:overview_2} is roughly
\begin{equation}\label{eq:overview_3}
\delta^{|\CcalQmarkint|}.
\end{equation}
Combining \Cref{eq:overview_2,eq:overview_3}, we derive the Property (b) of $\Wcal$ as
$$
\Pr[\text{encounter this $\Wcal$}]\lesssim2^{-\theta k|\Ccalint|}\cdot\delta^{|\CcalQmarkint|}\le\max\cbra{2^{-\theta k|\Wcal|},\delta^{|\Wcal|}}.
$$
Now to offset the number of possible $\Wcal$'s in the union bound, by \Cref{eq:overview_4} and \Cref{eq:overview_1}, we need to ensure
$$
\max\cbra{2^{-\theta k},\delta}\approx\max\cbra{2^{-\theta k},\alpha\cdot2^{-(1-\theta)k}}\lesssim1/\alpha
\quad\text{and}\quad
|\Wcal|\gtrsim\log(n).
$$
The former gives $\theta\le1/3$ and $\alpha\lesssim2^{k/3}$ as foreshadowed.
The latter, through some additional arguments, implies that the truncation threshold should be set to $\poly(k,\alpha)\cdot\log(n)$, which also explains why the above star graph example is not an obstacle here.

There are some technical difficulties that we choose to omit here for simplicity.
For example, the locally sparse property only holds up to certain size, and we need additional pruning ideas to make sure our witness enjoys the property.
After pruning, our witness is doomed to have an upper bound on its size.
This means, through final union bound in the witness analysis, the distance between the algorithm's output and a uniform solution has an inevitable lower bound.
Therefore, to handle the case where we want an extremely small output difference, we need another algorithm.
Similarly, some of the structural properties we use require a lower bound on the density $\alpha$.
Thus we also need a different algorithm for small densities.
We fix these issues by analyzing the naive rejection sampling algorithm and carefully balancing parameters for different algorithms.

\paragraph*{Organization.}
We give formal definitions in \Cref{sec:preliminaries}. 
Useful structural properties of random CNF formulas are provided in \Cref{sec:properties_of_random_cnf_formulas} and their proofs are deferred to \Cref{app:missing_proofs_in_sec:properties_of_of_random_cnf_formulas}.
In \Cref{sec:the_constructsep_algorithm}, we present the pre-processing algorithm to construct variable and clause separators.
In \Cref{sec:the_rejectionsampling_subroutine}, we analyze the naive rejection sampling on random formulas which gives the algorithms for the atypical setting.
In \Cref{sec:the_sampling_algorithm}, we introduce our main algorithms for the typical setting, the most technical part of which is the truncation analysis and is carried out in \Cref{sec:truncation_analysis}.
Finally we put everything together and prove \Cref{thm:main_algorithm} in \Cref{sec:putting_everything_together}.
\section{Preliminaries}\label{sec:preliminaries}

We use $\Naturale\approx2.71828$ to denote the \emph{natural} base, and we will frequently use the inequality $\binom ab\le(\Naturale a/b)^b$ for all $a,b\ge0$ where $0^0$ is defined as $1$.
We use $\log(\cdot)$ and $\ln(\cdot)$ to denote the logarithm with base $2$ and $\Naturale$ respectively.
For positive integer $n$, we use $[n]$ to denote the set $\cbra{1,2,\ldots,n}$.

For a finite set $\Xcal$ and a distribution $\Dcal$ over $\Xcal$, we use $x\sim\Dcal$ to denote that $x$ is a random variable sampled from $\Xcal$ according to distribution $\Dcal$. We also use $x\sim\Xcal$ when $\Dcal$ is the uniform distribution.

\paragraph*{Asymptotics.}
We only use $O(\cdot),\Omega(\cdot),o(\cdot),\omega(\cdot)$ to hide absolute constants that does not depend on any parameters we introduce.
In addition, $\tilde O(\cdot)$ is only used to bound algorithms' runtime which hides polynomial factors in $k,\alpha,\log(n/\eps)$, i.e., $\tilde O(f)=\poly(k,\alpha,\log(n/\eps))\cdot f$ for some fixed $\poly$.

\paragraph*{(Random) CNF Formula.}
A CNF formula is a disjunction of clauses.
Each clause is a conjunction of literals, and a literal is either a Boolean variable or the negation of a Boolean variable.
Given a CNF formula $\Phi=(\Vcal,\Ccal)$ with variable set $\Vcal$ and clause set $\Ccal$, we define the following measure for $\Phi$:
\begin{itemize}
\item The \emph{width} is $k(\Phi)=\max_{C\in\Ccal}\abs{\vbl(C)}$, where $\vbl(C)$ denotes the variables that $C$ depends on.
\item The \emph{variable degree} is $d(\Phi)=\max_{v\in\Vcal}\abs{\cbra{C\in\Ccal\mid v\in\vbl(C)}}$.
\item The \emph{constraint degree} is $\Delta(\Phi)=\max_{C\in\Ccal}\abs{\cbra{C'\in\Ccal\mid\vbl(C)\cap\vbl(C')\neq\emptyset}}$.\footnote{Note that in our definition, $\Delta$ is one plus the maximum degree of the dependency graph of $\Phi$.}
\item The \emph{maximum violation probability} is $p(\Phi)=\max_{C\in\Ccal}\Pr\sbra{C(\sigma)=\False}=\max_{C\in\Ccal}2^{-|\vbl(C)|}$.
\end{itemize}
In addition, we use $\mu(\Phi)$ to denote the uniform distribution over the solutions of $\Phi$. Note that $\mu$ is well defined whenever $\Phi$ is satisfiable.
In the rest of the paper, we will simply use $k,d,\Delta,p,\mu$ when $\Phi$ is clear from the context.

We use $\Phi(k,n,m)$ to denote a random $k$-CNF formula on $n$ variables and $m$ clauses, where $\Vcal=\cbra{v_1,v_2,\ldots,v_n}$ is the variable set, $\Ccal=\cbra{C_1,C_2,\ldots,C_m}$, and each clause is an independent disjunction of $k$ literals chosen independently and uniformly from $\cbra{v_1,v_2,\ldots,v_n,\neg v_1,\neg v_2,\ldots,\neg v_n}$.
We will simply use $\Phi$ to denote $\Phi(k,n,m)$ when context is clear.

\paragraph*{Partial Assignments and Restrictions.}
Our algorithm will sample an assignment by gradually fixing coordinates. To this end, we will work with partial assignments and restrictions of the formula on partial assignments.
We use $\EQmark$ for \emph{unaccessed variables} and use $\Qmark$ for \emph{accessed but unassigned variables} and a partial assignment $\sigma$ lies in the space $\cbra{0,1,\Qmark,\EQmark}^\Vcal$.
We define 
$$
\Lambda(\sigma)=\cbra{v\in\Vcal\mid\sigma(v)\in\cbra{\Qmark,\EQmark}}
$$ 
to be the set of unassigned variables.
We then abuse the notation to say $C(\sigma)=\True$ if fixing $v$ to $\sigma(v)$ for all $v\notin\Lambda(\sigma)$ already satisfies $C$.

For a partial assignment $\sigma$, let $\Phi^\sigma=(\Vcal^\sigma,\Ccal^\sigma)$ be the CNF formula after we fix $v$ to be $\sigma(v)$ for each $v\notin\Lambda(\sigma)$.
Note that $\Vcal^\sigma=\Lambda(\sigma)$ and each clause in $\Ccal^\sigma$ depends only on variables in $\Lambda(\sigma)$.

We use $\mu^\sigma=\mu(\Phi^\sigma)$ to denote the uniform distribution over solutions of $\Phi^\sigma$. 
For each $v\in\Lambda(\sigma)$, we write $\mu_v^\sigma$ as the marginal distribution of $v$ under $\mu^\sigma$.
Then $\mu_v^\sigma(b)$ denotes the probability that $v$ is fixed to $b\in\bin$ under $\mu_v^\sigma$.
For multiple variables $S\subseteq\Lambda(\sigma)$, we use $\mu_S^\sigma$ to denote the marginal distribution of $S$ under $\mu^\sigma$.

\paragraph*{Incidence Graphs.}
Given a formula $\Phi=(\Vcal,\Ccal)$, we define two incidence graphs $G_\Phi$ and $H_\Phi$:
\begin{itemize}
\item The vertex set of $G_\Phi$ is $\Ccal$, and two clauses $C_1,C_2\in\Ccal$ are adjacent iff $\vbl(C_1)\cap\vbl(C_2)\neq\emptyset$. 
We say a set $S\subseteq\Ccal$ of clauses is connected if the induced sub-graph $G_\Phi[S]$ is connected.
\item The vertex set of $H_\Phi$ is $\Vcal$, and two variables $v,v'\in\Vcal$ are adjacent iff there exists some $C\in\Ccal$ with $v,v'\in\vbl(C)$. 
We say a set $T\subseteq\Vcal$ of variables is connected if the induced sub-graph $H_\Phi[T]$ is connected.
\end{itemize}

\paragraph*{Lov\'asz Local Lemma.}
The celebrated Lov\'asz local lemma \cite{erdHos1973problems} provides a sufficient condition for the existence of a solution of a constraint satisfaction problem.
Here we use a more general version for CNF formulas due to \cite{DBLP:journals/jacm/HaeuplerSS11}:
\begin{theorem}[{\cite[Theorem 2.1]{DBLP:journals/jacm/HaeuplerSS11}}]\label{thm:local_uniformity}
Let $\Phi=(\Vcal,\Ccal)$ be a CNF formula. If $\Naturale p\Delta\le1$, then $\Phi$ is satisfiable.
Moreover, for any event $B$ (not necessarily from $\Ccal$) we have
$$
\Pr_{\sigma\sim\mu}\sbra{B(\sigma)=\True}\le(1-\Naturale p)^{-\abs{\Gamma(B)}}\Pr_{\sigma\sim\bin^\Vcal}\sbra{B(\sigma)=\True},
$$
where $\Gamma(B)=\cbra{C\in\Ccal\mid\vbl(C)\cap\vbl(B)\neq\emptyset}$.
\end{theorem}
\section{Properties of Random CNF Formulas}\label{sec:properties_of_random_cnf_formulas}

For the rest of the paper, we will use $\Phi$ to denote a random $k$-CNF formula on $n$ variables $\Vcal=\cbra{v_1,\ldots,v_n}$ and $m$ clauses $\Ccal=\cbra{C_1,\ldots,C_m}$.
We reserve $\alpha=\alpha(\Phi)=m/n$ as the \emph{density} of $\Phi$.

For convenience and later reference, we list desirable properties of $\Phi$ here. In the next sections, we will assume $\Phi$ satisfies these structural properties, which happens with high probability, and prove the correctness and efficiency of our algorithm.

We first cite the following celebrated satisfiability result.
\begin{theorem}[{\cite[Theorem~$1$]{DSS2022}}]\label{thm:satisfiability}
For $k\ge\Omega(1)$, $\Phi$ has a sharp satisfiability threshold $\alpha_\star(k)$ such that for all $\eps > 0$, it holds that
$$
\lim_{n\to+\infty}\Pr\sbra{\Phi(k, n, m) \text{ is satisfiable}} = \begin{cases}
        1 & \text{ if $\alpha \le \alpha_{\star}(k) - \eps$},\\
        0 & \text{ if $\alpha \ge \alpha_{\star}(k) + \eps$}.
\end{cases}
$$
Roughly, $\alpha_{\star}(k)=2^k\ln(2)-(1+\ln(2))/2+o_k(1)<2^k$ as $k\to+\infty$.\footnote{The explicit value of $\alpha_\star(k)$ is characterized by a complicated proposition in~\cite{DSS2022}. We omit it here to simplify the statement. This asymptotic estimation is given by~\cite{KKKS1998} as an upper bound and by~\cite{CP2016} as a lower bound.}
\end{theorem}

By \Cref{thm:satisfiability}, it is reasonable to focus our attention to the case where $\alpha\le2^{O(k)}$. In particular, this justifies our assumption $\alpha\le2^k$ used below.
We remark that the proofs for the following properties are similar to the ones in \cite{DBLP:journals/siamcomp/GalanisGGY21,DBLP:journals/siamcomp/Coja-OghlanF14}. Therefore we defer them to \Cref{app:missing_proofs_in_sec:properties_of_of_random_cnf_formulas}.

The first property states that every clause in $\Phi$ has at most two duplicate variables.

\begin{proposition}\label{prop:width_k-2}
Assume $\alpha\le2^k$ and $n\ge2^{\Omega(k)}$.
Then with probability $1-o(1/n)$ over the random $\Phi$, $\abs{\vbl(C)}\ge k-2$ holds for every $C\in\Ccal$.
\end{proposition}

Intuitively, \Cref{prop:kvars_and_distinct_vars} and \Cref{prop:bkvars} show that typically the clauses in $\Phi$ are spread out in that they do not share many common variables.

\begin{proposition}\label{prop:kvars_and_distinct_vars}
Let $\eta=\eta(k)>0$ be a parameter. Assume $\alpha\le2^k$, $\frac k{\log(k)}\ge14\cdot\pbra{1+\frac1\eta}$, and $n\ge2^{\Omega(k)}$.
Then with probability $1-o(1/n)$ over the random $\Phi$, the following holds:
\begin{enumerate}
\item\label{itm:kvars_and_distinct_vars_1} 
For every $\Vcal'\subset\Vcal$ with $1\le|\Vcal'|\le n/2^{k/\log(k)}$, we have $\abs{\cbra{C\in\Ccal\mid\vbl(C)\subseteq \Vcal'}}\le(1+\eta)|\Vcal'|/k$.
\item\label{itm:kvars_and_distinct_vars_2} 
For every $\Ccal'\subset\Ccal$ with $1\le|\Ccal'|\le n/2^{2k/\log(k)}$, we have $\abs{\bigcup_{C\in\Ccal'}\vbl(C)}\ge k|\Ccal'|/(1+\eta)$.
\end{enumerate}
\end{proposition}

\begin{proposition}\label{prop:bkvars}
Let $\eta=\eta(k)\in(0,1)$ be a parameter.
Assume $\alpha\le2^k$, $\frac k{\log(k)}\ge14\cdot\pbra{1+\frac1\eta}$, and $n\ge2^{\Omega(k)}$.
Then with probability $1-o(1/n)$ over the random $\Phi$, the following holds:
For any $b\ge\eta$ and every $\Vcal'\subset\Vcal$ with $1\le|\Vcal'|\le n/2^{3k/\log(k)}$,
we have 
$$
|\Vcal'|\ge(b-\eta)k\cdot\abs{\cbra{C\in\Ccal\mid\abs{\vbl(C)\cap\Vcal'}\ge bk}}.
$$
\end{proposition}

Recall our definition of incidence graph $G_\Phi$ from \Cref{sec:preliminaries}, we can bound the number of induced connected sub-graphs in $G_\Phi$.

\begin{proposition}\label{prop:number_of_connected_sets}
With probability $1-o(1/n)$ over the random $\Phi$, the following holds:
For every $C\in\Ccal$ and $\ell\ge1$, there are at most $\alpha^2n^4(\Naturale k^2\alpha)^\ell$ many connected sets of clauses in $G_\Phi$ with size $\ell$ containing $C$.
\end{proposition}

In terms of incidence graph $H_\Phi$, we can bound the expansion of any connected set.

\begin{proposition}\label{prop:number_of_neighbors}
Assume $k\ge30$ and $\alpha\ge1/k^3$.
Then with probability $1-o(1/n)$ over the random $\Phi$, the following holds:
For any $\Vcal'\subseteq\Vcal$ connected in $H_\Phi$, we have
$$
\abs{\cbra{v\in\Vcal\mid v\in\Vcal'\text{ or }v\text{ is adjacent to }\Vcal'}}\le 3k^4\alpha\cdot\max\cbra{|\Vcal'|,\floorbra{k\log(n)}}.
$$
\end{proposition}

Given a set of clauses $\Ccal'\in\Ccal$ and a variable $v\in\Vcal$, we define the degree of $v$ in $\Ccal'$ as $\deg_{\Ccal'}(v)=\abs{\cbra{C\in\Ccal'\mid v\in\vbl(C)}}$.
Then $d(\Phi)=\max_{v\in\Vcal}\deg_\Ccal(v)$.
We first note a classical bound (See e.g., \cite[Theorem 1]{RS98}) on $d(\Phi)$.

\begin{proposition}\label{prop:maximum_degree}
With probability $1 - o(1/n)$ over the random $\Phi$, we have $d(\Phi) \leq4k\alpha+6\log(n)$.
\end{proposition}

We also need the following control over the number of high-degree variables.
\begin{proposition}\label{prop:high-degree}
Let $D=D(k,\alpha)$ be a parameter satisfying $D\ge8k(\alpha+1)$. Assume $k\ge2$, $\alpha\le2^k$, and $n\ge2^{\Omega(k)}$.
Then with probability $1-o(1/n)$ over the random $\Phi$, we have
$$
\abs{\cbra{v\in\Vcal\mid\deg_\Ccal(v)\ge D}}\le n/2^{4k}.
$$
\end{proposition}

We can also bound the fraction of high-degree variables in any connected set.
\begin{proposition}\label{prop:fraction_of_high-degrees}
Let $D=D(k,\alpha)$ be a parameter satisfying $6k^7(\alpha+1)\le D\le2^{2k}$.
Assume $k\ge2^{10}$, $1/k^3\le\alpha\le2^k$, and $n\ge2^{\Omega(k)}$.
Then with probability $1-o(1/n)$ over the random $\Phi$, the following holds:
Let $\Vcal'\subseteq\Vcal$ be connected in $H_\Phi$ and $|\Vcal'|\ge\log(n)$. Then 
$$
\abs{\cbra{v\in\Vcal'\mid\deg_\Ccal(v)\ge D}}\le|\Vcal'|/k^2.
$$
\end{proposition}

Finally, the following proposition characterizes peeling procedures: It shows that the process of introducing new variables by including more clauses should stop soon.

\begin{proposition}\label{prop:peeling}
Assume $k\ge12$, $\alpha\le2^k$, and $n\ge2^{\Omega(k)}$.
Then with probability $1-o(1/n)$ over the random $\Phi$, the following holds:
Fix an arbitrary $\Ccal'\subset\Ccal$ with $|\Ccal'|\le n/2^{4k}$. Let $C_{i_1},\ldots,C_{i_\ell}\in\Ccal\setminus\Ccal'$ be clauses with distinct indices. For each $s\in[\ell]$, define $\Vcal_s=\bigcup_{C\in\Ccal'}\vbl(C)\cup\bigcup_{j=1}^{s-1}\vbl(C_{i_j})$.
If $|\vbl(C_{i_s})\cap\Vcal_s|\ge6$ holds for all $s\in[\ell]$, then $\ell\le|\Ccal'|$.
\end{proposition}

\subsection{Good and Nice Instances}\label{sec:good_nice}

At this point, we can assume $\Phi$ satisfies certain structural properties which exist with high probability over the random $\Phi$. 

To be specific, we define the following \Cref{def:good_instances} and \Cref{def:nice_instances}: 
The former provides structural properties for $\Phi$ when $\alpha$ has an upper bound, and the latter guarantees more structural properties by further assuming $\alpha\ge1/k^3$.
For convenience, we include $\eta$ and $D$ to be consistent with \Cref{sec:properties_of_random_cnf_formulas}.

\begin{definition}[Good Instances]\label{def:good_instances}
We say $(\Phi,k,\alpha,n,\xi,\eta,D)$ is \emph{good} if:
\begin{itemize}
\item $k\ge2^{20}$, $n\ge2^{\Omega(k)}$, $2^{-k/8}\le\xi\le1$, and $\alpha\le\xi\cdot2^{k/3}/k^{50}$.
\item $\eta=15\log(k)/k$ and $D=k^8(\alpha+1)/\xi$.
\item $\Phi$ is satisfiable and has the properties in \Cref{prop:width_k-2}, \Cref{prop:kvars_and_distinct_vars}, \Cref{prop:bkvars}, \Cref{prop:number_of_connected_sets}, \Cref{prop:maximum_degree}, \Cref{prop:high-degree}, and \Cref{prop:peeling}.
\end{itemize}
\end{definition}

\begin{definition}[Nice Instances]\label{def:nice_instances}
We say $(\Phi,k,\alpha,n,\xi,\eta,D)$ is \emph{nice} if:
\begin{itemize}
\item $(\Phi,k,\alpha,n,\xi,\eta,D)$ is good.
\item $\alpha\ge1/k^3$ and $\Phi$ has properties in \Cref{prop:number_of_neighbors} and \Cref{prop:fraction_of_high-degrees} additionally.
\end{itemize}
\end{definition}

\begin{remark}\label{rmk:parameter_valid}
By the choice of $\eta$ and $k\ge2^{20}$, $\eta$ satisfies $\frac k{\log(k)}\ge14\cdot\pbra{1+\frac1\eta}$. 
Since $2^{-k/8}\le\xi\le1$, $k\ge2^{20}$, and $\alpha\le2^{k/3}\cdot\xi$, we always have $6k^7(\alpha+1)\le D\le2^{2k}$.
This means that $\eta$ and $D$ are consistent with the structural statements in \Cref{sec:properties_of_random_cnf_formulas}.
\end{remark}

By the bounds in \Cref{sec:properties_of_random_cnf_formulas} and \Cref{rmk:parameter_valid}, we can indeed focus on good/nice instances.
Though checking whether it is indeed a good/nice instance may actually need exponential time, we will not do it in our algorithm.
Instead, we will assume the input enjoys the property, then run algorithm anyways and terminate it upon the prescribed maximum runtime.
The correctness of our algorithm is only guaranteed when the input is actually good/nice.

\begin{corollary}\label{cor:good_nice_whp}
Assume $k,\alpha,n,\xi,\eta,D$ satisfy the relations in \Cref{def:good_instances} (resp., \Cref{def:nice_instances}).
Then with probability $1-o(1/n)$ over the random $\Phi$, either $\Phi$ is not satisfiable or $(\Phi,k,\alpha,n,\xi,\eta,D)$ is good (resp., nice).
\end{corollary}

We remark that though \Cref{thm:satisfiability} asserts that the satisfiability probability of $\Phi$ approaches $1$ as $n$ goes to infinity, it only holds for $k$ sufficiently large (potentially much larger than $2^{20}$ in our setting).
In addition, it does not control the convergence rate. Therefore we cannot simply say we have good/nice instances with probability $1-o(1/n)$.
\section{Separating High-Degree Variables}\label{sec:the_constructsep_algorithm}

Define $\HD(\Vcal')=\cbra{v\in\Vcal'\mid\deg_\Ccal(v)\ge D}$ to be the set of high-degree variables in $\Vcal'$. 
Similar to \cite{DBLP:journals/siamcomp/GalanisGGY21}, our algorithm will start with high-degree variables and propagates them to form a separator.
We use $\Vcalsep$ and $\Ccalsep$ to denote the variable separators and clause separators obtained from \ConstructSep{$\Vcal$} respectively.

\begin{algorithm2e}[ht]
\caption{The \texttt{ConstructSep} Algorithm}\label{alg:the_constructsep_algorithm}
\DontPrintSemicolon
\KwIn{Variables $\Vcal'\subseteq\Vcal$}
\KwOut{Variable separators $\Vcalsep'\subseteq\Vcal$ and clause separators $\Ccalsep'\subseteq\Ccal$}
\nl Initialize $\Vcalsep'\gets\HD(\Vcal')$ and $\Ccalsep'\gets\emptyset$\;
\lnl{ln:sep_1} \While{$\exists C\in\Ccal\setminus\Ccalsep'$ such that $\abs{\vbl(C)\cap\Vcalsep'}\ge2\eta k$}{
\lnl{ln:sep_2} Update $\Vcalsep'\gets\Vcalsep'\cup\vbl(C)$ and $\Ccalsep'\gets\Ccalsep'\cup\cbra{C}$
}
\nl \Return{$\Vcalsep'$ and $\Ccalsep'$}
\end{algorithm2e}

By dynamically monitoring and updating $\abs{\vbl(C)\cap\Vcalsep'}$ for each $C\in\Ccal$, \Cref{alg:the_constructsep_algorithm} can be done efficiently.

\begin{fact}\label{fct:construct_sep_efficiency}
The runtime of \ConstructSep{$\Vcal'$} is $\tilde O(n)$ for any $\Vcal'\subseteq\Vcal$.
\end{fact}

Here we list some useful properties regarding the variable separators and clause separators for future referencing.

\begin{fact}\label{fct:sep_increasing}
$\Vcalsep'\subseteq\Vcalsep$ and $\Ccalsep'\subseteq\Ccalsep$ hold for any $\Vcal'\subseteq\Vcal$.
\end{fact}

We first bound the number of variable separators in terms of the number of high-degree variables.

\begin{lemma}\label{lem:fraction_of_high-degree_starting_from_high-degree}
Assume $(\Phi,k,\alpha,n,\xi,\eta,D)$ is good.
Then $|\Vcalsep'|\le2|\Vcal'|/\eta$ holds for any $\Vcal'\subseteq\HD(\Vcal)$.
\end{lemma}
\begin{proof}
By \Cref{prop:high-degree}, $|\Vcal'|\le|\HD(\Vcal)|\le n/2^{4k}<n/2^{3k/\log(k)}$.
Let
$$
\Ccal'=\cbra{C\in\Ccal\mid|\vbl(C)\cap\Vcal'|\ge2\eta k}.
$$
Thus by \Cref{prop:bkvars} with $b=2\eta$, we have $|\Ccal'|\le|\Vcal'|/(\eta k)$.
Moreover, $|\Ccal'|\le|\Vcal'|\le n/2^{4k}$.

Observe that, starting from $\Ccal'$, each clause newly added to $\Ccalsep'$ intersects at least $2\eta k\ge6$ variables with existing clauses.
Then by \Cref{prop:peeling} with $C_{i_1},\ldots,C_{i_\ell}$ being $\Ccalsep'\setminus\Ccal'$, we have $|\Ccalsep'\setminus\Ccal'|\le|\Ccal'|$, which implies $|\Ccalsep'|\le2|\Ccal'|\le2|\Vcal'|/(\eta k)$.
Thus $|\Vcalsep'|\le k|\Ccalsep'|\le2|\Vcal'|/\eta$.
\end{proof}

\begin{lemma}[{\cite[Lemma 8.9]{DBLP:journals/siamcomp/GalanisGGY21}}]\label{lem:construct_sep_in_component}
Let $\Vcal'\subseteq\Vcalsep$ be an arbitrary maximal connected component in $H_\Phi[\Vcalsep]$.
Then $\Vcalsep'=\Vcal'$.
\end{lemma}

\begin{lemma}\label{lem:fraction_of_high-degree_in_vsep_component}
Assume $(\Phi,k,\alpha,n,\xi,\eta,D)$ is good.
Let $\Vcal'\subseteq\Vcalsep$ consist of maximal connected components in $H_\Phi[\Vcalsep]$. Then $|\Vcal'|\le2|\HD(\Vcal')|/\eta$.
\end{lemma}
\begin{proof}
Note that it suffices to prove the bound for every maximal connected component in $H_\Phi[\Vcalsep]$ and then add them up.
Therefore we assume without loss of generality $\Vcal'$ is connected in $H_\Phi[\Vcalsep]$.

Observe that $\Vcalsep'$ and $\Ccalsep'$ equal the output of \ConstructSep{$\HD(\Vcal')$}.
By \Cref{lem:construct_sep_in_component}, we have $\Vcal'=\Vcalsep'$ and the desired bound follows immediately from \Cref{lem:fraction_of_high-degree_starting_from_high-degree}.
\end{proof}

Now we bound the fraction of $\Vcalsep$ in any large connected component in $H_\Phi$.

\begin{lemma}\label{lem:fraction_of_vsep}
Assume $(\Phi,k,\alpha,n,\xi,\eta,D)$ is nice.
Let $\Vcal'\subseteq\Vcal$ be connected in $H_\Phi$ of size $|\Vcal'|\ge\log(n)$.
Then $|\Vcal'\cap\Vcalsep|\le|\Vcal'|/k$.
\end{lemma}
\begin{proof}
Let $\Vcal_1,\ldots,\Vcal_\ell\subseteq\Vcalsep$ be distinct maximal connected components in $H_\Phi[\Vcalsep]$ and they intersect $\Vcal'$.
Let $\tilde\Vcal=\Vcal'\cup\Vcal_1\cup\cdots\cup\Vcal_\ell$.
Then $\tilde\Vcal$ is connected in $H_\Phi$ and $\HD(\tilde\Vcal)=\HD(\Vcal')\cup\HD(\Vcal_1)\cup\cdots\cup\HD(\Vcal_\ell)$.

Now by \Cref{lem:fraction_of_high-degree_in_vsep_component}, we have $|\Vcal_i|\le2|\HD(\Vcal_i)|/\eta$.
By \Cref{prop:fraction_of_high-degrees}, we also have $|\HD(\tilde\Vcal)|\le|\tilde\Vcal|/k^2$.
Note that $\eta k\ge2$, we have
$$
|\tilde\Vcal\cap\Vcalsep|=\sum_{i=1}^\ell|\Vcal_i|\le\frac2\eta\sum_{i=1}^\ell|\HD(\Vcal_i)|
\le\frac{2|\HD(\tilde\Vcal)|}\eta\le\frac{2|\tilde\Vcal|}{\eta k^2}\le\frac{|\tilde\Vcal|}k.
$$
Since $\tilde\Vcal\setminus\Vcal'\subset\Vcalsep$, we have
\begin{equation*}
\frac{|\Vcal'\cap\Vcalsep|}{|\Vcal'|}\le\frac{|\Vcal'\cap\Vcalsep|+|\tilde\Vcal\setminus\Vcal'|}{|\Vcal'|+|\tilde\Vcal\setminus\Vcal'|}=\frac{|\tilde\Vcal\cap\Vcalsep|}{|\tilde\Vcal|}\le\frac1k.
\tag*{\qedhere}
\end{equation*}
\end{proof}

As a corollary, we obtain the following bound on the fraction of $\Ccalsep$ in any large connected component in $G_\Phi$.

\begin{corollary}\label{cor:fraction_of_csep}
Assume $(\Phi,k,\alpha,n,\xi,\eta,D)$ is nice.
Let $\Ccal'\subseteq\Ccal$ be connected in $G_\Phi$ of size $|\Ccal'|\ge\log(n)$. Then $|\Ccal'\cap\Ccalsep|\le(1+\eta)|\Ccal'|/k$.
\end{corollary}
\begin{proof}
Let $\Vcal'=\bigcup_{C\in\Ccal'}\vbl(C)$. 
Then $\Vcal'$ is connected in $H_\Phi$.

First we prove for the case $|\Ccal'|\le n/2^{2k/\log(k)}$.
By \Cref{itm:kvars_and_distinct_vars_2} of \Cref{prop:kvars_and_distinct_vars}, we have $|\Vcal'|\ge k|\Ccal'|/(1+\eta)\ge\log(n)$.
Then by \Cref{lem:fraction_of_vsep}, we have $|\Vcal'\cap\Vcalsep|\le|\Vcal'|/k$. 
Since $\Ccal'\cap\Ccalsep$ supports on $\Vcal'\cap\Vcalsep$, applying \Cref{itm:kvars_and_distinct_vars_1} of \Cref{prop:kvars_and_distinct_vars}, we have
$$
|\Ccal'\cap\Ccalsep|\le\frac{1+\eta}k\cdot|\Vcal'\cap\Vcalsep|\le\frac{1+\eta}{k^2}\cdot|\Vcal'|\le\frac{1+\eta}k\cdot|\Ccal'|,
$$
where $|\Vcal'|\le k|\Ccal'|\le n/2^{k/\log(k)}$ as required.

Now we turn to the case $|\Ccal'|\ge n/2^{2k/\log(k)}$.
By \Cref{lem:fraction_of_high-degree_in_vsep_component} and \Cref{prop:high-degree}, we have
$$
|\Vcalsep|\le\frac{2|\HD(\Vcalsep)|}\eta=\frac{2|\HD(\Vcal)|}\eta\le\frac{2\cdot n}{\eta\cdot2^{4k}}.
$$
Since $\frac k{\log(k)}\ge14\cdot\pbra{1+\frac1\eta}$, we have $|\Vcalsep|\le n/2^{3k/\log(k)}$.
Meanwhile, $\Ccalsep$ supports on $\Vcalsep$. Thus by \Cref{prop:bkvars} with $b=1$, we have
$$
|\Ccalsep|\le\frac{|\Vcalsep|}{(1-\eta)k}\le\frac{2\cdot n}{\eta(1-\eta)k\cdot2^{4k}}\le\frac n{k\cdot2^{2k/\log(k)}}\le\frac{|\Ccal'|}k.
$$
Thus $|\Ccal'\cap\Ccalsep|\le|\Ccalsep|\le|\Ccal'|/k$.
\end{proof}
\section{The Naive Rejection Sampling Algorithm}\label{sec:the_rejectionsampling_subroutine}

The naive way to sample a solution is the rejection sampling algorithm, where we simply sample a uniform assignment and check if it happens to be a solution.

Starting with a (possibly empty) partial assignment $\sigma$, we can factorize $\Phi^\sigma$ into maximal connected components $\Phi_1,\Phi_2,\ldots$, where each $\Phi_i$ supports on disjoint subsets of the unassigned variables $\Lambda(\sigma)$.
Then $\mu^\sigma=\mu_1\times\mu_2\times\cdots$ is a product distribution where $\mu_i$ is the uniform distribution over solutions of $\Phi_i$.

Now assume we want to get a sample from $\mu_S^\sigma$, i.e., the marginal distribution of variables in $S\subseteq\Lambda(\sigma)$ in a uniform solution of $\Phi^\sigma$.
Assume $S=S_1\cup S_2\cup\cdots$ and each $S_i$ is contained in the support of $\Phi_i$. 
Then it suffices to get a sample from the marginal distribution of $S_i$ under $\mu_i$ for each $i$ independently and glue them together.
This is formalized in \Cref{alg:the_rejectionsampling_algorithm}.

Recall that $\Lambda(\sigma)$ is the set of unassigned (i.e., $\Qmark$ or $\EQmark$) variables in $\sigma$.
Our rejection sampling algorithm does not distinguish $\Qmark$ and $\EQmark$.

\begin{algorithm2e}[ht]
\caption{The \texttt{RejectionSampling} Algorithm}\label{alg:the_rejectionsampling_algorithm}
\DontPrintSemicolon
\KwIn{$\sigma\in\{0,1,\Qmark,\EQmark\}^\Vcal$ and $S\subseteq\Lambda(\sigma)$}
\KwOut{A random assignment distributed as $\mu_S^\sigma$}
\lnl{ln:rej_1} Let $\Phi_i=(\Vcal_i,\Ccal_i),i=1,2,\ldots$ be the maximal connected components in $\Phi^\sigma$ intersecting $S$\;
\nl \ForEach{$\Phi_i$}{
\lnl{ln:rej_2} \lRepeat{$\pi(\Vcal_i)$ is a solution of $\Phi_i$}{Sample $\pi(\Vcal_i)\sim\bin^{\Vcal_i}$}
}
\nl \Return{$\pi(S)$}
\end{algorithm2e}

We first note the simple correctness guarantee of \Cref{alg:the_rejectionsampling_algorithm}.

\begin{fact}\label{fct:rejectionsampling_correctness}
If $\Phi^\sigma$ is satisfiable, then \RejectionSampling{$\sigma,S$} terminates almost surely and has output distribution exactly $\mu_S^\sigma$.
\end{fact}

To analyze the efficiency, we will make the following assumption on the partial assignment and it will be preserved throughout our algorithm.
The intuition here is that, the partial assignment will not touch $\Ccalsep$ which involves high-degree variables, and for the other clauses it leaves enough number of variables alive that guarantees satisfiability and efficient sampling using \Cref{thm:local_uniformity}.

\begin{assumption}\label{as:marking_k'}
$\Vcalsep\subseteq\Lambda(\sigma)$ and for every clause $C\in\Ccal\setminus\Ccalsep$, either $C(\sigma)=\True$ or $|\vbl(C)\cap\Lambda(\sigma)\setminus\Vcalsep|\ge k'$ for some $k'\le(1-2\eta)k-2$.
\end{assumption}

We remark that the condition $k'\le(1-2\eta)k-2$ is for analysis convenience and is also reasonable considering \Cref{prop:width_k-2} and \Cref{ln:sep_1,ln:sep_2} of \ConstructSep{$\Vcal$}.
Later we will use it with $k'=(1-2\eta)k-2$ and $k'=(2/3-2\eta)k$ respectively in different scenarios.

\begin{lemma}\label{lem:marginal_local_uniformity_k'}
Assume $(\Phi,k,\alpha,n,\xi,\eta,D)$ is good and $\sigma$ satisfies \Cref{as:marking_k'}.
If $\Naturale2^{-k'}\cdot kD\le1$, then $\Phi^\sigma$ is satisfiable.
Moreover, for each $b\in\bin$ and $v\in\Lambda(\sigma)\setminus\Vcalsep$, we have
$$
\frac{1-\Naturale2^{-k'}D}2\le\mu_v^\sigma(b)\le\frac{1+\Naturale2^{-k'}D}2.
$$
\end{lemma}
\begin{proof}
Note that clauses in $\Ccalsep$ only depend on $\Vcalsep$. 
Since $\Phi$ is satisfiable, there exists a partial assignment $\pi$ extending $\sigma$ by fixing values of $\Vcalsep$ to $0$/$1$ and satisfying all clauses in $\Ccalsep$. 

Observe that $\pi$ only additionally fixes variables in $\Vcalsep$. Thus $\Lambda(\pi)=\Lambda(\sigma)\setminus\Vcalsep$.
Now it suffices to show for any such $\pi$, $\Phi^\pi$ is satisfiable and we have
$$
\frac{1-\delta}2
\le\mu_v^\sigma(b\,|\,\pi)=\mu_v^\pi(b)
\le\frac{1+\delta}2,
$$
where $\mu_v^\sigma(\cdot\,|\,\pi)$ is $\mu_v^\sigma$ conditioned on $\pi(\Lambda(\sigma)\cap\Vcalsep)$.

Since each clause $C\in\Ccal^\pi$ satisfies $|\vbl(C)|\ge k'$ where $\vbl(C)\subseteq\Lambda(\pi)$ is the set of remaining variables.
Thus
$$
\Pr_{\pi' \sim \bin^{\Lambda(\pi)}}\sbra{C(\pi') = \False} \leq 2^{-k'}.
$$
Since $\Lambda(\pi)\cap\Vcalsep=\emptyset$, every variable in $\Lambda(\pi)$ has variable degree at most $D$ in $\Phi^\pi$, and the constraint degree of $\Phi^\pi$ is at most $kD$.

Assuming $\Naturale2^{-k'}\cdot kD\le1$ and by \Cref{thm:local_uniformity}, $\Phi^\pi$ is satisfiable.
Moreover, with $B$ being event ``$v$ is assigned to $b$'' which correlates with at most $D$ clauses in $\Phi^\pi$, we have
\begin{align*}
\mu_v^\pi(b)
\le\frac{\pbra{1 - \Naturale2^{-k'}}^{-D}}2 
\le\frac{1 + \Naturale2^{-k'}D}2
\end{align*}
and the other direction follows from $\mu_v^\pi(b)=1-\mu_v^\pi(1-b)$ and the upper bound of $\mu_v^\pi(1-b)$.
\end{proof}

Now we show \Cref{ln:rej_2} of \RejectionSampling{$\sigma,S$} is efficient if $\Phi_i$ is small and $\sigma$ satisfies \Cref{as:marking_k'}.

\begin{lemma}\label{lem:rejectionsampling_efficiency}
Assume $(\Phi,k,\alpha,n,\xi,\eta,D)$ is good and $\sigma$ satisfies \Cref{as:marking_k'}. 
If $\Naturale2^{-k'}\cdot kD\le1$, then for $\Phi_i=(\Vcal_i,\Ccal_i)$ from \Cref{ln:rej_1} of \RejectionSampling{$\sigma,S$} we have
$$
\Pr_{\pi(\Vcal_i)\sim\bin^{\Vcal_i}}\sbra{\pi(\Vcal_i)\text{ is a solution of }\Phi_i}\ge
\exp\cbra{-\min\cbra{\frac{\xi|\Ccal_i|}{k^6(\alpha+1)},\frac{kn}{2^{4k}}+\frac{\Naturale|\Ccal_i|}{2^{k'}}}}.
$$
\end{lemma}
\begin{proof}
Let $\bar\Vcal=\cbra{v\in\Vcal_i\mid\deg_{\Ccal_i}(v)\ge D}\subseteq\HD(\Vcal_i)\subseteq\HD(\Vcal)$. Then $|\bar\Vcal|\le k|\Ccal_i|/D$ and by \Cref{prop:high-degree}, $|\bar\Vcal|\le|\HD(\Vcal)|\le n/2^{4k}$. 
Recall that $\bar\Vcal_{\mathsf{sep}}$ and $\bar\Ccal_{\mathsf{sep}}$ are the outputs of \ConstructSep{$\bar\Vcal$}.
Then by \Cref{lem:fraction_of_high-degree_starting_from_high-degree} and $2/\eta\le k$, 
\begin{equation}\label{eq:lem:rejectionsampling_efficiency_1}
|\bar\Vcal_{\mathsf{sep}}|\le\frac{2|\bar\Vcal|}\eta
\le\frac2\eta\cdot\min\cbra{\frac{k|\Ccal_i|}D,\frac n{2^{4k}}}
\le\min\cbra{\frac{2k|\Ccal_i|}{\eta D},\frac{kn}{2^{4k}}}.
\end{equation}

Let $\Vcal'=\Vcal_i\cap\bar\Vcal_\mathsf{sep}$ and $\Ccal'=\Ccal_i\cap\bar\Ccal_\mathsf{sep}$.
By \Cref{lem:marginal_local_uniformity_k'}, $\Phi^\sigma$ is satisfiable, and thus $\Phi_i$ is also satisfiable.
Therefore there exists a partial assignment $\tilde\pi$ extending $\sigma$ by fixing values of $\Vcal'$ to $0$/$1$ and satisfying all clauses in $\Ccal'$.
Then
\begin{align}
\Pr_{\pi(\Vcal_i)\sim\bin^{\Vcal_i}}\sbra{\Phi_i(\pi(\Vcal_i))=\True}
&\ge2^{-|\Vcal'|}\Pr_{\pi(\Vcal_i)\sim\bin^{\Vcal_i}}\sbra{\Phi_i(\pi(\Vcal_i))=\True\mid\pi(\Vcal')=\tilde\pi(\Vcal')}
\notag\\
&\ge2^{-|\bar\Vcal_\mathsf{sep}|}\Pr_{\pi(\Vcal_i)\sim\bin^{\Vcal_i}}\sbra{\Phi_i(\pi(\Vcal_i))=\True\mid\pi(\Vcal')=\tilde\pi(\Vcal')}
\notag\\
&\ge\Naturale^{-|\bar\Vcal_\mathsf{sep}|}\Pr_{\pi(\Vcal_i)\sim\bin^{\Vcal_i}}\sbra{\Phi_i(\pi(\Vcal_i))=\True\mid\pi(\Vcal')=\tilde\pi(\Vcal')}.
\label{eq:lem:rejectionsampling_efficiency_2}
\end{align}

By our choice of $\tilde\pi$, clauses in $\Ccal'$ are already satisfied.
On the other hand, since $\Vcal'\subseteq\bar\Vcal_\mathsf{sep}\subseteq\Vcalsep$ by \Cref{fct:sep_increasing}, every clause $C\in\Ccal_i\setminus\Ccal'$ that is not satisfied by $\tilde\pi$ falls into one of the following cases:
\begin{itemize}
\item If $C$ was not originally in $\Ccalsep$, then it contains at least $k'$ unassigned variables in $\tilde\pi$ by \Cref{as:marking_k'} since $\Vcal'\subseteq\Vcalsep$ and $\Vcalsep\subseteq\Lambda(\sigma)$.
\item Otherwise, $C$ was originally in $\Ccalsep$.
Then in $\sigma$, it contains at least $k-2$ unassigned variables by \Cref{prop:width_k-2} and \Cref{as:marking_k'}.
Now in $\tilde\pi$, at most $2\eta k$ variables are in $\bar\Vcal_\mathsf{sep}$ and thus fixed, which means at least $k-2-2\eta k\ge k'$ variables remain.
\end{itemize}
In addition, all the remaining variables $\Vcal_i\setminus\Vcal'$ have degree at most $D$.

Let $\Phi''=(\Vcal'',\Ccal'')$ where $\Vcal''=\Vcal_i\setminus\Vcal'$ and $\Ccal''\subseteq\Ccal_i\setminus\Ccal'$.
Then $p(\Phi'')\le2^{-k'}$ and $\Delta(\Phi'')\le kD$.
Since $\Naturale2^{-k'}\cdot kD\le1$, by \Cref{thm:local_uniformity} with $B$ being the event ``$\Phi''$ is satisfied'' which correlates with all $|\Ccal''|\le|\Ccal_i|$ clauses, we have
$$
\Pr_{\pi(\Vcal_i)\sim\bin^{\Vcal_i}}\sbra{\Phi_i(\pi(\Vcal_i))=\True\mid\pi(\Vcal')=\tilde\pi(\Vcal')}
\ge(1-\Naturale2^{-k'})^{|\Ccal''|}
\ge\exp\cbra{-\frac{\Naturale|\Ccal_i|}{2^{k'}}}.
$$
Putting \Cref{eq:lem:rejectionsampling_efficiency_1} and \Cref{eq:lem:rejectionsampling_efficiency_2} back, we have
\begin{align*}
\Pr_{\pi(\Vcal_i)\sim\bin^{\Vcal_i}}\sbra{\Phi_i(\pi(\Vcal_i))=\True}
&\ge\exp\cbra{-\min\cbra{\frac{2k|\Ccal_i|}{\eta D},\frac{kn}{2^{4k}}}-\frac{\Naturale|\Ccal_i|}{2^{k'}}}\\
&\ge\exp\cbra{-\min\cbra{\frac{2k|\Ccal_i|}{\eta D}+\frac{|\Ccal_i|}{kD},\frac{kn}{2^{4k}}+\frac{\Naturale|\Ccal_i|}{2^{k'}}}}
\tag{since $\Naturale2^{-k'}\cdot kD\le1$}\\
&\ge\exp\cbra{-\min\cbra{\frac{k^2|\Ccal_i|}D,\frac{kn}{2^{4k}}+\frac{\Naturale|\Ccal_i|}{2^{k'}}}}
\tag{since $3/k\le\eta\le1$}\\
&=\exp\cbra{-\min\cbra{\frac{\xi|\Ccal_i|}{k^6(\alpha+1)},\frac{kn}{2^{4k}}+\frac{\Naturale|\Ccal_i|}{2^{k'}}}},
\tag{since $D=k^8(\alpha+1)/\xi$}
\end{align*}
as desired.
\end{proof}

\begin{corollary}\label{cor:rejectionsampling_efficiency}
Assume $(\Phi,k,\alpha,n,\xi,\eta,D)$ is good and $\sigma$ satisfies \Cref{as:marking_k'}.
If $\Naturale2^{-k'}\cdot kD\le1$, then \RejectionSampling{$\sigma,S$} runs in expected time
$$
\tilde O\pbra{\sum_i|\Vcal_i|\cdot
\exp\cbra{\min\cbra{\frac{\xi|\Ccal_i|}{k^6(\alpha+1)},\frac{kn}{2^{4k}}+\frac{\Naturale|\Ccal_i|}{2^{k'}}}}},
$$
where each $\Phi_i=(\Vcal_i,\Ccal_i)$ is from \Cref{ln:rej_1} of \RejectionSampling{$\sigma,S$}.
\end{corollary}

\subsection{Algorithms for the Atypical Setting}\label{sec:atypical_setting}

To give a sense of the bound in \Cref{cor:rejectionsampling_efficiency}, we use it to analyze the atypical setting of \Cref{thm:main_algorithm} where either $\eps$ or $\alpha$ is too small.
Indeed, in these cases the naive rejection sampling algorithm is already highly efficient.

\begin{lemma}[Small Error Setting]\label{lem:small_eps}
Assume $(\Phi,k,\alpha,n,\xi,\eta,D)$ is good and $\eps\le\exp\cbra{-n/2^{k/2}}$.
Then \RejectionSampling{$\EQmark^\Vcal,\Vcal$} runs in expected time $\tilde O\pbra{(1/\eps)^{\xi/k}}$ and has output distribution exactly $\mu$.
\end{lemma}
\begin{proof}
By \Cref{fct:rejectionsampling_correctness}, we only need to bound the expected runtime.
By \Cref{prop:width_k-2} and \Cref{ln:sep_1,ln:sep_2} of \ConstructSep{$\Vcal$}, we set $k'=(1-2\eta)k-2$ in \Cref{cor:rejectionsampling_efficiency}.
Since $\eta=15\log(k)/k$, $D=k^8(\alpha+1)/\xi$, and $\alpha\le\xi\cdot2^{k/3}/k^{50}$ with $\xi\ge2^{-k/8}$ and $k\ge2^{20}$, we have
$$
\Naturale2^{-k'}\cdot kD
=4\Naturale2^{-k}\cdot k^{39}(\alpha+1)/\xi
\le8\Naturale k^{-11}\cdot2^{-2k/3}
\le1.
$$
Then by \Cref{cor:rejectionsampling_efficiency}, the expected runtime is upper bounded by
\begin{align*}
\tilde O\pbra{\sum_i|\Vcal_i|\cdot\exp\cbra{\frac{kn}{2^{4k}}+\frac{4\Naturale k^{30}|\Ccal_i|}{2^k}}}
&\le
\tilde O\pbra{\sum_i|\Vcal_i|\cdot\exp\cbra{\frac{kn}{2^{4k}}+\frac{4\Naturale n}{2^{2k/3}}}}
\tag{since $|\Ccal_i|\le|\Ccal|=\alpha n\le n\cdot2^{k/3}/k^{30}$}\\
&\le
\tilde O\pbra{\sum_i|\Vcal_i|\cdot\exp\cbra{\frac n{2k\cdot2^{5k/8}}}}
\tag{since $k\ge2^{20}$}\\
&=
\tilde O\pbra{n\cdot\exp\cbra{\frac n{2k\cdot2^{5k/8}}}}
\le
\tilde O\pbra{\exp\cbra{\frac n{k\cdot2^{5k/8}}}}
\tag{since $n\ge2^{\Omega(k)}$}\\
&\le
\tilde O\pbra{(1/\eps)^{\xi/k}}
\tag{since $\eps\le\exp\cbra{-n/2^{k/2}}$ and $\xi\ge2^{-k/8}$}
\end{align*}
as desired.
\end{proof}

\begin{lemma}[Small Density Setting]\label{lem:small_alpha}
Assume $(\Phi,k,\alpha,n,\xi,\eta,D)$ is good and $\alpha\le1/k^3$.
Then \RejectionSampling{$\EQmark^\Vcal,\Vcal$} runs in expected time $\tilde O\pbra{n^{1+\xi/k}}$ and has output distribution exactly $\mu$.
\end{lemma}
\begin{proof}
Similar analysis as in the proof of \Cref{lem:small_eps}.
In addition, by \Cref{prop:number_of_connected_sets} with $\ell=\ln n$, we have $\alpha^2n^4(\Naturale k^2\alpha)^\ell<1$ and thus the maximal connected component in $G_\Phi$ has size at most $\ln n$, i.e., each $\Ccal_i$ in \RejectionSampling{$\EQmark^\Vcal,\Vcal$} has size at most $\ln n$.
Then by \Cref{cor:rejectionsampling_efficiency}, the expected runtime is upper bounded by
\begin{equation*}
\tilde O\pbra{\sum_i|\Vcal_i|\cdot\exp\cbra{\frac{\xi\ln n}{k^6(\alpha+1)}}}
\le
\tilde O\pbra{n\cdot\exp\cbra{\frac{\xi\ln n}k}}
=
\tilde O\pbra{n^{1+\xi/k}}
\tag*{\qedhere}
\end{equation*}
\end{proof}
\section{Algorithms for the Typical Setting}\label{sec:the_sampling_algorithm}

In this section, we present the sampling algorithm for the typical setting: $\alpha\ge1/k^3$ and $\eps\ge\exp\cbra{-n/2^{k/2}}$.
We will conveniently assume our instance is nice (in particular, $\alpha\ge1/k^3$), though some of the results also hold with weaker assumptions.
From now on, unless specifically mentioned, we assume $(\Phi,k,\alpha,n,\xi,\eta,D)$ is nice and save the space of repeatedly putting this in the statements.

Our main algorithm is a modification of the ones in \cite{he2022sampling}. Hence some of our notation and definitions will be similar to theirs, which we hope is easier to understand if the reader is already familiar with \cite{he2022sampling}.

Given a partial assignment $\sigma$ and $\Vcalsep$ constructed above, we define $\Vcalalive^\sigma$: For each $v\in\Vcal$, $v\in\Vcalalive^\sigma$ iff (i) $\sigma(v)=\EQmark$ and $v\notin\Vcalsep$, and (ii) for every clause $C\in\Ccal\setminus\Ccalsep$, either $C(\sigma)=\True$ or $|\vbl(C)\cap\Lambda(\sigma)\setminus\pbra{\Vcalsep\cup\cbra{v}}|\ge(2/3-2\eta)k$.
Intuitively, $v\in\Vcalalive^\sigma$ means after fixing $v$, each unsatisfied clause will still contain many unassigned variables, consistent with \Cref{as:marking_k'}.

Now we present our \SolutionSampling{$\Phi$} algorithm in \Cref{alg:the_main_algorithm} similar to \cite[Algorithm 4]{he2022sampling}.

\begin{algorithm2e}[ht]
\caption{The \texttt{SolutionSampling} Algorithm for the Typical Setting}\label{alg:the_main_algorithm}
\DontPrintSemicolon
\KwIn{A random $k$-CNF formula $\Phi=(\Vcal,\Ccal)$}
\KwOut{A random assignment $\sigma$ distributed as $\mu$}
\lnl{ln:sol_1} Obtain $\Vcalsep,\Ccalsep\gets\ConstructSep{$\Vcal$}$\;
\nl Initialize $\sigma\gets\EQmark^\Vcal$\;
\nl \ForEach{$i=1$ \KwTo $n$}{
\lnl{ln:sol_2} \lIf{$v_i\in\Vcalalive^\sigma$}{
Update $\sigma(v_i)\gets\MarginSample{$\sigma,v_i$}$
}
}
\lnl{ln:sol_3} $\sigma\gets\RejectionSampling{$\sigma,\Lambda(\sigma)$}$\;
\nl \Return{$\sigma$}
\end{algorithm2e}

By dynamically maintaining and updating the size of each $\vbl(C)\cap\Lambda(\sigma)\setminus\Vcalsep$, the total runtime for checking whether $v_i\in\Vcalalive^\sigma$ is very efficient.

\begin{fact}\label{fct:vcalalive_efficiency}
The runtime of all the checking $v_i\in\Vcalalive^\sigma$ combined is $\tilde O(n)$.
\end{fact}

\Cref{as:marking_k'} will be preserved with $k'=(2/3-2\eta)k$ throughout if we only update the alive variables.
This is indeed the case in \Cref{alg:the_main_algorithm} and we highlight it as the following formal statements.
For future referencing, we explicitly write \Cref{as:marking_k'} with $k'=(2/3-2\eta)k$ as \Cref{as:marking}.

\begin{assumption}\label{as:marking}
$\Vcalsep\subseteq\Lambda(\sigma)$ and for every clause $C\in\Ccal\setminus\Ccalsep$, either $C(\sigma)=\True$ or $|\vbl(C)\cap\Lambda(\sigma)\setminus\Vcalsep|\ge(2/3-2\eta)k$.
\end{assumption}

\begin{fact}\label{fct:marking}
Assume we construct a partial assignment $\sigma$ by starting with $\sigma=\EQmark^\Vcal$ and repeatedly fixing variables in $\Vcalalive^\sigma$.
Then $\sigma$ satisfies \Cref{as:marking}.
\end{fact}
\begin{proof}
We prove by induction.
The base case $\sigma=\EQmark^\Vcal$ trivially holds since for any $C\in\Ccal\setminus\Ccalsep$, we have $|\vbl(C)|\ge k-2$ by \Cref{prop:width_k-2} and $|\vbl(C)\cap\Vcalsep|\le2\eta k$ by \Cref{ln:sep_2} of \ConstructSep{$\Vcal$}.

For the inductive case, assume we fix $v\in\Vcalalive^\sigma$ and obtain $\sigma'$.
Then by the definition of $\Vcalalive^\sigma$, we know $v\not\in\Vcalsep$, which means $\Lambda(\sigma')\supseteq\Lambda(\sigma)\setminus\cbra{v}\supseteq\Vcalsep$ by induction hypothesis.
On the other hand, for every clause $C\in\Ccal\setminus\Ccalsep$, 
\begin{itemize}
\item if $C(\sigma)=\True$, then $C(\sigma')=\True$,
\item otherwise, $|\vbl(C)\cap\Lambda(\sigma')\setminus\Vcalsep|\ge|\vbl(C)\cap\Lambda(\sigma)\setminus(\Vcalsep\cup\cbra{v})|\ge(2/3-2\eta)k$ since $v\in\Vcalalive^\sigma$.
\qedhere
\end{itemize}
\end{proof}

As a corollary of \Cref{lem:marginal_local_uniformity_k'}, we have good control for the marginal of every remaining variable outside $\Vcalsep$, and in particular, any $v\in\Vcalalive^\sigma$.

\begin{corollary}[Local Uniformity]\label{cor:marginal_local_uniformity}
Assume $\sigma$ satisfies \Cref{as:marking}.
Then $\Phi^\sigma$ is satisfiable.
Moreover, for each $b\in\bin$ and $v\in\Lambda(\sigma)\setminus\Vcalsep$, we have
$$
\frac{1-\delta}2\le\mu_v^\sigma(b)\le\frac{1+\delta}2,
$$
where $\delta=\xi/(k^{40}\alpha)$.
\end{corollary}
\begin{proof}
Let $k'=(2/3-2\eta)k$.
Since $\eta=15\log(k)/k$, $D=k^8(\alpha+1)/\xi$, and $\alpha\le\xi\cdot2^{k/3}/k^{50}$ with $\xi\ge2^{-k/8}$ and $k\ge2^{20}$, we have
$$
\Naturale2^{-k'}\cdot kD
=\Naturale2^{-2k/3}\cdot k^{39}(\alpha+1)/\xi
\le2\Naturale k^{-11}
\le1.
$$
Then by \Cref{lem:marginal_local_uniformity_k'}, we know $\Phi^\sigma$ is satisfiable and
\begin{align*}
\abs{\mu_v^\sigma(b)-\frac12}
&\le\frac12\cdot\Naturale2^{-k'}D
=\frac12\cdot\Naturale2^{-2k/3}\cdot k^{38}(\alpha+1)/\xi\\
&\le\frac12\cdot2\Naturale2^{-2k/3}\cdot k^{41}\cdot\alpha/\xi.
\tag{since $\alpha\ge1/k^3$}
\end{align*}
Since $\alpha\le\xi\cdot2^{k/3}/k^{50}$, we have
$$
\frac{2\Naturale2^{-2k/3}\cdot k^{41}\cdot\alpha/\xi}\delta
=\frac{2\Naturale k^{81}\cdot\alpha^2}{2^{2k/3}\cdot\xi^2}
\le\frac{2\Naturale}{k^{19}}
\le1
$$
and thus $\abs{\mu_v^\sigma(b)-1/2}\le\delta/2$.
\end{proof}

Similarly, we have the following efficiency bound for the rejection sampling after replacing \Cref{as:marking_k'} with \Cref{as:marking} in \Cref{cor:rejectionsampling_efficiency}.

\begin{corollary}\label{cor:rejectionsampling_efficiency_marking}
Assume $\sigma$ satisfies \Cref{as:marking}.
Then \RejectionSampling{$\sigma,S$} runs in expected time
$$
\tilde O\pbra{\sum_i|\Vcal_i|\cdot
\exp\cbra{\frac{\xi|\Ccal_i|}{k^6(\alpha+1)}}},
$$
where each $\Phi_i=(\Vcal_i,\Ccal_i)$ is from \Cref{ln:rej_1} of \RejectionSampling{$\sigma,S$}.
\end{corollary}

For convenience, we will reserve $\delta=\xi/(k^{40}\alpha)$ as the local uniformity parameter from now on.

To obtain the correct marginal distribution for each \Cref{ln:sol_2}, \MarginSample{$\sigma,v$} should sample from $\mu_v^\sigma$.
By \Cref{cor:marginal_local_uniformity}, this distribution is $\delta$-close to an unbiased coin.
This inspires us to define the following distribution $\tau$ as a ``lower bound'' for any $\mu_v^\sigma$ that $v\in\Vcalalive^\sigma$:
$$
\tau=\begin{cases}
0 & \text{w.p}\quad(1-\delta)/2,\\
1 & \text{w.p}\quad(1-\delta)/2,\\
\Qmark & \text{w.p}\quad\delta.
\end{cases}
$$
As described in \Cref{alg:the_marginsample_algorithm}, \MarginSample{$\sigma,v$} will first naively sample from $\tau$, and resample using \MarginOverflow{$\sigma,v$} if obtained $\Qmark$ from $\tau$.

\begin{algorithm2e}[ht]
\caption{The \texttt{MarginSample} Algorithm}\label{alg:the_marginsample_algorithm}
\DontPrintSemicolon
\KwIn{$\sigma\in\{0,1,\Qmark,\EQmark\}^\Vcal$ and $v\in\Vcalalive^\sigma$}
\KwOut{A binary random variable distributed as $\mu_v^\sigma$}
\nl Sample $\sigma(v)\sim\tau$\;
\nl \lIf{$\sigma(v)=\Qmark$}{\Return{\MarginOverflow{$\sigma,v$}}}
\nl \lElse{\Return{$\sigma(v)$}}
\end{algorithm2e}

Naturally, \MarginOverflow{$\sigma,v$} should complete $\tau$ into $\mu_v^\sigma$. 
Thus it should output a binary bit distributed proportional to $\mu_v^\sigma-\tau$, which we define as $\nu_v^\sigma$: For each $b\in\bin$, we set
$$
\nu_v^\sigma(b)=\frac{\mu_v^\sigma(b)-\tau(b)}{\tau(\Qmark)}=\frac{\mu_v^\sigma(b)-(1-\delta)/2}\delta.
$$
On the other hand, there exists a standard toolbox \cite{nacu2005fast,huber2016nearly,dughmi2021bernoulli}, called \emph{Bernoulli factory}, to provide samples from $\nu_v^\sigma$, which is a linear function of $\mu_v^\sigma$, using samples from $\mu_v^\sigma$.
Here we use the statement in \cite{he2022sampling}:
\begin{lemma}[{\cite[Appendix A]{he2022sampling}}]\label{lem:bernoulli_factory}
There exists a Las Vegas algorithm \BernoulliFactory{} such that the following holds:
Assume $b_1,b_2,\ldots$ are independent samples from $\mu_v^\sigma$, where $\mu_v^\sigma$ is unknown to the algorithm and $\mu_v^\sigma(b)\ge\tau(b)$ for $b\in\bin$.
Then \BernoulliFactory{$b_1,b_2,\ldots$} runs in expected time $\tilde O(1/\delta^2)=\tilde O(1/\xi^2)$ and has output distribution exactly $\nu_v^\sigma$.
\end{lemma}

Samples from $\mu_v^\sigma$ can be provided by executing \RejectionSampling{$\sigma,v$}, but simply doing so is just self-referencing: Why not let \MarginSample{$\sigma,v$} be \RejectionSampling{$\sigma,v$} in the first place?

The trick here is to recursively fix more variables in $\sigma$ and postpone the Bernoulli factory to the end. Hopefully at that point, most variables in $\sigma$ are fixed and $\Phi^\sigma$ can be factorized into small components, which makes the rejection sampling efficient.
This will become rigorous as we describe \MarginOverflow{$\sigma,v$} shortly.

\subsection{The Margin Overflow Algorithm and Truncation}\label{sec:marginoverflow}

To describe and analyze \MarginOverflow{}, we need the following notation to make rigorous our recursive sampling order: Let $\sigma$ be a partial assignment.
For each $C\in\Ccal$,
\begin{itemize}
\item $\CcalQmark^\sigma$: $C\in\CcalQmark^\sigma$ iff there exists some $v\in\vbl(C)$ that $\sigma(v)=\Qmark$.
\item $\Ccalfrozen^\sigma$: $C\in\Ccalfrozen^\sigma$ iff (i) $C(\sigma)\neq\True$ and $C\notin\Ccalsep$, and (ii) $|\vbl(C)\cap\Lambda(\sigma)\setminus\Vcalsep|<1+(2/3-2\eta)k$.
\item $\Ccalbad^\sigma$: $C\in\Ccalbad^\sigma$ iff (i) $C(\sigma)\neq\True$ and $C\notin\Ccalfrozen^\sigma\cup\Ccalsep$, and (ii) for any $v\in\vbl(C)\setminus\Vcalsep$ with $\sigma(v)=\EQmark$, there exists some $C'\in\Ccalfrozen^\sigma$ such that $v\in\vbl(C')$.
\end{itemize}

We remark that, though $\Ccalfrozen^\sigma\cap\Ccalbad^\sigma=\emptyset$, it is possible that $\CcalQmark^\sigma\cap\Ccalfrozen^\sigma\neq\emptyset$ and $\CcalQmark^\sigma\cap\Ccalbad^\sigma\neq\emptyset$.
The definition of $\Ccalfrozen^\sigma$ is a direct opposite of $\Vcalalive^\sigma$, while $\Ccalbad^\sigma$ intermediately violates $\Vcalalive^\sigma$ due to $\Ccalfrozen^\sigma$.
This is formalized in the following fact.

\begin{fact}\label{fct:alive_not_frozen_bad_sep}
For any $C\in\Ccalsep\cup\Ccalfrozen^\sigma\cup\Ccalbad^\sigma$, we have $\vbl(C)\cap\Vcalalive^\sigma=\emptyset$.
As a consequence, we have $\CcalQmark^\sigma\subseteq\CcalQmark^{\sigma'}$, $\Ccalfrozen^\sigma\subseteq\Ccalfrozen^{\sigma'}$, and $\Ccalbad^\sigma\subseteq\Ccalbad^{\sigma'}$ if $\sigma'$ extends $\sigma$ by fixing some variable in $\Vcalalive^\sigma$.
\end{fact}
\begin{proof}
Recall that $v\in\Vcalalive^\sigma$ iff (a) $\sigma(v)=\EQmark$, (b) $v\notin\Vcalsep$, and (c) for every clause $C\in\Ccal\setminus\Ccalsep$ and $C(\sigma)\neq\True$, $|\vbl(C)\cap\Lambda(\sigma)\setminus(\Vcalsep\cup\cbra{v})|\ge(2/3-2\eta)k$.

Now assume $v\in\vbl(C)\cap\Vcalalive^\sigma$.
\begin{itemize}
\item If $C\in\Ccalsep$, then by the definition of $\Vcalsep$, we have $v\in\Vcalsep$ and contradict to (b).
\item If $C\in\Ccalfrozen^\sigma$, then by the definition of $\Ccalfrozen^\sigma$, we have $|\vbl(C)\cap\Lambda(\sigma)\setminus(\Vcalsep\cup\cbra{v})|<(2/3-2\eta)k$ and contradict to (c).
\item If $C\in\Ccalbad^\sigma$, then by (a), we know $\sigma(v)=\EQmark$. Then by the definition of $\Ccalbad^\sigma$, we have $v\in\vbl(C')$ for some $C'\in\Ccalfrozen^\sigma$ and contradict to the last item with $C$ replaced by $C'$.
\qedhere
\end{itemize}
\end{proof}

To preserve \Cref{as:marking} and by \Cref{fct:marking}, we can only afford to sample variables in $\Vcalalive^\sigma$.
By \Cref{fct:alive_not_frozen_bad_sep}, this means we need to avoid $\Ccalfrozen^\sigma\cup\Ccalbad^\sigma\cup\Ccalsep$.

On the other hand, the marginal distribution of a variable depends on all the variables and clauses connected to it.
This motivates us to define, for each $v$ with $\sigma(v)=\Qmark$, the bad interior $\Ccalint^\sigma(v)\subseteq\Ccalfrozen^\sigma\cup\Ccalbad^\sigma\cup\Ccalsep$, which contains variables connected to $v$ (but unfortunately none of them is alive).
Formally, we put $C\in\Ccalfrozen^\sigma\cup\Ccalbad^\sigma\cup\Ccalsep$ into $\Ccalint^\sigma(v)$ iff either $v\in\vbl(C)$ or there exists some $C'\in\Ccalint^\sigma(v)$ that $\vbl(C')\cap\vbl(C)\cap\Lambda(\sigma)\neq\emptyset$.

To obtain alive variables to sample from, we need to take one step further to form the current component $\Ccalcon^\sigma(v)$.
Formally, we put $C\in\Ccal$ into $\Ccalcon^\sigma(v)$ iff $C\in\Ccalint^\sigma(v)$, or $v\in\vbl(C)$, or there exists some $C'\in\Ccalint^\sigma(v)$ that $\vbl(C')\cap\vbl(C)\cap\Lambda(\sigma)\neq\emptyset$.

Then we take the union of the current components of all the $\Qmark$'s, since we care about the marginals of these variables.
Define $\Ccalcon^\sigma$ to be the union of $\Ccalcon^\sigma(v)$ for all $v$ with $\sigma(v)=\Qmark$.
Let $\Vcalcon^\sigma=\bigcup_{C\in\Ccalcon^\sigma}\vbl(C)$.
Then we define
$$
\NextVar(\sigma)=\begin{cases}
v_i\in\Vcalalive^\sigma\cap\Vcalcon^\sigma\text{ with smallest }i & \text{if }\Vcalalive^\sigma\cap\Vcalcon^\sigma\neq\emptyset,\\
\bot & \text{otherwise},
\end{cases}
$$
which will be the function for selecting the next variable to perform marginal sampling.

Now we give the pseudo-code of \MarginOverflow{$\sigma,v$}.
We remark that \Cref{ln:mo_3,ln:mo_4} is equivalent to calling \MarginSample{$\sigma,u$}. 
To avoid confusion from subroutines calling each other, we choose to expand it out as the current presentation.

\begin{algorithm2e}[ht]
\caption{The \texttt{MarginOverflow} Algorithm}\label{alg:the_marginoverflow_algorithm}
\DontPrintSemicolon
\KwIn{$\sigma\in\{0,1,\Qmark,\EQmark\}^\Vcal$ and $v\in\Vcal$ with $\sigma(v)=\Qmark$ and $v\in\Vcalalive^{\bar\sigma}$ where $\bar\sigma$ equals $\sigma$ except $\bar\sigma(v)=\EQmark$}
\KwOut{A binary random variable distributed as $\nu_v^\sigma$}
\lnl{ln:mo_1} Let $u\gets\NextVar(\sigma)$\;
\nl \eIf{$u\neq\bot$}{
\lnl{ln:mo_3} Sample $\sigma(u)\sim\tau$\;
\lnl{ln:mo_4} \lIf{$\sigma(u)=\Qmark$}{Update $\sigma(u)\gets\MarginOverflow{$\sigma,u$}$}
\lnl{ln:mo_5} \Return{\MarginOverflow{$\sigma,v$}}
}{
\nl \Return{\BernoulliFactory{$b_1,b_2,\ldots$}} where $b_1,b_2,\ldots$ are independent samples provided by executing \RejectionSampling{$\sigma,v$}
}
\end{algorithm2e}

By pre-processing the maximal connected components in $\Ccalsep$, we can dynamically maintain $\CcalQmark^\sigma,\Ccalfrozen^\sigma,\Ccalbad^\sigma,\Vcalalive^{\sigma}$ and their connectivity relation with $\Ccalsep$.
Then we can dynamically update $\Ccalcon^\sigma$ and $\Vcalcon^\sigma$ by joining the new connected components.

Therefore each computation of $\NextVar()$ can be done in worst case time $\poly(k,d)$, where $d=d(\Phi)$ is the maximum variable degree of $\Phi$. 
By \Cref{prop:maximum_degree}, we obtain the following bound.

\begin{fact}\label{fct:nextvar_efficiency}
With $\tilde O(n)$ pre-processing time, the runtime of each $\NextVar()$ is $\tilde O(1)$.
\end{fact}

Note that whenever $u$ from \Cref{ln:mo_1} is not $\bot$, we have $u\in\Vcalalive^\sigma$.
Therefore by \Cref{fct:marking}, \Cref{as:marking} is preserved throughout the algorithm.
This provides us the following correctness guarantee, the proof of which is almost identical to the inductive proof of \cite[Theorem 5.5]{he2022sampling}.

\begin{lemma}\label{lem:correctness_of_marginoverflow}
Assume $\sigma$ satisfies \Cref{as:marking}.
Then \MarginOverflow{$\sigma,v$} terminates almost surely and has output distribution exactly $\nu_v^\sigma$.
\end{lemma}
\begin{proof}
Observe that each deeper recursion will have the value of $u$ changed from $\EQmark$ to $0$/$1$/$\Qmark$.
Therefore the number of $\EQmark$'s in $\sigma$ is decreasing and thus the recursion ends eventually.

Now we prove the statement by induction on $\sigma$.
The base case corresponds to the leaf of the recursion, where $u=\bot$ from \Cref{ln:mo_1}.
By \Cref{cor:marginal_local_uniformity}, we know $\Phi^\sigma$ is satisfiable.
Thus by \Cref{fct:rejectionsampling_correctness}, \RejectionSampling{$\sigma,v$} terminates almost surely and has output distribution exactly $\mu_v^\sigma$.
Now by \Cref{cor:marginal_local_uniformity}, $\mu_v^\sigma$ is lower bounded by $\tau$.
Therefore by \Cref{lem:bernoulli_factory}, \BernoulliFactory{$b_1,b_2,\ldots$} has output distribution exactly $\nu_v^\sigma$.
This proves the base case.

For the inductive case that $u\neq\bot$, let $\sigma_0,\sigma_1,\sigma_\Qmark$ equal $\sigma$ except $\sigma_0(u)=0/1/\Qmark$ respectively.
Then by induction hypothesis, \MarginOverflow{$\sigma_\Qmark,u$} returns a bit distributed as 
$$
\nu_{\sigma_\Qmark}^u(b)=\frac{\mu_\sigma^u(b)-\tau(b)}{\tau(\Qmark)}
\quad\text{for $b\in\bin$.}
$$
Let $\sigma'$ be the updated $\sigma$ upon reaching \Cref{ln:mo_5}.
Thus $\sigma'$ equals $\sigma$ except
\begin{equation}\label{eq:lem:correctness_of_marginoverflow}
\Pr[\sigma'(u)=b]=\tau(b)+\tau(\Qmark)\cdot\nu_{\sigma_\Qmark}^u(b)=\mu_\sigma^u(b)
\quad\text{for $b\in\bin$.}
\end{equation}
By induction hypothesis again, \Cref{ln:mo_5} terminates almost surely and obtains distribution $\nu_{\sigma'}^v$ where
\begin{align*}
\nu_{\sigma'}^v(b)
&=\Pr[\sigma'(u)=0]\cdot\nu_{\sigma_0}^v(b)+\Pr[\sigma'(u)=1]\cdot\nu_{\sigma_1}^v(b)\\
&=\mu_\sigma^u(0)\cdot\nu_{\sigma_0}^v(b)+\mu_\sigma^u(1)\cdot\nu_{\sigma_1}^v(b)
\tag{by \Cref{eq:lem:correctness_of_marginoverflow}}\\
&=\frac1{\tau(\Qmark)}\cdot\pbra{\mu_\sigma^u(0)\cdot\mu_{\sigma_0}^v(b)+\mu_\sigma^u(1)\cdot\mu_{\sigma_1}^v(b)-\tau(b)}
\tag{by the definition of $\nu$}\\
&=\frac1{\tau(\Qmark)}\cdot\pbra{\mu_\sigma^v(b)-\tau(b)}
=\nu_\sigma^v(b)
\end{align*}
for $b\in\bin$ as desired.
\end{proof}

As an immediate corollary, we obtain the correctness of \MarginSample{$\sigma,v$}.

\begin{corollary}\label{cor:correctness_of_marginsample}
Assume $\sigma$ satisfies \Cref{as:marking}.
Then \MarginSample{$\sigma,v$} terminates almost surely and has output distribution exactly $\mu_v^\sigma$.
\end{corollary}

Then by the chain rule of conditional probability, \Cref{fct:marking}, and \Cref{fct:rejectionsampling_correctness}, we obtain the correctness of our main algorithm.
\begin{corollary}\label{cor:correctness_of_main_algorithm}
\SolutionSampling{$\Phi$} terminates almost surely and has output distribution exactly $\mu$.
\end{corollary}

Ideally, we only need to bound the expected runtime of each \MarginOverflow{$\sigma,v$} and the final rejection sampling in \SolutionSampling{$\Phi$}; then we obtain the runtime of the whole algorithm.
This will actually be a \emph{perfect} sampler that outputs an uniform solution exactly, and is indeed the case for the standard $k$-CNFs in the local lemma regime \cite{he2022sampling}.
But the issue here is that $\Phi$ is random, and its structural properties break down when we analyze components of large size.
Therefore to ensure that we have good structural properties at hand, we will have to halt when the component goes beyond a certain size.
Fortunately, we are able to show that this truncation happens with small probability, and thus only incur small deviation in the total variation distance.

To give some intuition about the truncation, we analyze the efficiency of the leaf recursion of \MarginOverflow{$\sigma,v$}.
Recall our definition of $\Vcalcon^\sigma$ and $\Ccalcon^\sigma$ at the beginning of this subsection.

\begin{lemma}\label{lem:bf_efficiency_given_ccon}
Assume $\sigma$ satisfies \Cref{as:marking}.
If $\NextVar(\sigma)=\bot$, then \MarginOverflow{$\sigma,v$} runs in expected runtime
$$
\tilde O\pbra{\frac{|\Ccalcon^\sigma|}{\xi^2}\cdot\exp\cbra{\frac{\xi\cdot|\Ccalcon^\sigma|}{k^6(\alpha+1)}}}.
$$
\end{lemma}
\begin{proof}
Let $\Phi'=(\Vcal',\Ccal')$ be the maximal connected component in $\Phi^\sigma$ intersecting $v$.
Then by \Cref{cor:rejectionsampling_efficiency_marking}, the expected runtime of \RejectionSampling{$\sigma,v$} is bounded by
$$
\tilde O\pbra{|\Vcal'|\cdot\exp\cbra{\frac{\xi\cdot|\Ccal'|}{k^6(\alpha+1)}}}
=\tilde O\pbra{|\Ccal'|\cdot\exp\cbra{\frac{\xi\cdot|\Ccal'|}{k^6(\alpha+1)}}},
$$
where we use the fact that $|\Vcal'|\le k\cdot|\Ccal'|=\tilde O(|\Ccal'|)$.
Thus by \Cref{lem:bernoulli_factory} and the analysis of \Cref{lem:correctness_of_marginoverflow}, the expected runtime of \MarginOverflow{$\sigma,v$} is
$$
\tilde O\pbra{\frac{|\Ccal'|}{\xi^2}\cdot\exp\cbra{\frac{\xi\cdot|\Ccal'|}{k^6(\alpha+1)}}}.
$$
Now it suffices to show $\Ccal'\subseteq\Ccalcon^\sigma$.

Note that $\Ccal'$ can be constructed as follows:
Starting with $\Ccal'=\emptyset$, we repeatedly put $C\in\Ccal$ into $\Ccal'$ if $\Ccal(\sigma)\neq\True$ and, either (1) $v\in\vbl(C)$ or (2) $\vbl(C)\cap\vbl(C')\cap\Lambda(\sigma)\neq\emptyset$ for some $C'\in\Ccal'$.
Assume towards contradiction that $C$ is the first clause included in $\Ccal'$ but not in $\Ccalcon^\sigma$.
\begin{itemize}
\item If $C$ satisfies condition (1), we know $C\in\Ccalcon^\sigma(v)\subseteq\Ccalcon^\sigma$ since $\sigma(v)=\Qmark$. A contradiction.
\item Otherwise, $C$ satisfies condition (2).
Since $C'$ is included in both $\Ccal'$ and $\Ccalcon^\sigma$, there exists some $v'$ such that $\sigma(v')=\Qmark$ and $C'\in\Ccalcon^\sigma(v')\subseteq\Ccalcon^\sigma$.
Then we have the following cases:
\begin{itemize}
\item If $C'\in\Ccalfrozen^\sigma\cup\Ccalbad^\sigma\cup\Ccalsep$, then $C'\in\Ccalint^\sigma(v')$. Thus $C\in\Ccalcon^\sigma(v')\subseteq\Ccalcon^\sigma$. A contradiction.
\item If $C'\in\CcalQmark^\sigma$, then there exists some $v''\in\vbl(C')$ such that $\sigma(v'')=\Qmark$. Then $C\in\Ccalcon^\sigma(v'')\subseteq\Ccalcon^\sigma$. A contradiction.
\item Otherwise, $C'\notin\Ccalfrozen^\sigma\cup\Ccalbad^\sigma\cup\Ccalsep\cup\CcalQmark^\sigma$.
Since $\NextVar(\sigma)=\bot$ and $C'\in\Ccalcon^\sigma$, we have $\vbl(C')\cap\Vcalalive^\sigma=\emptyset$.
Note that $\vbl(C')$ has no $\Qmark$. 
Thus for any $u\in\vbl(C')\cap\Lambda(\sigma)\setminus\Vcalsep$, there exists some $C''\in\Ccal\setminus\Ccalsep$ such that $C''(\sigma)\neq\True$ and 
$$
|\vbl(C'')\cap\Lambda(\sigma)\setminus(\Vcalsep\cup\cbra{u})|<(2/3-2\eta)k.
$$
These $C''$'s satisfy $|\vbl(C'')\cap\Lambda(\sigma)\setminus\Vcalsep|<1+(2/3-2\eta)k$ and are thus in $\Ccalfrozen^\sigma$.
This, together with $C'(\sigma)\neq\True$ and $\Ccal'\notin\Ccalfrozen^\sigma\cup\Ccalsep$, implies that $C'\in\Ccalbad^\sigma$. A contradiction.
\qedhere
\end{itemize}
\end{itemize}
\end{proof}

Similar to \Cref{lem:bf_efficiency_given_ccon} and by \Cref{cor:rejectionsampling_efficiency_marking}, the efficiency of the final rejection sampling boils down to the size of the remaining components in $\Phi^\sigma$ where $\sigma$ is the partial assignment on \Cref{ln:sol_3} of \SolutionSampling{$\Phi$}.
Guided by these intuition, we will keep track of the size of $\Ccalcon^\sigma$ and truncate the program if it gets too large.
In addition, we will halt the program if some component in $\Phi^\sigma$ is large upon the final rejection sampling.

Let $s\ge1$ be the truncation parameter to be optimized later.
We formalize our actual algorithms in \Cref{alg:the_actual_algorithm} and highlight the place where truncation happens.
For convenience, we overload \SolutionSampling{}, \MarginSample{}, and \MarginOverflow{} with the addition parameter $s$, and they reduce to the original version if $s=+\infty$.

\begin{algorithm2e}[ht]
\caption{The Actual Algorithms}\label{alg:the_actual_algorithm}
\DontPrintSemicolon
\Procedure{\SolutionSampling{$\Phi,s$}}{
\lnl{ln:actual_sol_1} Obtain $\Vcalsep,\Ccalsep\gets\ConstructSep{$\Vcal$}$\;
\lnl{ln:actual_sol_2} Initialize $\sigma\gets\EQmark^\Vcal$\;
\lnl{ln:actual_sol_3} \ForEach{$i=1$ \KwTo $n$}{
\lnl{ln:actual_sol_4} \lIf{$v_i\in\Vcalalive^\sigma$}{
Update $\sigma(v_i)\gets\MarginSample{$\sigma,v_i$}$
}
}
\lnl{ln:actual_sol_5} \lIf(\tcc*[f]{Truncation}){some connected component in $\Phi^\sigma$ has $>s$ clauses}{\textbf{Halt}}
\lnl{ln:actual_sol_6} $\sigma\gets\RejectionSampling{$\sigma,\Lambda(\sigma)$}$\;
\nl \Return{$\sigma$}
}
\setcounter{AlgoLine}{0}
\Procedure{\MarginSample{$\sigma,v,s$}}{
\nl Sample $\sigma(v)\sim\tau$\;
\nl \lIf{$\sigma(v)=\Qmark$}{\Return{\MarginOverflow{$\sigma,v,s$}}}
\nl \lElse{\Return{$\sigma(v)$}}
}
\setcounter{AlgoLine}{0}
\Procedure{\MarginOverflow{$\sigma,v,s$}}{
\lnl{ln:actual_mo_1} \lIf(\tcc*[f]{Truncation}){$|\Ccalcon^\sigma|>s$}{\textbf{Halt}}
\nl Let $u\gets\NextVar(\sigma)$\;
\nl \eIf{$u\neq\bot$}{
\nl Sample $\sigma(u)\sim\tau$\;
\nl \lIf{$\sigma(u)=\Qmark$}{Update $\sigma(u)\gets\MarginOverflow{$\sigma,u,s$}$}
\nl \Return{\MarginOverflow{$\sigma,v,s$}}
}{
\lnl{ln:actual_mo_7} \Return{\BernoulliFactory{$b_1,b_2,\ldots$}} where $b_1,b_2,\ldots$ are independent samples provided by executing \RejectionSampling{$\sigma,v$}
}
}
\end{algorithm2e}

Similarly as \Cref{fct:nextvar_efficiency}, checking components' sizes can be done efficiency.

\begin{fact}\label{fct:truncate_check_efficiency}
With $\tilde O(n)$ pre-processing time, the runtime of \Cref{ln:actual_sol_5} of \SolutionSampling{$\Phi,s$} and \Cref{ln:actual_mo_1} of \MarginOverflow{$\sigma,v,s$} is $\tilde O(1)$.
\end{fact}

An immediate corollary of \Cref{lem:bf_efficiency_given_ccon} is the following efficiency guarantee for the leaf recursion of \MarginOverflow{$\sigma,v,s$}.

\begin{corollary}\label{cor:bf_efficiency_given_ccon}
If $\sigma$ satisfies \Cref{as:marking} and $\NextVar(\sigma)=\bot$, then the expected runtime of \MarginOverflow{$\sigma,v,s$} is
$$
\tilde O\pbra{\frac{s}{\xi^2}\cdot\exp\cbra{\frac{\xi\cdot s}{k^6(\alpha+1)}}}.
$$
\end{corollary}

The runtime of the final rejection sampling is also controlled by the truncation parameter and \Cref{cor:rejectionsampling_efficiency_marking}.

\begin{corollary}\label{lem:final_rej_efficiency}
Assume $\sigma$ satisfies \Cref{as:marking}.
Then \RejectionSampling{$\sigma,\Lambda(\sigma)$} on \Cref{ln:actual_sol_6} of \SolutionSampling{$\Phi,s$} runs in expected time
$$
\tilde O\pbra{n\cdot\exp\cbra{\frac{\xi\cdot s}{k^6(\alpha+1)}}},
$$
where each $\Ccal_i$ is from \Cref{ln:rej_1} of \RejectionSampling{$\sigma,\Lambda(\sigma)$}.
\end{corollary}

In addition, since the difference only comes from the truncation, \Cref{cor:correctness_of_main_algorithm} allows us to bound the distance between algorithm's output and a uniform solution of $\Phi$ in terms of the probability that the program halts (i.e., truncation happens).

\begin{corollary}\label{cor:correctness_of_main_algorithm_trunc}
\SolutionSampling{$\Phi,s$} terminates almost surely and has output distribution $p_\textsf{halt}(\Phi,s)$-close to $\mu$ in the total variation distance, where
$$
p_\textsf{halt}(\Phi,s)=\Pr\sbra{\text{truncation happens during the algorithm}}.
$$
\end{corollary}

\subsection{The Recursive Cost Tree and the Simulation Tree}\label{sec:the_rct_and_sim}

Now we turn to the most technical part: The analysis of the efficiency and $p_\textsf{halt}(\Phi,s)$.

To this end, we use the notion of the recursive cost tree and the simulation tree similar to \cite{he2022sampling}.
The former captures the execution of a single \MarginOverflow{$\sigma,v,s$}, and the latter represents the whole execution of \SolutionSampling{$\Phi,s$}.

\begin{definition}[Recursive Cost Tree]\label{def:rct}
Let $\sigma$ be a partial assignment satisfying \Cref{as:marking}.
We define the \emph{recursive cost tree} for $\sigma$ as $\Tcal_\sigma$.
Here $\Tcal_\sigma$ is a rooted tree with nodes labeled by distinct\footnote{The nodes are distinct by the definition, where the partial assignments of the child nodes of $\pi$ fix the value of $u=\NextVar(\pi)$ from $\EQmark$ to $0$/$1$/$\Qmark$.} partial assignments $\pi$ and edges labeled by values $\rho$ in $[0,1]$ as follows:
\begin{itemize}
\item The root of $\Tcal_\sigma$ is $\sigma$ and its depth is defined to be $0$.
\item For $i=0,1,\ldots$, let $\pi\in\Tcal_\sigma$ be a node of depth $i$.

If $|\Ccalcon^\pi|>s$, then we leave $\pi$ as a \emph{recursing truncated} leaf node.

Otherwise, let $u=\NextVar(\pi)$ and we proceed as follows:
\begin{itemize}
\item If $u=\bot$, then we leave $\pi$ as a \emph{Bernoulli} leaf node.
\item Otherwise, let $\pi_0,\pi_1,\pi_\Qmark$ equal $\pi$ except that we fix $u$ to $0,1,\Qmark$ respectively. 
Then we append $\pi_0,\pi_1,\pi_\Qmark$ as the child nodes of $\pi$ and label the edges by
$$
\rho(\pi\to\pi_0)=\mu_u^{\pi}(0),
\quad\rho(\pi\to\pi_1)=\mu_u^{\pi}(1),
\quad\rho(\pi\to\pi_\Qmark)=\delta.
$$
\end{itemize}
\end{itemize}
\end{definition}

The edge values reflect the \MarginOverflow{$\sigma,v,+\infty$} recursion without truncation, and it is an overestimate for the \MarginOverflow{$\sigma,v,s$}.
In addition, $\Tcal_\sigma$ stops either at a Bernoulli leaf node, which corresponds to a regular leaf recursion and is ready for Bernoulli factory on \Cref{ln:actual_mo_7}, or at a recursing truncated leaf node, which corresponds to a truncation on \Cref{ln:actual_mo_1} during the recursion.

We remark that the edge value only depends on the partial assignments of the endpoints. This is why we can use a single symbol $\rho$ without confusion.

\begin{remark}\label{rmk:rct}
Let $\sigma$ be a partial assignment satisfying \Cref{as:marking} where $\sigma(v)=\Qmark$ and the rest values are $0$/$1$/$\EQmark$.
We show a one-to-one correspondence between nodes in $\Tcal_\sigma$ and the execution of \MarginOverflow{$\sigma,v,s$}.

The starting point \MarginOverflow{$\sigma,v,s$} corresponds to the root of $\Tcal_\sigma$.
Recall the algorithm description from \Cref{alg:the_actual_algorithm}.
Assume we just enter \MarginOverflow{$\pi,w,s$}, which by induction corresponds to the node $\pi\in\Tcal_\sigma$.
Then after checking if $|\Ccalcon^\pi|>s$ (i.e., if $\pi$ is a recursing truncated leaf node) on \Cref{ln:actual_mo_1}, we will compute $u=\NextVar(\pi)$ and perform the Bernoulli factory if $u=\bot$, i.e., $\pi$ is a Bernoulli leaf node as designed.
If $|\Ccalcon^\pi|\le s$ and $u\neq\bot$, the algorithm will update the assignment of $u$.
Then we have $\tau(\Qmark)=\delta$ probability of executing \MarginOverflow{$\pi_\Qmark,u,s$} which means in $\Tcal_\sigma$ proceeding to the child node $\pi_\Qmark$.
Afterwards, we will run \MarginOverflow{$\pi_b,w,s$} for $b\in\bin$, corresponding to the child node $\pi_b\in\Tcal_\sigma$.

The definition of $\rho(\pi\to\pi_\Qmark)$ is already explained above.
For $b\in\bin$, the probability of visiting $\pi_b$ is upper bounded by the corresponding probability with no truncation, i.e., setting $s=+\infty$ in \MarginOverflow{$\pi_\Qmark,u,s$}, which, by \Cref{lem:correctness_of_marginoverflow}, is exactly $\mu_u^\pi(b)=\rho(\pi\to\pi_b)$.
\end{remark}

To study the runtime of the whole \SolutionSampling{$\Phi,s$}, we define the following simulation tree on top of recursive cost trees.

\begin{definition}[Simulation Tree]\label{def:sim}
We define the \emph{simulation tree} as $\Tcalsim$.
Here $\Tcalsim$ is a rooted tree with nodes labeled by distinct\footnote{This is also clear from the definition of the simulation tree and expanding the construction of the recursive cost tree. In general, the partial assignments of the child nodes fix the variable from $\EQmark$ to $0$/$1$/$\Qmark$.} partial assignments $\pi$ and edges labeled by values $\rho$ in $[0,1]$ as follows:
\begin{itemize}
\item The root of $\Tcalsim$ is $\EQmark^\Vcal$ and its depth is defined to be $0$.
\item For $i=0,1,\ldots$, let $\pi\in\Tcalsim$ be a node of depth $i$.
\begin{itemize}
\item If $\pi$ has a $\Qmark$, then we say $\pi$ is a \emph{recursing} node and we append $\Tcal_\pi$ here.
\item Otherwise, let $u$ be the variable in $\Vcalalive^\pi$ with minimal index:
\begin{itemize}
\item If $u$ does not exist and each connected component in $\Phi^\pi$ has at most $s$ clauses, then we leave $\pi$ as a \emph{sampling} leaf node.
\item If $u$ does not exist and some connected component in $\Phi^\pi$ has $>s$ clauses, then we leave $\pi$ as a \emph{sampling truncated} leaf node.
\item Otherwise $u$ exists. Let $\pi_0,\pi_1,\pi_\Qmark$ equal $\pi$ except that we fix $u$ to $0,1,\Qmark$ respectively. 
Then we append $\pi_0,\pi_1,\pi_\Qmark$ as the child nodes of $\pi$ and label the edges by
$$
\rho(\pi\to\pi_0)=\mu_u^\pi(0),
\quad
\rho(\pi\to\pi_1)=\mu_u^\pi(1),
\quad
\rho(\pi\to\pi_\Qmark)=\delta.
$$
\end{itemize}
\end{itemize}
\end{itemize}
\end{definition}

Intuitively corresponding to \SolutionSampling{$\Phi,s$}, a recursing node means that we are about to do \MarginOverflow{} inside a \MarginSample{} on \Cref{ln:actual_sol_4}, a sampling leaf node means that we now perform the final rejection sampling on \Cref{ln:actual_sol_6}, and a sampling truncated leaf node is analogous to the one in recursive cost tree that truncation happens on \Cref{ln:actual_sol_5}.

We remark that the edge value $\rho$ is indeed consistent with the one in the definition of the recursive cost tree, as both of them refer to the probability of the one-step update of the partial assignments of the endpoints of the edge without truncation. Thus we use the same symbol.

\begin{remark}\label{rmk:sim}
Similar to the recursive cost tree, there is a one-to-one correspondence between nodes in $\Tcalsim$ and the execution of \SolutionSampling{$\Phi,s$}.
Let $\sigma$ be the partial assignment that \SolutionSampling{$\Phi,s$} maintains.

At the beginning, $\sigma=\EQmark^\Vcal$ and it is the root of $\Tcalsim$.
Each time we update $\sigma(v_i)$ for $v_i\in\Vcalalive^\sigma$ on \Cref{ln:actual_sol_4}, this $v_i=u$ has the minimal index in $\Vcalalive^\sigma$ since the for-loop on \Cref{ln:actual_sol_3} goes in the ascending order.
Then, based on the outcome of $b\gets\MarginSample{$\sigma,v_i,s$}$, we update $\sigma$ to $\sigma_b,b\in\bin$.
Recall that \MarginSample{$\sigma,v_i,s$} may call \MarginOverflow{$\sigma_\Qmark,v_i,s$}.
Together with $\sigma_0,\sigma_1$, these $\sigma_b$'s are presented in $\Tcalsim$ as the child nodes of $\sigma$, where $\sigma_\Qmark$ is a recursing node and we append $\Tcal_{\sigma_\Qmark}$ and follow the correspondence in \Cref{rmk:rct}.

To see the edge values, the probability of obtaining $\sigma_\Qmark$ is precisely $\tau(\Qmark)=\delta=\rho(\sigma\to\sigma_\Qmark)$.
For $b\in\bin$, the probability of visiting $\sigma_b$ is upper bounded by the corresponding probability when we ignore truncation, which in turn is exactly $\mu_{v_i}^\sigma(b)=\rho(\sigma\to\sigma_b)$ by \Cref{cor:correctness_of_marginsample}.

Finally on \Cref{ln:actual_sol_5}, we reach a partial assignment $\sigma$ ready for the final rejection sampling \Cref{ln:actual_sol_6}.
Depending on the components' sizes in $\Phi^\sigma$, it gives a sampling (truncated) leaf node.
\end{remark}

By the correspondence above, we see that \Cref{as:marking} is always preserved.
\begin{fact}\label{fct:sim_marking}
\Cref{as:marking} holds for any node in $\Tcalsim$.
\end{fact}

For convenience, we define the following quantities:
\begin{itemize}
\item For a node $\pi$ in $\Tcalsim$, $\rho(\pi)$ denotes the product of the edge values from the root of $\Tcalsim$ to $\pi$.
\item $\Ncalrec$ denotes the set of recursing nodes of $\Tcalsim$, and define $\drec=\max_{\sigma\in\Ncalrec}\depth(\Tcal_\sigma)$ to be the maximal depth of the recursive cost trees encountered.
\item $\Ncalsamptrunc$ denotes the set of sampling truncated leaf nodes of $\Tcalsim$, corresponding to \Cref{ln:actual_sol_5} of \SolutionSampling{$\Phi,s$}.
\item $\Ncalrectrunc$ denotes the set of recursing truncated leaf nodes of $\Tcalsim$, corresponding to \Cref{ln:actual_mo_1} of \MarginOverflow{$\sigma,v,s$}.
\item $\Ncaltrunc=\Ncalsamptrunc\cup\Ncalrectrunc$ denotes the set of all truncated leaf nodes. 
\end{itemize}

At this point, we can bound $p_\textsf{halt}(\Phi,s)$ using the leaf nodes' information of $\Tcalsim$.
\begin{lemma}\label{lem:rct_sim_phalt}
$p_\mathsf{halt}(\Phi,s)\le\sum_{\pi\in\Ncaltrunc}\rho(\pi)$.
\end{lemma}
\begin{proof}
By \Cref{rmk:sim}, we have a one-to-one correspondence between the truncated leaf nodes in $\Tcalsim$ and the place where truncation happens during \SolutionSampling{$\Phi,s$}.
In addition, for any partial assignment $\pi$, $\rho(\pi)$ upper bounds the probability that the algorithm visits $\pi$.
Therefore by the definition of $\Ncaltrunc$ and $p_\mathsf{halt}(\Phi,s)$ from \Cref{cor:correctness_of_main_algorithm_trunc}, we have
\begin{equation*}
p_\mathsf{halt}(\Phi,s)=\Pr\sbra{\text{reaching some node in $\Ncaltrunc$ during the algorithm}}\le\sum_{\pi\in\Ncaltrunc}\rho(\pi).
\tag*{\qedhere}
\end{equation*}
\end{proof}

The runtime can also be analyzed similarly.
\begin{lemma}\label{lem:rct_sim_efficiency}
\SolutionSampling{$\Phi,s$} runs in expected time
$$
\tilde O\pbra{n\cdot(1+\delta)^\drec\cdot\frac s\xi\cdot\exp\cbra{\frac{\xi\cdot s}{k^6(\alpha+1)}}}.
$$
\end{lemma}
\begin{proof}
Recall the description of \SolutionSampling{$\Phi,s$} from \Cref{alg:the_actual_algorithm}.
The runtime of \Cref{ln:actual_sol_1,ln:actual_sol_2} is $\tilde O(n)$ by \Cref{fct:construct_sep_efficiency}.
\Cref{ln:actual_sol_5,ln:actual_sol_6} runs in expected time $\tilde O\pbra{n\cdot\exp\cbra{\frac{\xi\cdot s}{k^6(\alpha+1)}}}$ by \Cref{fct:truncate_check_efficiency} and \Cref{lem:final_rej_efficiency}.
Now it remains to bound the runtime of \Cref{ln:actual_sol_3,ln:actual_sol_4}.

Firstly checking condition on \Cref{ln:actual_sol_4} takes $\tilde O(n)$ time in total by \Cref{fct:vcalalive_efficiency}.
Now assume we call \MarginSample{$\sigma,v_i,s$} on \Cref{ln:actual_sol_4}.
Let $\sigma_\Qmark$ equal $\sigma$ except we fix $v_i$ to $\Qmark$.
By \Cref{rmk:sim}, $\sigma_\Qmark$ is a recursing node where we append the recursive cost tree $\Tcal_{\sigma_\Qmark}$ for \MarginOverflow{$\sigma_\Qmark,v_i,s$}.
The expected runtime of \MarginOverflow{$\sigma_\Qmark,v_i,s$} has two parts: 
\begin{enumerate}
\item[(i)] Visiting partial assignments $\pi$, checking $|\Ccalcon^\pi|$, calculating $\NextVar(\pi)$, and sampling from $\tau$.
\item[(ii)] Performing Bernoulli factory on leaf recursions if not truncated.
\end{enumerate}
Let $\rho'(\pi)$ be the product of the edge values from the root of $\Tcal_{\sigma_\Qmark}$ to $\pi$.
By the correspondence described in \Cref{rmk:rct}, the probability of visiting a partial assignment $\pi\in\Tcal_{\sigma_\Qmark}$ conditioned on starting at $\sigma_\Qmark$ is upper bounded by $\rho'(\pi)$.
Without loss of generality, we expand $\Tcal_{\sigma_\Qmark}$ to a complete ternary tree where the parent-to-child edge weights are $\zeta,1-\zeta,\delta$ respectively for some $\zeta\in[0,1]$.
This is consistent with the existing edge values $\rho$, where $\zeta=\mu_u^\pi(0)$ for node $\pi$ and $u=\NextVar(\pi)$.
At this point, we have
$$
\E[\text{runtime for (i)}]
\le\sum_{\pi\in\Tcal_{\sigma_\Qmark}}\rho'(\pi)\cdot\tilde O(1)
\le\sum_{d=0}^{\drec}(1+\delta)^d\cdot\tilde O(1)
=\tilde O\pbra{\frac{(1+\delta)^\drec}{\delta}}.
$$
By \Cref{cor:bf_efficiency_given_ccon}, we can bound the runtime of (ii) similarly
\begin{align*}
\E[\text{runtime for (ii)}]
&\le\sum_{\pi\in\Tcal_{\sigma_\Qmark}\text{ is a Bernoulli leaf node}}\rho'(\pi)\cdot\tilde O\pbra{\frac s{\xi^2}\cdot\exp\cbra{\frac{\xi\cdot s}{k^6(\alpha+1)}}}\\
&\le(1+\delta)^\drec\cdot\tilde O\pbra{\frac s{\xi^2}\cdot\exp\cbra{\frac{\xi\cdot s}{k^6(\alpha+1)}}},
\end{align*}
where we use the fact that Bernoulli factory happens only on leaf nodes.
Since we only have $\rho(\sigma\to\sigma_\Qmark)=\delta$ probability of executing \MarginOverflow{$\sigma_\Qmark,v_i,s$}, we have
\begin{align*}
\E[\text{runtime of \MarginSample{$\sigma,v_i$}}]
&=\tilde O(1)+\delta\cdot\pbra{\E[\text{runtime for (i)}]+\E[\text{runtime for (ii)}]}\\
&\le\tilde O\pbra{(1+\delta)^\drec\cdot\pbra{1+\frac{\delta s}{\xi^2}\cdot\exp\cbra{\frac{\xi\cdot s}{k^6(\alpha+1)}}}}\\
&\le\tilde O\pbra{(1+\delta)^\drec\cdot\pbra{1+\frac s\xi\cdot\exp\cbra{\frac{\xi\cdot s}{k^6(\alpha+1)}}}}\\
&=\tilde O\pbra{(1+\delta)^\drec\cdot\frac s\xi\cdot\exp\cbra{\frac{\xi\cdot s}{k^6(\alpha+1)}}},
\tag{since $s\ge1\ge\xi$}
\end{align*}
where we use $\alpha\ge1/k^3$ and $\delta=\xi/(k^{40}\alpha)\le\tilde O(\xi)$ in the third step.
Hence
$$
\E[\text{runtime of \Cref{ln:actual_sol_3,ln:actual_sol_4}}]
\le\tilde O\pbra{n\cdot(1+\delta)^\drec\cdot\frac s\xi\cdot\exp\cbra{\frac{\xi\cdot s}{k^6(\alpha+1)}}}.
$$

Putting everything together, we have
\begin{align*}
\E[\text{total runtime}]
&\le\tilde O\pbra{n+n\cdot(1+\delta)^\drec\cdot\frac s\xi\cdot\exp\cbra{\frac{\xi\cdot s}{k^6(\alpha+1)}}+n\cdot\exp\cbra{\frac{\xi\cdot s}{k^6(\alpha+1)}}}\\
&=\tilde O\pbra{n\cdot(1+\delta)^\drec\cdot\frac s\xi\cdot\exp\cbra{\frac{\xi\cdot s}{k^6(\alpha+1)}}}.
\tag*{\qedhere}
\end{align*}
\end{proof}
\section{Truncation Analysis}\label{sec:truncation_analysis}

Given \Cref{lem:rct_sim_phalt}, \Cref{cor:correctness_of_main_algorithm_trunc}, and \Cref{lem:rct_sim_efficiency}, we need to carefully select the truncation parameter $s$ such that both $\drec$ and $p_\mathsf{halt}(\Phi,s)$ can be bounded.
The goal of this section is to establish such relations and prove the following formal statements.
We will still assume $(\Phi,k,\alpha,n,\xi,\eta,D)$ is nice and omit it from all the statements.

\begin{lemma}\label{lem:drec_bound}
$\drec\le s\cdot k+1$.
\end{lemma}

\begin{lemma}\label{lem:phalt_s}
Assume $6k^4\alpha\log(n)<s\le n/2^{5k/\log(k)}$.
Then
$$
p_\mathsf{halt}(\Phi,s)\le n^{10}(1+\delta)^{\drec+1}\cdot k^{-s/(6k^4\alpha)}.
$$
\end{lemma}

Let $\sigma$ be a partial assignment.
For convenience, we recall the definitions:
\begin{itemize}
\item $v\in\Vcalalive^\sigma$ iff (i) $\sigma(v)=\EQmark$ and $v\notin\Vcalsep$, and (ii) for every clause $C\in\Ccal\setminus\Ccalsep$, either $C(\sigma)=\True$ or $|\vbl(C)\cap\Lambda(\sigma)\setminus\pbra{\Vcalsep\cup\cbra{v}}|\ge(2/3-2\eta)k$.
\item $C\in\CcalQmark^\sigma$ iff there exists some $v\in\vbl(C)$ that $\sigma(v)=\Qmark$.
\item $C\in\Ccalfrozen^\sigma$ iff (i) $C(\sigma)\neq\True$ and $C\notin\Ccalsep$, and (ii) $|\vbl(C)\cap\Lambda(\sigma)\setminus\Vcalsep|<1+(2/3-2\eta)k$.
\item $C\in\Ccalbad^\sigma$ iff (i) $C(\sigma)\neq\True$ and $C\notin\Ccalfrozen^\sigma\cup\Ccalsep$, and (ii) for any $v\in\vbl(C)\setminus\Vcalsep$ with $\sigma(v)=\EQmark$, there exists some $C'\in\Ccalfrozen^\sigma$ such that $v\in\vbl(C')$.
\item $C\in\Ccalint^\sigma(v)$ iff (i) $C\in\Ccalfrozen^\sigma\cup\Ccalbad^\sigma\cup\Ccalsep$, and (ii) either $v\in\vbl(C)$ or there exists some $C'\in\Ccalint^\sigma(v)$ that $\vbl(C')\cap\vbl(C)\cap\Lambda(\sigma)\neq\emptyset$.
\item $C\in\Ccalcon^\sigma(v)$ iff $C\in\Ccalint^\sigma(v)$, or $v\in\vbl(C)$, or there exists some $C'\in\Ccalint^\sigma(v)$ that $\vbl(C')\cap\vbl(C)\cap\Lambda(\sigma)\neq\emptyset$.
\item $\Ccalcon^\sigma$ is the union of $\Ccalcon^\sigma(v)$ for all $v$ with $\sigma(v)=\Qmark$.
\end{itemize}

We will prove \Cref{lem:drec_bound} in \Cref{sec:size2depth}.
Then we construct witnesses for truncated nodes in \Cref{sec:witness_for_trunc} and prove \Cref{lem:phalt_s} in \Cref{sec:refutation_of_large_component}.

\subsection{Size-to-Depth Reduction}\label{sec:size2depth}

We start by relating $s$ and $\drec$, and show that small truncation size implies small depth in the recursive cost trees.
To this end, we will use $|\Ccalcon^\pi|$ as an intermediate measure for partial assignments $\pi$ in recursive cost trees.
Indeed, $|\Ccalcon^\pi|$ is upper bounded by $s$ by truncation, and we only need to lower bound it in terms of $\drec$.

We start by proving the connectivity, which reduces to the following technical lemma showing that $\Ccalint^\pi$ and $\Ccalcon^\pi$ are increasing in $\pi$. 

\begin{lemma}\label{lem:cint_ccon_increasing}
Let $\pi$ and $\pi'$ be partial assignments.
Assume $\pi'$ extends $\pi$ by fixing some variable in $\Vcalalive^\pi$.
Then $\Ccalint^\pi(v)\subseteq\Ccalint^{\pi'}(v)$ and $\Ccalcon^\pi(v)\subseteq\Ccalcon^{\pi'}(v)$ hold for any $v$ with $\pi(v)=\Qmark$.
\end{lemma}
\begin{proof}
We first show $\Ccalint^\pi(v)\subseteq\Ccalint^{\pi'}(v)$.
Let $C\in\Ccalint^\pi(v)$ and we verify the conditions for $C\in\Ccalint^{\pi'}(v)$:
\begin{itemize}
\item Condition (i). 
By the condition (i) for $C\in\Ccalint^\pi(v)$, we have $C\in\Ccalfrozen^\pi\cup\Ccalbad^\pi\cup\Ccalsep$.
Then $C\in\Ccalfrozen^{\pi'}\cup\Ccalbad^{\pi'}\cup\Ccalsep$ since $\Ccalfrozen^\pi\subseteq\Ccalfrozen^{\pi'},\Ccalbad^\pi\subseteq\Ccalbad^{\pi'}$ by \Cref{fct:alive_not_frozen_bad_sep}.
\item Condition (ii). We have two cases based on the condition (ii) for $C\in\Ccalint^\pi(v)$:
\begin{itemize}
\item If $v\in\vbl(C)$, then the same reason holds for $C\in\Ccalint^{\pi'}(v)$.
\item Otherwise, there exists some $C'\in\Ccalint^\pi(v)$ that $\vbl(C')\cap\vbl(C)\cap\Lambda(\pi)\neq\emptyset$.
Now note that $\pi'$ extends $\pi$ on a variable in $\Vcalalive^\pi$, which, by condition (i) and \Cref{fct:alive_not_frozen_bad_sep}, is not contained in $C'$.
Thus $\vbl(C')\cap\vbl(C)\cap\Lambda(\pi')=\vbl(C')\cap\vbl(C)\cap\Lambda(\pi)\neq\emptyset$, which means the condition (ii) here holds due to the same $C'$.
\end{itemize}
\end{itemize}
Now we prove $\Ccalcon^\pi(v)\subseteq\Ccalcon^{\pi'}(v)$ with similar arguments.
Let $C\in\Ccalcon^\pi(v)$ and we verify $C\in\Ccalcon^{\pi'}(v)$:
\begin{itemize}
\item If $C\in\Ccalint^\pi(v)$, then $C\in\Ccalint^{\pi'}(v)\subseteq\Ccalcon^{\pi'}(v)$ since $\Ccalint^\pi(v)\subseteq\Ccalint^{\pi'}(v)$.
\item If $v\in\vbl(C)$, then $C\in\Ccalcon^{\pi'}(v)$ by the same reason.
\item Otherwise, there exists some $C'\in\Ccalint^\pi(v)$ that $\vbl(C')\cap\vbl(C)\cap\Lambda(\pi)\neq\emptyset$.
Note that we have $C'\in\Ccalint^{\pi'}(v)$ since $\Ccalint^\pi(v)\subseteq\Ccalint^{\pi'}(v)$. 
As $\pi'$ differs from $\pi$ on a variable in $\Vcalalive^\pi$, by \Cref{fct:alive_not_frozen_bad_sep}, this variable is not in $C'\in\Ccalcon^\pi(v)\subseteq\Ccalfrozen^\pi\cup\Ccalbad^\pi\cup\Ccalsep$.
Thus $\vbl(C')\cap\vbl(C)\cap\Lambda(\pi')=\vbl(C')\cap\vbl(C)\cap\Lambda(\pi)\neq\emptyset$, which, combined with $C'\in\Ccalint^{\pi'}(v)$, implies $C\in\Ccalcon^{\pi'}(v)$.
\qedhere
\end{itemize}
\end{proof}

As a result, we can lower bound $|\Ccalcon^\pi|$ by the depth of $\pi$ in $\Tcal_\sigma$.

\begin{corollary}\label{cor:ccon_bounds}
Let $\sigma\in\Ncalrec$ and $\pi\in\Tcal_\sigma$.
Then $|\Ccalcon^\pi|\ge\depth(\pi,\Tcal_\sigma)/k$ where $\depth(\pi,\Tcal_\sigma)$ is the depth of $\pi$ in $\Tcal_\sigma$.
\end{corollary}
\begin{proof}
Let $L=\depth(\pi,\Tcal_\sigma)$. Along the path from $\sigma$ to $\pi$, we fix $L$ distinct variables $v_1,v_2,\ldots,v_L$.
Let partial assignments $\pi_0,\pi_1,\ldots,\pi_L$ be the evolution of the process, i.e., $\pi_0=\sigma$, $\pi_L=\pi$, and each $\pi_i$ extends $\pi_{i-1}$ by fixing $v_i$.
Since $v_i=\NextVar(\pi_{i-1})\in\Vcalcon^{\pi_{i-1}}$, there exists $C_i\in\Ccalcon^{\pi_{i-1}}$ such that $v_i\in\vbl(C_i)$.
By \Cref{lem:cint_ccon_increasing}, these $C_i$'s remain in $\Ccalcon^\pi$.
Thus $|\Ccalcon^\pi|$ is at least the number of distinct clauses in $C_1,C_2,\ldots,C_L$, which, in turn, is at least $L/k$.
\end{proof}

Now \Cref{lem:drec_bound} follows immediately.
\begin{proof}[Proof of \Cref{lem:drec_bound}]
Recall that $\drec=\max_{\sigma\in\Ncalrec}\depth(\Tcal_\sigma)=\max_{\sigma\in\Ncalrec,\pi\in\Tcal_\sigma}\depth(\pi,\Tcal_\sigma)$.
Let $\sigma$ and $\pi$ achieve $\depth(\pi,\Tcal_\sigma)=\drec$.
If $\pi$ is a Bernoulli leaf node, then $|\Ccalcon^\pi|\le s$.
By \Cref{cor:ccon_bounds}, we have $|\Ccalcon^\pi|\ge\drec/k$ and thus $\drec\le s\cdot k$.
Now assume $\pi\in\Ncalrectrunc$ is a recursing truncated leaf node. 

If $\pi=\sigma$, then $\drec=0$ trivially.
Otherwise let $\pi'$ be the parent node of $\pi$.
Then $\depth(\pi',\Tcal_\sigma)=\drec-1$ and we have $|\Ccalcon^{\pi'}|\ge(\drec-1)/k$ by \Cref{cor:ccon_bounds}.
Since $\pi'$ is not truncated, we also have $|\Ccalcon^{\pi'}|\le s$, which implies $\drec\le s\cdot k+1$.
\end{proof}

\subsection{Witness for Truncation}\label{sec:witness_for_trunc}

To establish \Cref{lem:phalt_s}, we will construct succinct witnesses for truncated nodes $\Ncaltrunc$.
Then in \Cref{sec:refutation_of_large_component}, we will enumerate all possible witnesses and apply a union bound to show that with high probability none of them appears.
Though we have two types of truncation $\Ncalrectrunc$ and $\Ncalsamptrunc$, the witness construction is similar.

Let $\pi$ be a partial assignment triggering truncation.
In a nutshell, the witness will consist of many connected clauses where most of the clauses are either frozen (i.e., in $\Ccalfrozen^\pi$) or contains $\Qmark$ (i.e., in $\CcalQmark^\pi$).
The former case, together with the locally sparse properties of $\Phi$, indicates that many variables in $\pi$ are fixed towards the bad direction that does not satisfy the clauses.
The latter case should also be rare since, by local uniformity, each $\Qmark$ appears with probability $\delta\ll1$.

\subsubsection*{Truncation inside the Margin Overflow}

We start with the recursing truncated nodes $\pi\in\Ncalrectrunc$, which corresponds to partial assignments $\sigma\in\Ncalrec$ and $\pi\in\Tcal_\sigma\cap\Ncaltrunc$.
We want to zoom in to the frozen clauses $\Ccalfrozen^\pi$ of $\pi$ since each clause there is still not satisfied and yet many variables within are fixed (to the unsatisfying direction).

However, $\Ccalfrozen^\pi$ alone may not be connected and we cannot afford the enumeration.
Therefore, we put in $\Ccalbad^\pi$ and $\Ccalsep$. These clauses do not contain many fixed variables (indeed by definition), but at least they are also not satisfied and is close to clauses in $\Ccalfrozen^\pi$. Thus we still have control for the variables within.

Unfortunately, at this point we still cannot guarantee large connected components. At best, we will only have connected components $\Ccalint^\pi(v_i)$ where $v_1,v_2,\ldots,v_t$ are the $\Qmark$'s in $\pi$; and these are not necessarily connected to each other.
Indeed, the definition of $\Ccalcon^\pi$ involves taking one step further from $\Ccalint^\pi(v_i)$; only after that it will be truncated due to exceeding size $s$.

The final thing we can do is to incorporate clauses in $\CcalQmark^\pi$, which is still acceptable since we have control for the probability that we encounter any fixed $\Qmark$.
This is indeed the case here: By \Cref{lem:cint_ccon_increasing}, we connect $\Ccalint^\pi(v_i)$'s by including edges from $\CcalQmark^\pi$.
Put differently, each $v_i$ is contained in $\CcalQmark^\pi(v_j)$ for some previous $v_j$.

\begin{lemma}\label{lem:generating_qmarks}
Let $\sigma\in\Ncalrec$ and $\pi\in\Tcal_\sigma$.
Let $v_1,v_2,\ldots,v_t$ be the $\Qmark$'s in $\pi$ in the order of the path from $\sigma$ to $\pi$.\footnote{That is, $v_1$ is the unique $\Qmark$ in $\sigma$ and $v_t$ is the last variable fixed to $\Qmark$ before reaching $\pi$.}
Then for any $i\ge2$, there exists some $j<i$ and $C\in\Ccalcon^\pi(v_j)$ such that $v_i\in\vbl(C)$.
\end{lemma}
\begin{proof}
Let $\pi'\in\Tcal_\sigma$ be the ancestor of $\pi$ that fixes $v_i$ to $\Qmark$.
That is, $\NextVar(\pi')=v_i$.
By the definition of $\NextVar()$, there exists $C\in\Ccalcon^{\pi'}$ that $v_i\in\vbl(C)$.
Since the $\Qmark$'s in $\pi'$ are $v_1,\ldots,v_{i-1}$, we have $\Ccalcon^{\pi'}=\bigcup_{j<i}\Ccalcon^{\pi'}(v_j)$.
Therefore, there exists $j<i$ such that $C\in\Ccalcon^{\pi'}(v_j)$.
Now by \Cref{lem:cint_ccon_increasing}, we know $C$ remains in $\Ccalcon^\pi(v_j)$ as desired, since $\pi$ is obtained from $\pi'$ by repeatedly fixing alive variables.
\end{proof}

As a result, we can connect $\Ccalint^\pi(v_i)$'s efficiently in a spanning tree fashion using edges in $\CcalQmark^\pi$.

\begin{corollary}\label{cor:connect_qmark}
Let $\sigma\in\Ncalrec$ and $\pi\in\Tcal_\sigma$.
Let $v_1,v_2,\ldots,v_t$ be the $\Qmark$'s in $\pi$ in the order of the path from $\sigma$ to $\pi$.
Then there exists $\CcalQmarkint^\pi\subseteq\CcalQmark^\pi$ disjoint from $\bigcup_i\Ccalint^\pi(v_i)$ such that the following holds:
\begin{itemize}
\item $\CcalQmarkint^\pi\cup\bigcup_i\Ccalint^\pi(v_i)$ is connected in $G_\Phi$, and $\CcalQmarkint^\pi$ covers $v_1,\ldots,v_t$.
\item For any $\Ccal'\subseteq\CcalQmarkint^\pi$, we have $|\cbra{v\in\Vcal'\mid\pi(v)=\Qmark}|\ge|\Ccal'|$ where $\Vcal'=\bigcup_{C\in\Ccal'}\vbl(C)$.
\end{itemize}
\end{corollary}
\begin{proof}
We construct $\CcalQmarkint^\pi$ by inspecting $\Ccalint^\pi(v_i)$ sequentially.
At first, $\CcalQmarkint^\pi=\emptyset$ and $i=1$.

By definition, each $\Ccalint^\pi(v_i)$ is connected by itself.
If at some point $i>1$, $\CcalQmarkint^\pi\cup\bigcup_{j\le i}\Ccalint^\pi(v_j)$ is not connected.
By \Cref{lem:generating_qmarks}, there exists $j<i$ and $C\in\Ccalcon^\pi(v_j)$ such that $v_i\in\vbl(C)$ and we put this $C$ into $\CcalQmarkint^\pi$.
Note that $C\in\CcalQmark^\pi$ since $v_i\in\vbl(C)$.
On the other hand, before including $C$, $v_i$ is not covered in $\CcalQmarkint^\pi$ since otherwise $\CcalQmarkint^\pi\cup\bigcup_{j\le i}\Ccalint^\pi(v_j)$ is already connected.

Therefore, after this process, $\Ccalint^\pi(v_i)$'s are connected by $\CcalQmarkint^\pi$ which covers all the $\Qmark$'s. In addition, the second item holds since each newly included clause brings in at least one $\Qmark$ distinct from all previous ones.
\end{proof}

Define $\Ccalint^\pi=\bigcup_i\Ccalint^\pi(v_i)$.
At this point, the witness is already in shape: $\CcalQmarkint^\pi\cup\Ccalint^\pi$.
Indeed, every clause in $\Ccalint^\pi$ is not satisfied by $\pi$, and $\CcalQmarkint^\pi$ contains all the $\Qmark$'s.
Now we need to show that the size of $\CcalQmarkint^\pi\cup\Ccalint^\pi$ scales with the size of $\Ccalcon^\pi$, which is in turn lower bounded by $s$ upon truncation.

\begin{lemma}\label{lem:first_witness_size_bound}
Let $\sigma\in\Ncalrec$ and $\pi\in\Tcal_\sigma$. 
Then
$$
6k^4\alpha\cdot\max\cbra{|\CcalQmarkint^\pi\cup\Ccalint^\pi|,\log(n)}\ge\min\cbra{|\Ccalcon^\pi|,n/2^{2k/\log(k)}}.
$$
Moreover, if $\pi\in\Ncaltrunc$ and $6k^4\alpha\log(n)<s\le n/2^{2k/\log(k)}$, then
$$
|\CcalQmarkint^\pi\cup\Ccalint^\pi|\ge\frac s{6k^4\alpha}.
$$
\end{lemma}
\begin{proof}
Let $\Vcal'=\bigcup_{C\in\CcalQmarkint^\pi\cup\Ccalint^\pi}\vbl(C)$.
Then $\Vcal'$ is connected in $H_\Phi$ since $\CcalQmarkint^\pi\cup\Ccalint^\pi$ is connected in $G_\Phi$ by \Cref{cor:connect_qmark}.
Thus by \Cref{prop:number_of_neighbors},
\begin{align*}
\Vcal'':=\cbra{v\mid v\in\Vcal'\text{ or $v$ is adjacent to $\Vcal'$}}
&\le3k^4\alpha\cdot\max\cbra{|\Vcal'|,\floorbra{k\log(n)}}\\
&\le3k^5\alpha\cdot\max\cbra{|\CcalQmarkint^\pi\cup\Ccalint^\pi|,\log(n)}.
\end{align*}
By the definition of $\Ccalcon^\pi$, every clause in $\Ccalcon^\pi$ is contained in $\Ccalint^\pi$, or is connected to some clause in $\Ccalint^\pi$, or contains some $\Qmark$ in $\pi$ and thus is connected to some clause $\CcalQmarkint^\pi$ by \Cref{cor:connect_qmark}. 
Therefore $\Vcal''$ is the support of $\Ccalcon^\pi$.

Now let $\Ccal''\subseteq\Ccalcon^\pi$ be arbitrary and has size $\min\cbra{|\Ccalcon^\pi|,n/2^{2k/\log(k)}}$.
Then $\Vcal''$ is also the support of $\Ccal''$. 
By \Cref{itm:kvars_and_distinct_vars_2} of \Cref{prop:kvars_and_distinct_vars} and $\eta\le1$, we have
$$
|\Vcal''|\ge\abs{\bigcup_{C\in\Ccal''}\vbl(C)}\ge|\Ccal''|\cdot k/2=\min\cbra{|\Ccalcon^\pi|,n/2^{2k/\log(k)}}\cdot k/2,
$$
which completes the proof for the first half.

For the second half, notice that $\pi\in\Ncaltrunc$ additionally implies $|\Ccalcon^\pi|>s$. Therefore
$$
6k^4\alpha\cdot\max\cbra{|\CcalQmarkint^\pi\cup\Ccalint^\pi|,\log(n)}\ge\min\cbra{s,n/2^{2k/\log(k)}}=s.
$$
Since $s>6k^4\alpha\log(n)$, we must take the former inside the $\max$, which gives the desired bound.
\end{proof}

Now that we have a relatively large witness.
The next step for us is to show that $\Ccalint^\pi$ contains many fixed variables in $\pi$ using the locally sparse properties of $\Phi$.

The caveat here is that, most of the structural properties in \Cref{sec:properties_of_random_cnf_formulas} hold only when we don't have \emph{too many} clauses, whereas it is possible that $\CcalQmarkint^\pi\cup\Ccalint^\pi$ exceeds this threshold.
Though $\CcalQmarkint^\pi\cup\Ccalint^\pi$ is a subset of $\Ccalcon^\pi$ and we truncate once $|\Ccalcon^\pi|>s$, it is not guaranteed that the size increase of $\Ccalcon^\pi$ is smooth that we have a reasonable \emph{upper} bound upon truncation.

To circumvent this issue, we introduce a pruning process on $\CcalQmarkint^\pi\cup\Ccalint^\pi$ to obtain the actual witness of size not to large while maintaining some key properties useful later.

\begin{lemma}\label{lem:prune_witness}
Let $\sigma\in\Ncalrec$ and $\pi\in\Tcal_\sigma$.
There exist $\barCcalQmarkint^\pi\subseteq\CcalQmarkint^\pi$ and $\barCcalint^\pi\subseteq\Ccalint^\pi$ such that the following holds:
\begin{enumerate}
\item\label{itm:lem:prune_witness_1}
If $|\CcalQmarkint^\pi\cup\Ccalint^\pi|\le n/2^{4k/\log(k)}$, then $\barCcalQmarkint^\pi=\CcalQmarkint^\pi$ and $\barCcalint^\pi=\Ccalint^\pi$.

Otherwise we have $n/2^{5k/\log(k)}\le|\barCcalQmarkint^\pi\cup\barCcalint^\pi|\le n/2^{4k/\log(k)}$.
\item\label{itm:lem:prune_witness_2}
$\barCcalQmarkint^\pi\cup\barCcalint^\pi$ is connected in $G_\Phi$, and $\barCcalQmarkint^\pi$ covers at least $|\barCcalQmarkint^\pi|$ many $\Qmark$'s.
\item\label{itm:lem:prune_witness_3}
For any $C\in\barCcalint^\pi\cap\Ccalbad^\pi$ and $v\in\vbl(C)$ with $\pi(v)=\EQmark$, there exists some $C'\in\barCcalint^\pi$ such that $v\in\vbl(C')$ and $C'\in\Ccalfrozen^\pi\cup\Ccalsep$.
\item\label{itm:lem:prune_witness_4}
For any $C\in\barCcalint^\pi\cap\Ccalbad^\pi$ and $v\in\vbl(C)$ with $\pi(v)=\Qmark$, there exists some $C'\in\barCcalQmarkint^\pi$ such that $v\in\vbl(C')$.
\end{enumerate}
\end{lemma}
\begin{proof}
We start with $\barCcalQmarkint^\pi=\CcalQmarkint^\pi,\barCcalint^\pi=\Ccalint^\pi$ then perform pruning iteratively.
At the beginning, \Cref{itm:lem:prune_witness_2,itm:lem:prune_witness_4} follow from \Cref{cor:connect_qmark}, and \Cref{itm:lem:prune_witness_3} holds due to the definition of $\Ccalint^\pi,\Ccalbad^\pi$ and \Cref{fct:sim_marking}.

If $|\barCcalQmarkint^\pi\cup\barCcalint^\pi|>n/2^{4k/\log(k)}$, then we have the following pruning cases:
\begin{itemize}
\item
If there exists $\bar C\in\barCcalint^\pi\cap\Ccalbad^\pi$, then let $\Scal_1,\Scal_2,\ldots,\Scal_t$ be the maximal connected components of $\barCcalQmarkint^\pi\cup\barCcalint^\pi$ in $G_\Phi$ after removing $\bar C$.
Assume that $\Scal_1$ has the maximal size. Then we update 
$$
\barCcalQmarkint^\pi\gets\Scal_1\cap\barCcalQmarkint^\pi 
\quad\text{and}\quad
\barCcalint^\pi\gets\Scal_1\cap\barCcalint^\pi.
$$
\item
Otherwise $\barCcalint^\pi\cap\Ccalbad^\pi=\emptyset$.
Then let $\bar C\in\barCcalQmarkint^\pi\cup\barCcalint^\pi$ be arbitrary such that removing it does not disconnect $\barCcalQmarkint^\pi\cup\barCcalint^\pi$ in $G_\Phi$, and we update
$$
\barCcalQmarkint^\pi\gets\barCcalQmarkint^\pi\setminus\cbra{\bar C}
\quad\text{and}\quad
\barCcalint^\pi\gets\barCcalQmarkint^\pi\setminus\cbra{\bar C}.
$$
\end{itemize}
Now we verify the conditions.
The connectivity is trivially preserved, and the number of $\Qmark$'s is always lower bounded by $|\barCcalQmarkint^\pi|$ due to $\barCcalQmarkint^\pi\subseteq\CcalQmarkint^\pi$ and \Cref{cor:connect_qmark}.
Thus \Cref{itm:lem:prune_witness_2} holds.

\Cref{itm:lem:prune_witness_3,itm:lem:prune_witness_4} is trivial for the second pruning case since $\barCcalint^\pi\cap\Ccalbad^\pi=\emptyset$ there. For the first pruning case, notice that $\Scal_1,\ldots,\Scal_t$ are disjoint from each other.
Upon the update, for any $C\in\Scal_1\cap\barCcalint^\pi\cap\Ccalbad^\pi$ and $v\in\vbl(C)$ with $\pi(v)=\EQmark$, its previous witness $C'\in\barCcalint^\pi\cap(\Ccalfrozen^\pi\cup\Ccalsep)$ is different from $\bar C$ and is connected to $C$ in $G_\Phi$.
Thus $C'\in\Scal_1$ comes along and \Cref{itm:lem:prune_witness_3} holds.
Similar argument holds for \Cref{itm:lem:prune_witness_4}.

Finally we prove \Cref{itm:lem:prune_witness_1} when the iterative pruning stops.
Note that each time we start with size larger than $n/2^{4k/\log(k)}$ and fall into one of the two pruning cases.
The second case decreases the size by one, and thus if we stop afterwards, we have $|\barCcalQmarkint^\pi\cup\barCcalint^\pi|>n/2^{4k/\log(k)}-1>n/2^{5k/\log(k)}$.
The first case removes a clause $\bar C$ and decompose the component into $t$ disjoint parts.
Since $\bar C$ contains at most $k$ literals and the $t$ parts are connected by $\bar C$, we know $t\le k$. Thus if we stop after this case, by averaging argument we have
\begin{equation*}
|\barCcalQmarkint^\pi\cup\barCcalint^\pi|\ge\frac1t\sum_{i=1}^t|\Scal_i|\ge\frac{n/2^{4k\log(k)/k}-1}k\ge n/2^{5k/\log(k)}.
\tag*{\qedhere}
\end{equation*}
\end{proof}

Define $\Wcal^\pi=\barCcalQmarkint^\pi\cup\barCcalint^\pi$ as our witness for $\pi$.
\Cref{itm:lem:prune_witness_3,itm:lem:prune_witness_4} of \Cref{lem:prune_witness} show that the unassigned variables in bad clauses of $\Wcal^\pi$ are also contained as frozen ones (i.e., in $\Wcal^\pi\cap\Ccalfrozen^\pi$), or as separators (i.e., in $\Wcal^\pi\cap\Ccalsep$), or by some clause in $\barCcalQmarkint^\pi\subseteq\Wcal^\pi$.
This allows us to leverage the structural properties of $\Phi$ to show the following ``saturation'' result, which intuitively says that most clauses in the witness are the frozen ones or contain $\Qmark$'s.

\begin{lemma}\label{lem:witness_saturation}
Let $\sigma\in\Ncalrec$ and $\pi\in\Tcal_\sigma$.
If $|\Wcal^\pi|\ge\log(n)$, then 
$$
|\barCcalQmarkint^\pi|+|\barCcalint^\pi\cap\Ccalfrozen^\pi|\ge(1-5\eta)\cdot|\Wcal^\pi|.
$$
\end{lemma}
\begin{proof}
Since $\barCcalint^\pi\subseteq\Ccalint^\pi$ only contains clauses from $\Ccalfrozen^\pi,\Ccalbad^\pi,\Ccalsep$ which are disjoint, we expand
\begin{equation}\label{eq:lem:witness_saturation_1}
|\barCcalint^\pi|=|\barCcalint^\pi\cap\Ccalfrozen^\pi|+|\barCcalint^\pi\cap\Ccalbad^\pi|+|\barCcalint^\pi\cap\Ccalsep|.
\end{equation}
By \Cref{cor:fraction_of_csep} and assuming $|\Wcal^\pi|\ge\log(n)$, we have
\begin{equation}\label{eq:lem:witness_saturation_2}
|\barCcalint^\pi\cap\Ccalsep|
\le|\Wcal^\pi\cap\Ccalsep|
\le\frac{1+\eta}k\cdot|\Wcal^\pi|
\le\eta|\Wcal^\pi|,
\end{equation}
where we use the fact that $\eta=15\log(k)/k\ge1/(k-1)$.

Define $\Ccal_1=\barCcalQmarkint^\pi\cup(\barCcalint^\pi\setminus\Ccalbad^\pi)$ and $\Ccal_2=\barCcalint^\pi\cap\Ccalbad^\pi$.
Let $\Vcal_1=\bigcup_{C\in\Ccal_1}\vbl(C)$ and $\Vcal_2=\bigcup_{C\in\Ccal_2}\vbl(C)$.
Then for any $v\in\Vcal_2$, we have the following cases:
\begin{itemize}
\item If $\pi(v)=\EQmark$, by \Cref{itm:lem:prune_witness_3} of \Cref{lem:prune_witness}, $v$ is covered in $\barCcalint^\pi\cap(\Ccalfrozen^\pi\cup\Ccalsep)$ and thus is in $\Vcal_1$.
\item If $\pi(v)=\Qmark$, by \Cref{itm:lem:prune_witness_4} of \Cref{lem:prune_witness}, $v$ is covered in $\barCcalQmarkint^\pi\subseteq\Ccal_1$ and thus is in $\Vcal_1$.
\item Otherwise $\pi(v)\in\cbra{0,1}$.
Since $C(\pi)\neq\True$ and $\pi$ satisfies \Cref{as:marking} by \Cref{fct:sim_marking}, the number of options for $v$ is
$$
|\vbl(C)\setminus\Lambda(\pi)|=|\vbl(C)|-|\vbl(C)\cap\Lambda(\pi)|
\le|\vbl(C)|-|\vbl(C)\cap\Lambda(\pi)\setminus\Vcalsep|
\le k-(2/3-2\eta)k.
$$
\end{itemize}
Thus
\begin{equation}\label{eq:lem:witness_saturation_3}
|\Vcal_1\cup\Vcal_2|=|\Vcal_1|+|\Vcal_2\setminus\Vcal_1|\le k|\Ccal_1|+(1/3+2\eta)k\cdot|\Ccal_2|.
\end{equation}
On the other hand, since $\barCcalQmarkint^\pi\cap\barCcalint^\pi=\emptyset$ and $\Wcal^\pi=\barCcalQmarkint^\pi\cup\barCcalint^\pi$, we have $\Ccal_1\cap\Ccal_2=\emptyset$ and $\Ccal_1\cup\Ccal_2=\Wcal^\pi$ of size at most $n/2^{4k/\log(k)}$ by \Cref{itm:lem:prune_witness_1} of \Cref{lem:prune_witness}.
Thus by \Cref{itm:kvars_and_distinct_vars_2} of \Cref{prop:kvars_and_distinct_vars}, we have
\begin{equation}\label{eq:lem:witness_saturation_4}
|\Vcal_1\cup\Vcal_2|\ge\frac{k|\Ccal_1\cup\Ccal_2|}{1+\eta}=\frac{k|\Ccal_1|+k|\Ccal_2|}{1+\eta}.
\end{equation}
Combining \Cref{eq:lem:witness_saturation_3} and \Cref{eq:lem:witness_saturation_4}, we have
\begin{equation}\label{eq:lem:witness_saturation_5}
|\barCcalint^\pi\cap\Ccalbad^\pi|=|\Ccal_2|\le\frac{\eta|\Ccal_1|}{\frac23-\frac{7\eta}3-2\eta^2}\le4\eta|\Ccal_1|\le4\eta|\Wcal^\pi|,
\end{equation}
where we use the fact that $\eta\le1/9$.
Finally we obtain
\begin{align*}
|\barCcalQmarkint^\pi|+|\barCcalint^\pi\cap\Ccalfrozen^\pi|
&=|\barCcalQmarkint^\pi|+|\barCcalint^\pi|-|\barCcalint^\pi\cap\Ccalbad^\pi|-|\barCcalint^\pi\cap\Ccalsep|
\tag{by \Cref{eq:lem:witness_saturation_1}}\\
&=|\Wcal^\pi|-|\barCcalint^\pi\cap\Ccalbad^\pi|-|\barCcalint^\pi\cap\Ccalsep|\\
&\ge|\Wcal^\pi|\cdot\pbra{1-5\eta}
\tag{by \Cref{eq:lem:witness_saturation_2} and \Cref{eq:lem:witness_saturation_5}}
\end{align*}
as desired.
\end{proof}

The clauses in $\barCcalQmarkint^\pi$ contribute at least $|\barCcalQmarkint^\pi|$ $\Qmark$'s.
To complement, we show that the clauses in $\barCcalint^\pi\cap\Ccalfrozen^\pi$ contribute almost maximal amount of variables that are fixed towards the unsatisfying direction.
Indeed, each clause can access at most $(1/3+2\eta)k$ variables due to \Cref{as:marking}, and clauses can overlap on many variables.
Nevertheless, we use the structural properties of $\Phi$ to show that the clauses in $\barCcalint^\pi\cap\Ccalfrozen^\pi$ will achieve this extremal ratio.

\begin{lemma}\label{lem:unsat_in_frozen_witness}
Let $\sigma\in\Ncalrec$ and $\pi\in\Tcal_\sigma$.
Let 
$$
\Vcal'=\bigcup_{C\in\barCcalint^\pi\cap\Ccalfrozen^\pi}\cbra{v\in\vbl(C)\mid\pi(v)\in\{0,1,\Qmark\}}
$$ 
be the set of accessed variables contained in $\barCcalint^\pi\cap\Ccalfrozen^\pi$.
Then $|\Vcal'|\ge|\barCcalint^\pi\cap\Ccalfrozen^\pi|\cdot(1-4\eta)k/3$.
\end{lemma}
\begin{proof}
The number of accessed variables in $C\in\Ccalfrozen^\pi$ is
\begin{align*}
|\vbl(C)\setminus\cbra{v\in\vbl(C)\mid\pi(v)=\EQmark}|
&=|\vbl(C)\setminus\cbra{v\in\vbl(C)\cap\Lambda(\pi)\mid\pi(v)=\EQmark}|\\
&\ge|\vbl(C)|-|\vbl(C)\cap\Vcalsep|-|\vbl(C)\cap\Lambda(\pi)\setminus\Vcalsep|\\
&\ge(k-2)-2\eta k-(1+(2/3-2\eta)k)\\
&=k/3-3,
\end{align*}
where we use \Cref{prop:width_k-2}, \Cref{alg:the_constructsep_algorithm}, and the definition of $\Ccalfrozen^\pi$ for the third line.
Thus $|\vbl(C)\cap\Vcal'|\ge k/3-3\ge(1-\eta)k/3$.

Note that $\barCcalint^\pi\cap\Ccalfrozen^\pi\subseteq\Wcal^\pi$ has size at most $n/2^{4k/\log(k)}$ by \Cref{itm:lem:prune_witness_1} of \Cref{lem:prune_witness}.
By \Cref{prop:bkvars} with $b=(1-\eta)/3$, we have
\begin{equation*}
|\Vcal'|\ge|\Ccalint^{\sigma'}\cap\Ccalfrozen^{\sigma'}|\cdot(1-4\eta)k/3.
\tag*{\qedhere}
\end{equation*}
\end{proof}

\subsubsection*{Truncation before the Final Rejection Sampling}

Now we turn to the sampling truncated nodes $\pi\in\Ncalsamptrunc$, which corresponds to partial assignments $\pi\in\Ncaltrunc$ that do not contain any $\Qmark$.
In this case we truncate if some connected component in $\Phi^\pi$ has more than $s$ clauses.

Let $\Phi'=(\Vcal',\Ccal')$ be the maximal connected component in $\Phi^\pi$ of size $|\Ccal'|>s$.
To keep notation consistent, we start with $\CcalQmarkint^\pi$ and $\Ccalint^\pi$.
Here, the construction is simple: We set $\CcalQmarkint^\pi=\emptyset$ and $\Ccalint^\pi=\Ccal'$.
Now, comparing with \Cref{lem:first_witness_size_bound}, we now have a simpler and better lower bound:
$$
|\CcalQmarkint^\pi\cup\Ccalint^\pi|=|\Ccal'|\ge s.
$$
To deal with the same trouble of $|\CcalQmarkint^\pi\cup\Ccalint^\pi|$ exceeding the threshold for structural properties, we perform the pruning in \Cref{lem:prune_witness}.

\begin{claim}\label{clm:prune_witness_rej}
\Cref{lem:prune_witness} works for $\pi\in\Ncalsamptrunc$ as well.
\end{claim}
\begin{proof}
We only need to verify \Cref{itm:lem:prune_witness_3} of \Cref{lem:prune_witness} for the starting case $\barCcalQmarkint^\pi=\CcalQmarkint^\pi=\emptyset,\barCcalint^\pi=\Ccalint^\pi=\Ccal'$, and the rest follows the proof of \Cref{lem:prune_witness} identically.

Let $C\in\Ccalbad^\pi$ and $v\in\vbl(C)$ with $\pi(v)=\EQmark$ be arbitrary. 
According to \Cref{def:sim}, we have $\Vcalalive^\pi=\emptyset$ and thus $v\notin\Vcalalive^\pi$.
Recall the definition of $\Vcalalive^\pi$, and we have the following cases:
\begin{itemize}
\item If $v\in\Vcalsep$, then there exists $C'\in\Ccalsep$ such that $v\in\vbl(C')$ as well. Then $C'\in\Ccal'$ since $\Ccal'$ is maximally connected and $C'$ is not satisfied by $\pi$ due to \Cref{fct:sim_marking} and \Cref{as:marking}.
\item Otherwise, there exists $C'\in\Ccal\setminus\Ccalsep$ such that $|\vbl(C')\cap\Lambda(\pi)\setminus\pbra{\Vcalsep\cup\cbra{v}}|<(2/3-2\eta)k$ and $C'(\pi)\neq\True$.
If $v\notin\vbl(C')$, then $|\vbl(C')\cap\Lambda(\pi)\setminus\Vcalsep|<(2/3-2\eta)k$ and thus violating \Cref{as:marking} and \Cref{fct:sim_marking}.
Therefore $v\in\vbl(C')$ and $|\vbl(C')\cap\Lambda(\pi)\setminus\Vcalsep|<1+(2/3-2\eta)k$.
This means $C'\in\Ccalfrozen^\pi$ and $C'$ is connected to $C$ in $G_\Phi$, which implies $C'\in\Ccal'$ as $\Ccal'$ is maximally connected.
\qedhere
\end{itemize}
\end{proof}

After the pruning, we define our witness $\Wcal^\pi=\barCcalQmarkint^\pi\cup\barCcalint^\pi$ analogously.
Then \Cref{lem:witness_saturation} and \Cref{lem:unsat_in_frozen_witness} can be proved identically due to \Cref{clm:prune_witness_rej}.

\subsubsection*{Summarizing Properties of the Witness}

Finally we summarize the properties of the witness.
\begin{corollary}\label{cor:witness}
Assume $6k^4\alpha\log(n)<s\le n/2^{5k/\log(k)}$.
Let $\pi\in\Ncaltrunc$.
Then we have witness $\Wcal^\pi=\barCcalQmarkint^\pi\cup\barCcalint^\pi\subseteq\Ccal$ such that the following holds:
\begin{enumerate}
\item\label{itm:cor:witness_1} $s/(6k^4\alpha)\le|\Wcal^\pi|\le n/2^{4k/\log(k)}$, $\Wcal^\pi$ is connected in $G_\Phi$, and $\barCcalQmarkint^\pi,\barCcalint^\pi$ are disjoint.
\item\label{itm:cor:witness_2} $\barCcalQmarkint^\pi$ covers at least $|\barCcalQmarkint^\pi|$ many $\Qmark$'s in $\pi$.
\item\label{itm:cor:witness_3} $|\barCcalQmarkint^\pi|+|\barCcalint^\pi\cap\Ccalfrozen^\pi|\ge(1-5\eta)\cdot|\Wcal^\pi|$.
\item\label{itm:cor:witness_4} $\barCcalint^\pi\cap\Ccalfrozen^\pi$ accesses at least $|\barCcalint^\pi\cap\Ccalfrozen^\pi|\cdot(1-4\eta)k/3$ distinct variables in $\pi$.
\end{enumerate}
\end{corollary}
\begin{proof}
We verify for $\pi\in\Ncalrectrunc$ and the argument for $\Ncalsamptrunc$ is similar due to the discussion above.
By \Cref{lem:first_witness_size_bound}, we get $|\CcalQmarkint^\pi\cup\Ccalint^\pi|\ge s/(6k^4\alpha)$.
Then we perform the pruning in \Cref{lem:prune_witness}.
Since $s\le n/2^{5k/\log(k)}$ and $\alpha\ge1/k^3$, the obtained witness $\Wcal^\pi=\barCcalQmarkint^\pi\cup\barCcalint^\pi$ has size at least $s/(6k^4\alpha)$, which is at least $\log(n)$ as $s>6k^4\alpha\log(n)$.
The other properties follow directly from \Cref{lem:prune_witness}, \Cref{lem:witness_saturation}, and \Cref{lem:unsat_in_frozen_witness}.
\end{proof}

\subsection{Refutation of Witnesses}\label{sec:refutation_of_large_component}

We now show that the number of possible truncation witnesses is small and the algorithm visits any one of them in small probability.
These two combined establishes \Cref{lem:phalt_s} by a union bound.

To better describe the witness, we will provide side information on $\Wcal^\pi$ via the following augmentation. For technical issue, we need to provide the location $z$ for the first generated $\Qmark$ of $\pi$ in addition to $\Wcal^\pi$. This $\Qmark$ may not be covered in $\Wcal^\pi$ due to the pruning process \Cref{lem:prune_witness}.

\begin{definition}[Witness Augmentation]\label{def:witness_enc}
For $\pi\in\Ncaltrunc$, we augment $\Wcal^\pi=\barCcalQmarkint^\pi\cup\barCcalint^\pi$ to $(\ell,q,r,\Qcal,\Rcal,f,z)$ as follows:
\begin{itemize}
\item $\ell=|\Wcal^\pi|$, $q=|\barCcalint^\pi\cap\Ccalfrozen^\pi|$, and $r$ equals the number of $\Qmark$'s contained in $\barCcalQmarkint^\pi$.
\item $\Qcal=\barCcalint^\pi$, $\Rcal=\barCcalQmarkint^\pi$, and $f$ indicates the locations of the $r$ $\Qmark$'s in $\barCcalQmarkint^\pi$.
\item $z$ is the first generated\footnote{Formally, if $\pi\in\Tcal_\sigma$ for some $\sigma\in\Ncalrec$, then $z$ is the unique $\Qmark$ in $\sigma$.} $\Qmark$ of $\pi$ in $\Tcalsim$ (and set $z=\bot$ if $\pi$ has no $\Qmark$).
\end{itemize}
\end{definition}

By \Cref{cor:witness}, a large witness is guaranteed to exist if we set $s$ suitably large.
\begin{corollary}\label{cor:aug_size}
If $6k^4\alpha\log(n)<s\le n/2^{5k/\log(k)}$, then $\ell\ge s/(6k^4\alpha)$ in any witness augmentation.
\end{corollary}

In the reverse direction, we can count the number of possible witness augmentations satisfying properties in \Cref{cor:witness}.

\begin{lemma}\label{lem:aug_num}
Assume $6k^4\alpha\log(n)<s\le n/2^{5k/\log(k)}$.
For any fixed $\ell,q,r$, there are at most $n^7(k^3\alpha)^\ell k^{2r}$ possible $\Qcal,\Rcal,f,z$ such that $(\ell,q,r,\Qcal,\Rcal,f,z)$ is an augmentation for a witness satisfying properties in \Cref{cor:witness}.
Moreover, $\ell\ge s/(6k^4\alpha)$ and $q+r\ge(1-5\eta)\ell$.
\end{lemma}
\begin{proof}
Assume $(\ell,q,r,\Qcal,\Rcal,f,z)$ is the augmentation for $\Wcal^\pi,\pi\in\Ncaltrunc$.
Then $|\Qcal\cup\Rcal|=|\Wcal^\pi|=\ell$.
Since $\Wcal^\pi$ is connected in $G_\Phi$ by \Cref{itm:cor:witness_1} of \Cref{cor:witness}, $\Qcal\cup\Rcal$ has at most $m\cdot\alpha^2n^4(\Naturale k^2\alpha)^\ell$ possibilities by \Cref{prop:number_of_connected_sets}.
Then we enumerate $\Rcal$. By \Cref{itm:cor:witness_2} of \Cref{cor:witness}, $|\Rcal|\le r$, and thus $\Rcal$ has $\binom\ell{\le r}$ possibilities given $\Qcal\cup\Rcal$.
Now we list all possible $f$.
Since $\Rcal$ is fixed and contains at most $k|\Rcal|\le kr$ distinct variables, we know that $f$ has at most $\binom{kr}r$ probabilities.
Finally $z\in\Vcal\cup\cbra{\bot}$ has $n+1$ options.
In all, the total count is upper bounded by
$$
\alpha^2n^4m(\Naturale k^2\alpha)^\ell\cdot\binom\ell{\le r}\cdot\binom{kr}r\cdot(n+1)
\le n^7(k^3\alpha)^\ell k^{2r},
$$
where we use the fact that $m=\alpha n$, $n\ge2^{\Omega(k)}$, $\binom\ell{\le r}\le2^\ell$, and $\binom{kr}r\le(\Naturale k)^r\le k^{2r}$.
The ``moreover'' part follows directly from \Cref{itm:cor:witness_1,itm:cor:witness_2,itm:cor:witness_3} of \Cref{cor:witness}.
\end{proof}

Now we bound the probability of encountering any witness augmentation.
The idea here is that the witness augmentation determines the unsatisfied clauses $\Qcal$ and the clauses $\Rcal$ containing $\Qmark$'s.
Then as we keep this in mind and simulate the execution of the algorithm using the simulation tree $\Tcalsim$, whenever we need to fix a variable, we know it is fixed towards the unsatisfying direction if it appears in $\Qcal$, or it is fixed to $\Qmark$ if it is appears in $\Rcal$ and is indicated so by side information $f$.
For both cases, we have good probability bound by the edge values $\rho$ in $\Tcalsim$.

\begin{lemma}\label{lem:single_aug_prob}
Assume $6k^4\alpha\log(n)<s\le n/2^{5k/\log(k)}$.
Let $(\ell,q,r,\Qcal,\Rcal,f,z)$ be a witness augmentation.
Then
$$
\sum_{\substack{\pi\in\Ncaltrunc\\\text{augmented as }(\ell,q,r,\Qcal,\Rcal,f,z)}}\rho(\pi)\le
(2\delta)^r\pbra{2k^{20}2^{-k/3}}^q(1+\delta)^{\drec+1}.
$$
\end{lemma}
\begin{proof}
We first mark edges in $\Tcalsim$ leading to possible $\pi\in\Ncaltrunc$ augmented as $(\ell,q,r,\Qcal,\Rcal,f,z)$.

Let $\sigma\in\Tcalsim$ be an internal node. Let $\sigma_0,\sigma_1,\sigma_\Qmark$ be its child nodes which extend $\sigma$ by fixing variable $v$ to $0,1,\Qmark$ respectively.
We classify $\sigma$ into one of the following types and mark its outgoing edges accordingly:
\begin{enumerate}[label=(\roman*)]
\item If $v$ is identified by $f$ as a $\Qmark$ in $\Rcal$, then we say $\sigma$ is \textsf{T1} and mark the edge $\sigma\to\sigma_\Qmark$.
\item Else if $v\in\vbl(C)$ for some clause $C\in\Qcal$, then we say $\sigma$ is \textsf{T2} and mark edges $\sigma\to\sigma_\Qmark,\sigma\to\sigma_b$, where $b\in\bin$ is the unique value that does not satisfy $C$ if assigned to $v$.\footnote{Pedantically, if $v$ appears in $C$ as $v$, then $b=0$; otherwise $v$ appears in $C$ as $\neg v$, then $b=1$.}
\item Else if $z=\bot$, then we say $\sigma$ is \textsf{T3} and mark edges $\sigma\to\sigma_0,\sigma\to\sigma_1$.
\item Else if $v\neq z$ and $\sigma(z)=\EQmark$, then we say $\sigma$ is \textsf{T4} and mark edges $\sigma\to\sigma_0,\sigma\to\sigma_1$.
\item Else, we say $\sigma$ is \textsf{T5} and mark all edges $\sigma\to\sigma_0,\sigma\to\sigma_1,\sigma\to\sigma_\Qmark$.
\end{enumerate}

For correctness, we need to show that we do not miss any truncated leaf node.
Assume towards contradiction that $\pi\in\Ncaltrunc$ is missed and $\Wcal^\pi$ is augmented as $(\ell,q,r,\Qcal,\Rcal,f,z)$.
Along the path from root to $\pi$, let $\sigma\in\Tcalsim$ be the last node that the marked edges lead to. Define $\sigma_0,\sigma_1,\sigma_\Qmark,v$ as above, and it means the edge $\sigma\to\sigma_{\pi(v)}$ is missed.
Then we have the following case analysis:
\begin{itemize}
\item $\sigma$ is \textsf{T1}.
This cannot happen since $f$ indicates $\pi(v)=\Qmark$ and the edge is already marked.
\item $\sigma$ is \textsf{T2}.
Since $\Qcal=\barCcalint^\pi\subseteq\Ccalint^\pi$, we have $C(\pi)\neq\True$ by the definition of $\Ccalint^\pi$. Thus $\pi(v)$ equals $\Qmark$ or the unique $b\in\bin$ that does not satisfy $C$ if assigned to $v$. Since both edges $\sigma\to\sigma_\Qmark,\sigma\to\sigma_b$ are already marked, this is a contradiction.
\item $\sigma$ is \textsf{T3}.
This cannot happen since $\pi(v)\neq\Qmark$ by $z=\bot$, and the edge is already marked.
\item $\sigma$ is \textsf{T4}.
This means $\pi(v)=\Qmark$ since otherwise the edge $\sigma\to\sigma_{\pi(v)}$ is already marked.
By definition, $z$ is the first generated $\Qmark$ of $\pi$. Then due to $v\neq z$ and $\pi(v)=\Qmark$, $z$ is already visited before reaching $\sigma$ and updating $v$ to $\Qmark$. This contradicts $\sigma(v)=\EQmark$.
\item $\sigma$ is \textsf{T5}.
This cannot happen since all three edges are marked.
\end{itemize}

Let $\pi\in\Ncaltrunc$ be arbitrary and augmented as $(\ell,q,r,\Qcal,\Rcal,f,z)$.
Now we count the number of \textsf{T1/T2/T5} nodes on the path from the root to $\pi$ in $\Tcalsim$.
It is easy to see that there are exactly $r$ \textsf{T1} nodes, corresponding to $\Qmark$'s in $\Rcal=\barCcalQmarkint^\pi$.
The number of \textsf{T2} nodes is the number of accessed variables in $\Qcal=\barCcalint^\pi\supseteq\barCcalint^\pi\cap\Ccalfrozen^\pi$, minus some \textsf{T1} nodes that are also accessed in $\Qcal$. Therefore
\begin{align*}
\#\textsf{T2}
&\ge\#\text{accessed variables in $\barCcalint^\pi\cap\Ccalfrozen^\pi$}-\#\textsf{T1}\\
&\ge|\barCcalint^\pi\cap\Ccalfrozen^\pi|\cdot(1-4\eta)k/3-r
\tag{by \Cref{itm:cor:witness_4} of \Cref{cor:witness}}\\
&=(1-4\eta)kq/3-r.
\end{align*}
To bound the number of \textsf{T5} nodes, we observe that it can only appear upon and after $z\neq\bot$ is access on the path.
Since $z$ is the first generated $\Qmark$ in $\pi$, this means, by the definition of $\Tcalsim$, we enter a recursive cost tree after accessing $z$.
Hence its number of upper bounded by one plus the depth of this particular recursive cost tree, which is in turn bounded by $1+\drec$ by definition.

Let $\Tcal$ be the sub-tree of $\Tcalsim$ consisting of all paths from the root to truncated leaf nodes $\Ncaltrunc$ augmented as $(\ell,q,r,\Qcal,\Rcal,f,z)$.
By the argument above, all edges in $\Tcal$ are marked. In addition, along any root-to-leaf path of $\Tcal$, there are
\begin{equation}\label{eq:lem:single_aug_prob_1}
\text{$r$ \textsf{T1} nodes, $\ge(1-4\eta)kq/3-r$ \textsf{T2} nodes, $\le1+\drec$ \textsf{T5} nodes,}
\end{equation}
and the rest are \textsf{T3/4} nodes.
In addition, by the definition of the types and edge weights $\rho$, we have
\begin{equation}\label{eq:lem:single_aug_prob_2}
\text{total weight of the outgoing marked edges of a}
\begin{cases}
\textsf{T1}\text{ node}\\
\textsf{T2}\text{ node}\\
\textsf{T5}\text{ node}\\
\textsf{T3/T4}\text{ node}
\end{cases}
\text{ is }
\begin{cases}
=\delta,\\
\le(1+3\delta)/2,\\
=1+\delta,\\
=1,
\end{cases}
\end{equation}
where we use \Cref{cor:marginal_local_uniformity} for \textsf{T2} nodes.
As a result, we have
\begin{align*}
\sum_{\substack{\pi\in\Ncaltrunc\\\text{augmented as }(\ell,q,r,\Qcal,\Rcal,f,z)}}\rho(\pi)
&=\sum_{\text{leaf }\pi\in\Tcal}\rho(\pi)
=\sum_{\substack{\text{root-to-leaf path }\Pcal\text{ in }\Tcal\\\text{edge }e\text{ on }\Pcal}}\rho(e)\\
&\le\max_\Pcal
\delta^{\#\textsf{T1}\in\Pcal}\pbra{\frac{1+3\delta}2}^{\#\textsf{T2}\in\Pcal}(1+\delta)^{\#\textsf{T5}\in\Pcal}
\tag{by \Cref{eq:lem:single_aug_prob_2}}\\
&\le\delta^r\pbra{\frac{1+3\delta}2}^{(1-4\eta)kq/3-r}(1+\delta)^{1+\drec}
\tag{by \Cref{eq:lem:single_aug_prob_1}}\\
&\le(2\delta)^r\pbra{\pbra{\frac{1+3\delta}2}^{k/3}\cdot2^{4\eta k/3}}^q(1+\delta)^{1+\drec}\\
&=(2\delta)^r\pbra{\pbra{\frac{1+3\delta}2}^{k/3}\cdot k^{20}}^q(1+\delta)^{1+\drec}
\tag{since $\eta=15\log(k)/k$}\\
&\le(2\delta)^r\pbra{2^{-k/3}\cdot2\cdot k^{20}}^q(1+\delta)^{1+\drec},
\end{align*}
where we use the fact $\delta=\xi/(k^{40}\alpha)\le1/k^{37}$ by $\alpha\ge1/k^3$ and $\xi\le1$ for the last line.
\end{proof}

Finally we are ready to prove \Cref{lem:phalt_s} .

\begin{proof}[Proof of \Cref{lem:phalt_s}]
By \Cref{lem:rct_sim_phalt}, it suffices to enumerate all possible witness augmentations using \Cref{lem:aug_num} and apply \Cref{lem:single_aug_prob} for each fixed one:
\begin{align*}
p_\textsf{halt}(\Phi,s)
&\le\sum_{\text{possible }(\ell,q,r,\Qcal,\Rcal,f,z)}\sum_{\substack{\pi\in\Ncaltrunc\\\text{augmented as }(\ell,q,r,\Qcal,\Rcal,f,z)}}\rho(\pi)\\
&\le\sum_{\text{possible }(\ell,q,r,\Qcal,\Rcal,f,z)}(2\delta)^r\pbra{2k^{20}2^{-k/3}}^q(1+\delta)^{\drec+1}
\tag{by \Cref{lem:single_aug_prob}}\\
&\le\sum_{\substack{\ell,q,r\\\ell\ge s/(6k^4\alpha)\\q+r\ge(1-5\eta)\ell}}n^7(k^3\alpha)^\ell k^{2r}\cdot(2\delta)^r\pbra{2k^{20}2^{-k/3}}^q(1+\delta)^{\drec+1}
\tag{by \Cref{lem:aug_num}}\\
&\le n^7(1+\delta)^{\drec+1}\sum_{\substack{\ell,q,r\\\ell\ge s/(6k^4\alpha)}}(k^3\alpha)^\ell(2k^2\delta)^{(1-5\eta)\ell-q}\pbra{2k^{20}2^{-k/3}}^q
\tag{since $2k^2\delta\le1$}\\
&=n^7(1+\delta)^{\drec+1}\sum_{\substack{\ell,q,r\\\ell\ge s/(6k^4\alpha)}}\pbra{\frac{2k^5\alpha\cdot\delta}{(2k^2\delta)^{5\eta}}}^\ell\pbra{\frac{k^{18}2^{-k/3}}\delta}^q\\
&=n^7(1+\delta)^{\drec+1}\sum_{\substack{\ell,q,r\\\ell\ge s/(6k^4\alpha)}}\pbra{\frac{2k^{-35}\xi}{(2k^{-38}\xi/\alpha)^{5\eta}}}^\ell\pbra{\frac{k^{58}2^{-k/3}\alpha}\xi}^q
\tag{since $\delta=\xi/(k^{40}\alpha)$}\\
&\le n^7(1+\delta)^{\drec+1}\sum_{\substack{\ell,q,r\\\ell\ge s/(6k^4\alpha)}}\pbra{\frac{2k^{-35}\xi}{(2^{-k/3})^{5\eta}}}^\ell k^{8q}
\tag{since $\alpha\le\xi\cdot2^{k/3}/k^{50}$}\\
&\le n^7(1+\delta)^{\drec+1}\sum_{\substack{\ell,q,r\\\ell\ge s/(6k^4\alpha)}}\pbra{\frac{2k^{-27}}{k^{-25}}}^\ell
\tag{since $q\le\ell$, $\xi\le1$, and $\eta=15\log(k)/k$}\\
&\le n^{10}(1+\delta)^{\drec+1}\sum_{\ell\ge s/(6k^4\alpha)}\pbra{2k^{-2}}^\ell
\tag{since $r\le n$ and $q\le m=\alpha n$}\\
&\le n^{10}(1+\delta)^{\drec+1}\cdot k^{-s/(6k^4\alpha)}
\tag{since $k\ge2^{20}$}
\end{align*}
as desired.
\end{proof}
\section{Putting Everything Together}\label{sec:putting_everything_together}

Now we put everything together and prove \Cref{thm:main_algorithm}.
As we mentioned before, the final algorithm is a combination of two different ones for different ranges of parameters.
The atypical setting refers to the case where $\alpha$ or $\eps$ is very small, and it will be handled by the naive rejection sampling described in \Cref{sec:the_rejectionsampling_subroutine}.
The typical setting is the case where both $\alpha$ and $\eps$ are reasonably large, then it will be handled by the more sophisticated \SolutionSampling{$\Phi,s$} presented in \Cref{sec:the_sampling_algorithm}.

\begin{proof}[Proof of \Cref{thm:main_algorithm}]
If $\eps\le\exp\cbra{-n/2^{k/2}}$ or $\alpha\le1/k^3$, we run \RejectionSampling{$\EQmark^\Vcal,\Vcal$} for $\tilde O\pbra{(n/\eps)^{1+\xi/k}}$ steps.
By Markov's inequality and \Cref{lem:small_eps,lem:small_alpha}, the probability of \RejectionSampling{$\EQmark^\Vcal,\Vcal$} not terminating within these number of steps is at most $\eps$ when the instance is good, which, by \Cref{cor:good_nice_whp}, happens with probability $1-o(1/n)$.
In addition, the total variation distance between the output and $\mu$ is guaranteed to be at most $\eps$ as desired.

If $\alpha\ge1/k^3$ and $\eps\ge\exp\cbra{-n/2^{k/2}}$, we run \SolutionSampling{$\Phi,s$} with 
$$
s=6k^4\alpha\log(n/\eps)=\tilde O(1)
$$
for $\tilde O\pbra{(n/\eps)^{1+\xi/k}/\xi}$ steps. Now we prove the correctness assuming the input is a nice instance, which happens with probability $1-o(1/n)$ by \Cref{cor:good_nice_whp}.
By \Cref{lem:rct_sim_efficiency} and \Cref{lem:drec_bound}, the expected runtime of \SolutionSampling{$\Phi,s$} is bounded by
\begin{align*}
\tilde O\pbra{\frac n\xi\cdot(1+\delta)^{s\cdot k+1}\cdot\exp\cbra{\frac{\xi\cdot s}{k^6(\alpha+1)}}}
&\le\tilde O\pbra{\frac{n\cdot(n/\eps)^{\xi/(2k)}}\xi\cdot\exp\cbra{\frac{\xi\cdot s}{k^6(\alpha+1)}}}\\
&\le\tilde O\pbra{\frac{n\cdot(n/\eps)^{\xi/(2k)}}\xi\cdot\exp\cbra{\frac{6\xi\cdot\log(n/\eps)}{k^2}}}\\
&\le\tilde O\pbra{(n/\eps)^{\xi/k}\cdot n/\xi},
\end{align*}
where we use the bound
\begin{equation}\label{eq:everything_1}
(1+\delta)^{s\cdot k}\le\Naturale^{s\cdot k\cdot\delta}=\exp\cbra{\frac{6\xi\cdot\log(n/\eps)}{k^{35}}}\le(n/\eps)^{\xi/(2k)}
\end{equation}
for the first inequality.
Thus the probability of not terminating within the prescribed number of steps is at most $\eps/2$ by Markov's inequality.
Since $\exp\cbra{-n/2^{k/2}}\le\eps<1$ and $\alpha\le2^{k/3}/k^{50}$, we have
$$
6k^4\alpha\log(n)<s\le 6k^4\alpha\cdot2n/2^{k/2}\le n/2^{k/6}\le n/2^{5k/\log(k)}.
$$
By \Cref{lem:rct_sim_phalt} and \Cref{lem:phalt_s}, we have
\begin{align*}
p_\textsf{halt}(\Phi,s)
&\le n^{10}(1+\delta)^{s\cdot k+2}\cdot k^{-s/(6k^4\alpha)}
\tag{by \Cref{lem:drec_bound}}\\
&\le4n^{11}/\eps\cdot k^{-s/(6k^4\alpha)}
\tag{by \Cref{eq:everything_1} and $\xi/(2k)\le1$}\\
&=4n^{11}/\eps\cdot k^{-\log(n/\eps)}\\
&\le\eps/2.
\tag{since $k\ge2^{20}$}
\end{align*}
Then by \Cref{cor:correctness_of_main_algorithm_trunc}, the total variation distance of the output distribution and $\mu$ is at most $\eps/2+p_\textsf{halt}(\Phi,s)\le\eps$ as desired, where the first $\eps/2$ comes from algorithm not terminating within $\tilde O\pbra{(n/\eps)^{1+\xi/k}/\xi}$ steps.
\end{proof}

\section*{Acknowledgement}
We thank anonymous SODA'23 reviewers for helpful comments.
KW and KY want to thank Wen Cao for providing discussion rooms in Spring 2021.
KW also wants to thank Christian Borgs and Guilhem Semerjian for helpful references regarding random $k$-CNF formulas.

\bibliographystyle{alpha} 
\bibliography{ref}

\newcommand{\etalchar}[1]{$^{#1}$}
\begin{thebibliography}{GGGY21}

\bibitem[ACO08]{AC2008}
Dimitris Achlioptas and Amin Coja-Oghlan.
\newblock {Algorithmic Barriers from Phase Transitions}.
\newblock In {\em Proceedings of the 49th Annual IEEE Symposium on Foundations
  of Computer Science}, pages 793--802. IEEE, oct 2008.

\bibitem[AJ22]{anand2022perfect}
Konrad Anand and Mark Jerrum.
\newblock Perfect sampling in infinite spin systems via strong spatial mixing.
\newblock {\em SIAM Journal on Computing}, 51(4):1280--1295, 2022.

\bibitem[AM02]{AM2002}
D.~Achlioptas and C.~Moore.
\newblock {The Asymptotic Order of the Random {$k$}-{SAT} Threshold}.
\newblock In {\em The 43rd Annual IEEE Symposium on Foundations of Computer
  Science, 2002. Proceedings.}, pages 779--788. IEEE Comput. Soc, 2002.

\bibitem[AP03]{AP2003}
Dimitris Achlioptas and Yuval Peres.
\newblock {The threshold for random $k$-{SAT} is $2^k(\ln 2 - O(k))$}.
\newblock In {\em Proceedings of the thirty-fifth ACM symposium on Theory of
  computing - STOC '03}, page 223, New York, New York, USA, 2003. ACM Press.

\bibitem[AZ08]{ardelius2008exhaustive}
John Ardelius and Lenka Zdeborov{\'a}.
\newblock Exhaustive enumeration unveils clustering and freezing in the random
  3-satisfiability problem.
\newblock {\em Physical Review E}, 78(4):040101, 2008.

\bibitem[BGG{\etalchar{+}}19]{BGGGS19}
Ivona Bez{\'{a}}kov{\'{a}}, Andreas Galanis, Leslie~A. Goldberg, Heng Guo, and
  Daniel {\v{S}}tefankovi{\v{c}}.
\newblock Approximation via correlation decay when strong spatial mixing fails.
\newblock {\em {SIAM} J. Comput.}, 48(2):279--349, 2019.

\bibitem[BH22]{BH2021}
Guy Bresler and Brice Huang.
\newblock The algorithmic phase transition of random k-sat for low degree
  polynomials.
\newblock In {\em 2021 IEEE 62nd Annual Symposium on Foundations of Computer
  Science (FOCS)}, pages 298--309, 2022.

\bibitem[B{\'o}n06]{bona2006walk}
Mikl{\'o}s B{\'o}na.
\newblock {\em A walk through combinatorics: an introduction to enumeration and
  graph theory}.
\newblock World Scientific, 2006.

\bibitem[CF14]{DBLP:journals/siamcomp/Coja-OghlanF14}
Amin Coja{-}Oghlan and Alan~M. Frieze.
\newblock Analyzing walksat on random formulas.
\newblock {\em {SIAM} J. Comput.}, 43(4):1456--1485, 2014.

\bibitem[CMM23]{chen2023algorithms}
Zongchen Chen, Nitya Mani, and Ankur Moitra.
\newblock From algorithms to connectivity and back: finding a giant component
  in random k-sat.
\newblock In {\em Proceedings of the 2023 Annual ACM-SIAM Symposium on Discrete
  Algorithms (SODA)}, pages 3437--3470. SIAM, 2023.

\bibitem[COHH17]{CHH2017}
A.~Coja-Oghlan, A.~Haqshenas, and S.~Hetterich.
\newblock {{\tt {W}alksat} Stalls Well Below Satisfiability}.
\newblock {\em SIAM Journal on Discrete Mathematics}, 31(2):1160--1173, jan
  2017.

\bibitem[Coj10]{Coja2010}
Amin Coja{-}Oghlan.
\newblock A better algorithm for random k-sat.
\newblock {\em {SIAM} J. Comput.}, 39(7):2823--2864, 2010.

\bibitem[COP16]{CP2016}
Amin Coja-Oghlan and Konstantinos Panagiotou.
\newblock The asymptotic k-sat threshold.
\newblock {\em Advances in Mathematics}, pages 985--1068, 2016.

\bibitem[DHKN21]{dughmi2021bernoulli}
Shaddin Dughmi, Jason Hartline, Robert~D Kleinberg, and Rad Niazadeh.
\newblock Bernoulli factories and black-box reductions in mechanism design.
\newblock {\em Journal of the ACM (JACM)}, 68(2):1--30, 2021.

\bibitem[DSS22]{DSS2022}
Jian Ding, Allan Sly, and Nike Sun.
\newblock Proof of the satisfiability conjecture for large $ k$.
\newblock {\em Annals of Mathematics}, 196(1):1--388, 2022.

\bibitem[EL73]{erdHos1973problems}
Paul Erd{\H{o}}s and L{\'a}szl{\'o} Lov{\'a}sz.
\newblock Problems and results on 3-chromatic hypergraphs and some related
  questions.
\newblock In {\em Colloquia Mathematica Societatis Janos Bolyai 10. Infinite
  and Finite Sets, Keszthely (Hungary)}. Citeseer, 1973.

\bibitem[ES91]{erdos1991lopsided}
P~Erdos and Joel Spencer.
\newblock Lopsided lovsz local lemma and latin transversals.
\newblock {\em Discrete Applied Mathematics}, 30(151-154):10--1016, 1991.

\bibitem[FGYZ21]{FGYZ20}
Weiming Feng, Heng Guo, Yitong Yin, and Chihao Zhang.
\newblock Fast sampling and counting $k$-sat solutions in the local lemma
  regime.
\newblock {\em Journal of the ACM (JACM)}, 68(6):1--42, 2021.

\bibitem[FHY21]{feng2021sampling}
Weiming Feng, Kun He, and Yitong Yin.
\newblock Sampling constraint satisfaction solutions in the local lemma regime.
\newblock In {\em Proceedings of the 53rd Annual ACM SIGACT Symposium on Theory
  of Computing}, pages 1565--1578, 2021.

\bibitem[Fri99]{FR1999}
Ehud Friedgut.
\newblock Sharp thresholds of graph properties, and the {$k$}-sat problem.
\newblock {\em J. Amer. Math. Soc.}, 12(4):1017--1054, 1999.
\newblock With an appendix by Jean Bourgain.

\bibitem[GGGH22]{DBLP:journals/corr/abs-2206-15308}
Andreas Galanis, Leslie~Ann Goldberg, Heng Guo, and Andr{\'{e}}s
  Herrera{-}Poyatos.
\newblock Fast sampling of satisfying assignments from random k-sat.
\newblock {\em CoRR}, abs/2206.15308, 2022.

\bibitem[GGGY21]{DBLP:journals/siamcomp/GalanisGGY21}
Andreas Galanis, Leslie~Ann Goldberg, Heng Guo, and Kuan Yang.
\newblock Counting solutions to random {CNF} formulas.
\newblock {\em {SIAM} J. Comput.}, 50(6):1701--1738, 2021.

\bibitem[GST16]{gebauer2016local}
Heidi Gebauer, Tibor Szab{\'o}, and G{\'a}bor Tardos.
\newblock The local lemma is asymptotically tight for sat.
\newblock {\em Journal of the ACM (JACM)}, 63(5):1--32, 2016.

\bibitem[Het16]{Het2016}
Samuel Hetterich.
\newblock Analysing survey propagation guided decimationon random formulas.
\newblock In {\em {ICALP}}, volume~55 of {\em LIPIcs}, pages 65:1--65:12.
  Schloss Dagstuhl - Leibniz-Zentrum fuer Informatik, 2016.

\bibitem[HSS11]{DBLP:journals/jacm/HaeuplerSS11}
Bernhard Haeupler, Barna Saha, and Aravind Srinivasan.
\newblock New constructive aspects of the lov{\'{a}}sz local lemma.
\newblock {\em J. {ACM}}, 58(6):28:1--28:28, 2011.

\bibitem[HSW21]{HSW21}
Kun He, Xiaoming Sun, and Kewen Wu.
\newblock Perfect sampling for (atomic) lov{\'{a}}sz local lemma.
\newblock {\em CoRR}, abs/2107.03932, 2021.

\bibitem[HSZ19]{DBLP:journals/rsa/HermonSZ19}
Jonathan Hermon, Allan Sly, and Yumeng Zhang.
\newblock Rapid mixing of hypergraph independent sets.
\newblock {\em Random Struct. Algorithms}, 54(4):730--767, 2019.

\bibitem[Hub16]{huber2016nearly}
Mark Huber.
\newblock Nearly optimal bernoulli factories for linear functions.
\newblock {\em Combinatorics, Probability and Computing}, 25(4):577--591, 2016.

\bibitem[HWY22]{he2022sampling}
Kun He, Chunyang Wang, and Yitong Yin.
\newblock Sampling lov{\'{a}}sz local lemma for general constraint satisfaction
  solutions in near-linear time.
\newblock In {\em 63rd {IEEE} Annual Symposium on Foundations of Computer
  Science, {FOCS} 2022, Denver, CO, USA, October 31 - November 3, 2022}, pages
  147--158. {IEEE}, 2022.

\bibitem[HWY23]{HWY22Deterministic}
Kun He, Chunyang Wang, and Yitong Yin.
\newblock Deterministic counting lov{\'{a}}sz local lemma beyond linear
  programming.
\newblock In Nikhil Bansal and Viswanath Nagarajan, editors, {\em Proceedings
  of the 2023 {ACM-SIAM} Symposium on Discrete Algorithms, {SODA} 2023,
  Florence, Italy, January 22-25, 2023}, pages 3388--3425. {SIAM}, 2023.

\bibitem[JPV21]{Vishesh21sampling}
Vishesh Jain, Huy~Tuan Pham, and Thuy~Duong Vuong.
\newblock On the sampling lov{\'{a}}sz local lemma for atomic constraint
  satisfaction problems.
\newblock {\em CoRR}, abs/2102.08342, 2021.

\bibitem[KKKS98]{KKKS1998}
Lefteris~M. Kirousis, Evangelos Kranakis, Danny Krizanc, and Yannis~C.
  Stamatiou.
\newblock {Approximating the unsatisfiability threshold of random formulas}.
\newblock {\em {Random Structures \& Algorithms}}, 12(3):253--269, 1998.

\bibitem[MMZ05]{MMZ2005}
M.~M\'ezard, T.~Mora, and R.~Zecchina.
\newblock Clustering of solutions in the random satisfiability problem.
\newblock {\em Phys. Rev. Lett.}, 94:197205, 2005.

\bibitem[Moi19]{Moi19}
Ankur Moitra.
\newblock Approximate counting, the lov{\'{a}}sz local lemma, and inference in
  graphical models.
\newblock {\em J. {ACM}}, 66(2):10:1--10:25, 2019.

\bibitem[MPZ02]{MPZ2002}
M.~M{\'e}zard, G.~Parisi, and R.~Zecchina.
\newblock Analytic and algorithmic solution of random satisfiability problems.
\newblock {\em Science}, 297(5582):812--815, 2002.

\bibitem[MS07]{MS2007}
Andrea Montanari and Devavrat Shah.
\newblock {Counting good truth assignments of random $k$-{SAT} formulae}.
\newblock In {\em Proceedings of the Eighteenth Annual ACM-SIAM Symposium on
  Discrete Algorithms (SODA 2007)}, pages 1255--1264, jul 2007.

\bibitem[MU17]{mitzenmacher2017probability}
Michael Mitzenmacher and Eli Upfal.
\newblock {\em Probability and computing: Randomization and probabilistic
  techniques in algorithms and data analysis}.
\newblock Cambridge university press, 2017.

\bibitem[NP05]{nacu2005fast}
{\c{S}}erban Nacu and Yuval Peres.
\newblock Fast simulation of new coins from old.
\newblock {\em The Annals of Applied Probability}, 15(1A):93--115, 2005.

\bibitem[QWZ22]{QWZ22}
Guoliang Qiu, Yanheng Wang, and Chihao Zhang.
\newblock A perfect sampler for hypergraph independent sets.
\newblock In Mikolaj Bojanczyk, Emanuela Merelli, and David~P. Woodruff,
  editors, {\em 49th International Colloquium on Automata, Languages, and
  Programming, {ICALP} 2022, July 4-8, 2022, Paris, France}, volume 229 of {\em
  LIPIcs}, pages 103:1--103:16. Schloss Dagstuhl - Leibniz-Zentrum f{\"{u}}r
  Informatik, 2022.

\bibitem[RS98]{RS98}
Martin Raab and Angelika Steger.
\newblock "balls into bins" - {A} simple and tight analysis.
\newblock In Michael Luby, Jos{\'{e}} D.~P. Rolim, and Maria~J. Serna, editors,
  {\em Randomization and Approximation Techniques in Computer Science, Second
  International Workshop, RANDOM'98, Barcelona, Spain, October 8-10, 1998,
  Proceedings}, volume 1518 of {\em Lecture Notes in Computer Science}, pages
  159--170. Springer, 1998.

\end{thebibliography}

\appendix

\section{Proofs of the Structural Properties}\label{app:missing_proofs_in_sec:properties_of_of_random_cnf_formulas}

\begin{proof}[Proof of \Cref{prop:width_k-2}]
Assume $k\ge3$.
For each $C\in\Ccal$, let $\Ecal(C)$ be the event that $\abs{\vbl(C)}\le k-3$ and $\indicator_{\Ecal(C)}\in\bin$ be the indicator of $\Ecal(C)$.
Then by union bound, we have
\begin{align*}
\E\sbra{\indicator_{\Ecal(C)}}
&=\Pr\sbra{\Ecal(C)}\le\binom n{k-3}\pbra{\frac{k-3}n}^k\le\pbra{\frac{\Naturale n}{k-3}}^{k-3}\pbra{\frac{k-3}n}^k=\frac{\Naturale^{k-3}(k-3)^3}{n^3}\\
&\le\frac1{\alpha n^{2.5}}=\frac1{mn^{1.5}}.
\tag{assume $\alpha\le2^k$ and $n\ge2^{\Omega(k)}$}
\end{align*}
Then by Markov's inequality, we have
\begin{equation*}
\Pr\sbra{\exists\text{ such }\Ecal(C)}=\Pr\sbra{\sum_{C\in\Ccal}\indicator_{\Ecal(C)}\ge1}\le\frac1{n^{1.5}}=o(1/n).
\tag*{\qedhere}
\end{equation*}
\end{proof}

\begin{proof}[Proof of \Cref{prop:kvars_and_distinct_vars}]
We first prove \Cref{itm:kvars_and_distinct_vars_1}.
By \Cref{prop:width_k-2}, we have $|\vbl(C)|\ge k-2$ for all $C\in\Ccal$ with probability $1-o(1/n)$. Given this, we focus on the case $k-2\le|\Vcal'|\le n/2^{k/\log(k)}$.

Let $s$ be an integer that $k-2\le s\le n/2^{k/\log(k)}$. Define $t=\ceilbra{(1+\eta)s/k}$.
For any fixed subset $X$ of variables of size $s$ and subset $Y$ of clauses of size $t$, we have
$$
\Pr\sbra{\vbl(C)\subseteq X,\forall C\in Y}=\pbra{\frac sn}^{k\cdot t}.
$$
Then by enumerating all possible $Y$, we have
\begin{align*}
\Pr\sbra{\exists Y,\vbl(C)\subseteq X,\forall C\in Y}
&\le\binom mt\pbra{\frac sn}^{kt}
\le\pbra{\frac{\Naturale\alpha n}t}^t\pbra{\frac sn}^{kt}
\tag{since $m=\alpha n$}\\
&\le\pbra{\frac{\Naturale\alpha n}{s/k}}^t\pbra{\frac sn}^{kt}
=(\Naturale k\alpha)^t\cdot\pbra{\frac sn}^{(k-1)t}
\tag{since $t\ge s/k$}\\
&\le2^{4\cdot((1+\eta)s+k)}\cdot\pbra{\frac sn}^{(k-1)t}
\tag{since $t\le(1+\eta)s/k+1$ and $\alpha\le2^k$}\\
&\le2^{4\cdot((1+\eta)s+k)}\cdot\pbra{\frac sn}^{(1-1/k)(1+\eta)s}.
\tag{since $t\ge(1+\eta)s/k$}
\end{align*}
Thus by union bound over all possible $X$, we have
\begin{align*}
\Pr\sbra{\exists\text{ such }X,Y}
&\le\sum_{s=k-2}^{\floorbra{n/2^{k/\log(k)}}}\binom ns
\cdot2^{4\cdot((1+\eta)s+k)}\cdot\pbra{\frac sn}^{(1-1/k)(1+\eta)s}\\
&\le\sum_{s=k-2}^{\floorbra{n/2^{k/\log(k)}}}\pbra{\frac{4n}s}^s
\cdot2^{4\cdot((1+\eta)s+k)}\cdot\pbra{\frac sn}^{(1-1/k)(1+\eta)s}\\
&=2^{4k}\sum_{s=k-2}^{\floorbra{n/2^{k/\log(k)}}}
\pbra{2^{4\eta+6}\pbra{\frac sn}^{\eta-\frac1k-\frac\eta k}}^s
\le2^{4k}\sum_{s=k-2}^{\floorbra{n/2^{k/\log(k)}}}
\pbra{2^{6(\eta+1)}\pbra{\frac sn}^{\eta/2}}^s
\tag{assume $\eta-\frac1k-\frac\eta k\ge\eta/2$}\\
&\le2^{4k}\sum_{s=k-2}^{\floorbra{\ln^2n}}
\pbra{2^{6(\eta+1)}\pbra{\frac{\ln^2n}n}^{\eta/2}}^s
+2^{4k}\sum_{s=\floorbra{\ln^2n}+1}^{\floorbra{n/2^{k/\log(k)}}}
\pbra{2^{6(\eta+1)}2^{-\frac{\eta k}{2\log(k)}}}^s\\
&\le2^{4k}n^{-\frac{\eta(k-2)}4}\sum_{s=k-2}^{\floorbra{\ln^2n}}
\pbra{2^{6(\eta+1)}\pbra{\frac{\ln^2n}{\sqrt n}}^{\eta/2}}^s
+2^{4k}\sum_{s=\floorbra{\ln^2n}+1}^{\floorbra{n/2^{k/\log(k)}}}
\pbra{2^{6(\eta+1)}2^{-\frac{\eta k}{2\log(k)}}}^s\\
&\le 2^{4k}n^{-2}\sum_{s=k-2}^{\floorbra{\ln^2n}}
2^{-s}
+2^{4k}\sum_{s=\floorbra{\ln^2n}+1}^{\floorbra{n/2^{k/\log(k)}}}
2^{-s}
\tag{assume $\eta(k-2)\ge8$, $n\ge2^{\Omega(1+1/\eta)}$, and $\frac{\eta k}{2\log(k)}\ge6\eta+7$}\\
&=o(1/n).
\tag{assume $n\ge2^{\Omega(k)}$}
\end{align*}
Finally we note that if $\alpha\le2^k$, $k/\log(k)\ge14(1+1/\eta)$, and $n\ge2^{\Omega(k)}$, then all the assumptions above are satisfied.

Now we turn to \Cref{itm:kvars_and_distinct_vars_2}. 
Fix an arbitrary $\Ccal'\subset\Ccal$ with $|\Ccal'|\le n/2^{2k/\log(k)}$. Let $\Vcal'=\bigcup_{C\in\Ccal'}\vbl(C)$ which satisfies $|\Vcal'|\le k|\Ccal'|\le n/2^{k/\log(k)}$. Then by \Cref{itm:kvars_and_distinct_vars_1}, we have
$$
|\Ccal'|\le\abs{\cbra{C\in\Ccal\mid\vbl(C)\subseteq\Vcal'}}\le(1+\eta)|\Vcal'|/k,
$$
which implies $\abs{\bigcup_{C\in\Ccal'}\vbl(C)}=|\Vcal'|\ge k|\Ccal'|/(1+\eta)$.
\end{proof}

To prove \Cref{prop:bkvars}, we will need the following technical lemma.

\begin{proposition}\label{prop:etakvars}
Let $\eta=\eta(k)\in(0,1)$ be a parameter.
Assume $\alpha\le2^k$, $\frac k{\log(k)}\ge\frac5\eta$, and $n\ge2^{\Omega(k)}$. Then with probability $1-o(1/n)$ over the random $\Phi$, the following holds: For every $\Vcal'\subset\Vcal$ with $|\Vcal'|\le n/2^{k/\log(k)}$, we have
$$
\abs{\cbra{C\in\Ccal\mid\abs{\vbl(C)\cap\Vcal'}\ge\eta k}}\le k|\Vcal'|.
$$
\end{proposition}
\begin{proof}
Let $s\le n/2^{k/\log(k)}$ be an integer. 
For any fixed subset $X$ of variables of size $s$ and subset $Y$ of clauses of size $ks$, we have
$$
\Pr\sbra{\abs{\vbl(C)\cap X}\ge\eta k,\forall C\in Y}
\le\pbra{\binom k{\ceilbra{\eta k}}\cdot\pbra{\frac sn}^{\ceilbra{\eta k}}}^{ks}
\le\pbra{2^k\cdot\pbra{\frac sn}^{\eta k}}^{ks}
=\pbra{\frac{2^{1/\eta}\cdot s}n}^{\eta k^2s},
$$
where $\binom k{\eta k}$ chooses the (first) $\eta k$ locations in $C$ that use variables from $X$.
Thus by union bound, we have
\begin{align*}
\Pr\sbra{\exists\text{ such }X,Y}
&\le\sum_{s=1}^{\floorbra{n/2^{k/\log(k)}}}\binom ns\binom m{ks}\cdot\pbra{\frac{2^{1/\eta}\cdot s}n}^{\eta k^2s}\\
&\le\sum_{s=1}^{\floorbra{n/2^{k/\log(k)}}}\pbra{\frac{\Naturale n}s}^s\cdot\pbra{\frac{\Naturale m}{ks}}^{ks}\cdot\pbra{\frac{2^{1/\eta}\cdot s}n}^{\eta k^2s}\\
&=\sum_{s=1}^{\floorbra{n/2^{k/\log(k)}}}\pbra{\frac{\Naturale^{k+1}2^{k^2}\alpha^ks^{\eta k^2-k-1}}{k^kn^{\eta k^2-k-1}}}^{s}
\tag{since $m=\alpha n$}\\
&\le\sum_{s=1}^{\floorbra{n/2^{k/\log(k)}}}\pbra{\frac{2^{2k^2}s^{\eta k^2-k-1}}{n^{\eta k^2-k-1}}}^{s}
\tag{since $\alpha\le2^k$ and assume $e^{k+1}\le k^k$}\\
&\le\sum_{s=1}^{\floorbra{n/2^{k/\log(k)}}}\pbra{\frac{2^4\cdot s^{\eta}}{n^{\eta}}}^{k^2s/2}
\tag{assume $\eta k^2-k-1\ge\eta k^2/2$}\\
&\le\sum_{s=1}^{\floorbra{\ln^2n}}\pbra{\frac{2^4\cdot {\ln^{2\eta} n}}{n^{\eta}}}^{k^2s/2}+\sum_{s=\floorbra{\ln^2n}+1}^{\floorbra{n/2^{k/\log(k)}}}\pbra{\frac{2^4}{2^{\eta k/\log(k)}}}^{k^2s/2}\\
&\le n^{-\eta k^2/4}\sum_{s=1}^{\floorbra{\ln^2n}}\pbra{\frac{2^{22}\cdot {\ln^{2\eta} n}}{n^{\eta/2}}}^{k^2s/2}+\sum_{s=\floorbra{\ln^2n}+1}^{\floorbra{n/2^{k/\log(k)}}}\pbra{\frac{2^4}{2^{\eta k/\log(k)}}}^{k^2s/2}\\
&\le n^{-2}\sum_{s=1}^{\floorbra{\ln^2n}}2^{-k^2s/2}+\sum_{s=\floorbra{\ln^2n}+1}^{\floorbra{n/2^{k/\log(k)}}}2^{-k^2s/2}
\tag{assume $\eta k^2\ge8$, $n\ge2^{\Omega(1/\eta)}$, and $\eta k/\log(k)\ge5$}\\
&=o(1/n).
\end{align*}
Finally we note that if $k/\log(k)\ge5/\eta$, $\alpha\le2^k$, and $n\ge2^{\Omega(k)}$, then all the assumptions above are satisfied.
\end{proof}

Now we proceed to the proof of \Cref{prop:bkvars}.

\begin{proof}[Proof of \Cref{prop:bkvars}]
We assume $\Phi$ satisfies the properties in \Cref{prop:etakvars} and \Cref{prop:kvars_and_distinct_vars}, which by union bound happens with probability $1-o(1/n)$.
We also assume $b\le1$ since otherwise the statement trivially holds.

Fix an arbitrary $\Vcal'\subset\Vcal$ with $|\Vcal'|\le n/2^{3k/\log(k)}$. Let $\Ccal'=\cbra{C\in\Ccal\mid\abs{\vbl(C)\cap\Vcal'}\ge bk}$. Then
$$
\abs{\bigcup_{C\in\Ccal'}\vbl(C)\setminus\Vcal'}\le(1-b)k\cdot|\Ccal'|.
$$
By \Cref{prop:etakvars}, we know $|\Ccal'|\le k|\Vcal'|\le n/2^{2k/\log(k)}$. 
Then by \Cref{itm:kvars_and_distinct_vars_2} of \Cref{prop:kvars_and_distinct_vars}, we have
$$
\frac{k|\Ccal'|}{1+\eta}\le\abs{\bigcup_{C\in \Ccal'}\vbl(C)}=|\Vcal'|+\abs{\bigcup_{C\in \Ccal'}\vbl(C)\setminus\Vcal'}\le|\Vcal'|+(1-b)k\cdot|\Ccal'|,
$$
which implies $|\Vcal'|\ge k|\Ccal'|\cdot\pbra{\frac1{1+\eta}-(1-b)}\ge(b-\eta)k\cdot|\Ccal'|$.
\end{proof}

\begin{lemma}[{\cite[Lemma 8.5]{DBLP:journals/siamcomp/GalanisGGY21}}]\label{lem:prob_of_tree}
For any labeled tree $T$ on a subset of $\Ccal$, the probability that $T$ is a sub-graph of $G_\Phi$ is at most $(k^2/n)^{|V(T)|-1}$ where $V(T)$ is the number of nodes of $T$.
\end{lemma}

\begin{proof}[Proof of \Cref{prop:number_of_connected_sets}]
Let $C\in\Ccal$ be arbitrary and let $U\subseteq\Ccal$ be a size-$\ell$ set of clauses containing $C$.
For any fixed labeled spanning tree on $U$, by \Cref{lem:prob_of_tree} it appears in $G_\Phi$ with probability at most $(k^2/n)^{\ell-1}$.
Meanwhile by standard result (See e.g., \cite{bona2006walk}), there are $\ell^{\ell-2}$ many possible $U$.
Thus by union bound, we have
$$
\Pr\sbra{G_\Phi[U]\text{ is connected}}\le\ell^{\ell-2}(k^2/n)^{\ell-1}.
$$
Now let $Z_{\ell,C}$ be the number of connected sets of clauses with size $\ell$ containing $C$. Then 
\begin{align*}
\E\sbra{Z_{\ell,C}}
&=\sum_{U\subseteq\Ccal:C\in U,|U|=\ell}\Pr\sbra{G_\Phi[U]\text{ is connected}}\\
&\le\binom{m-1}{\ell-1}\cdot\ell^{\ell-2}\pbra{\frac{k^2}n}^{\ell-1}
\le\pbra{\frac{\Naturale(m-1)}{\ell-1}}^{\ell-1}\cdot\ell^{\ell-2}\pbra{\frac{k^2}n}^{\ell-1}\\
&\le\pbra{\frac{\Naturale mk^2\cdot\ell^{\frac{\ell-2}{\ell-1}}}{n\cdot(\ell-1)}}^{\ell-1}
\le(\Naturale k^2\alpha)^{\ell-1}.
\tag{since $\ell^{\ell-2}\le(\ell-1)^{\ell-1}$ and $m=\alpha n$}
\end{align*}
Then by Markov's inequality, we have
$$
\Pr\sbra{Z_{\ell,C}\ge\alpha^2n^4(\Naturale k^2\alpha)^{\ell-1}}\le\alpha^{-2}n^{-4}=n^{-2}m^{-2}.
$$
Finally, by union bound, we have
\begin{equation*}
\Pr\sbra{\exists\text{ such }Z_{\ell,C}\ge\alpha^2n^3(\Naturale k^2\alpha)^{\ell-1}}\le m^2\cdot n^{-2}m^{-2}=1/n^2=o(1/n).
\tag*{\qedhere}
\end{equation*}
\end{proof}

\begin{proof}[Proof of \Cref{prop:number_of_neighbors}]
Define $\tilde\Vcal=\abs{\cbra{v\in\Vcal\mid v\in\Vcal'\text{ or }v\text{ is adjacent to }\Vcal'}}$.
Let 
$$
\Ccal'=\cbra{C\in\Ccal\mid\vbl(C)\cap\Vcal'\neq\emptyset}.
$$
Since $|\tilde\Vcal|\le k|\Ccal'|$, it suffices to bound $|\Ccal'|\le3k^3\alpha\max\cbra{|\Vcal'|,\floorbra{k\log(n)}}$.

We first focus on the case $|\Vcal'|\ge\floorbra{k\log(n)}$.
Since $H_\Phi[\Vcal']$ is connected, there exists some $\Ccal''\subseteq\Ccal'$ such that $|\Vcal'|/k\le|\Ccal''|\le|\Vcal'|$ and $\Vcal'$ is connected in $H_\Phi$ using $\Ccal''$.
In particular, we have 
$$
|\Ccal''|\ge\floorbra{k\log(n)}/k\ge\log(n)-1.
$$
Let $\tilde\Ccal=\Ccal'\setminus\Ccal''$. 
Since $k^3\alpha\ge1$, it suffices to bound $|\Ccal'|\le2k^3\alpha|\Vcal'|+|\Vcal'|$.
Then plugging in $|\Ccal'|=|\tilde\Ccal|+|\Ccal''|$ and $|\Ccal''|\le|\Vcal'|$, it suffices to prove $|\tilde\Ccal|\le2k^3\alpha|\Vcal'|$.
Now for any fixed $\Ccal'',\Vcal',\tilde\Ccal$ satisfying:
\begin{itemize}
\item $|\Ccal''|\ge\log(n)-1$, $|\Vcal'|\ge|\Ccal''|$, $|\tilde\Ccal|\ge2k^3\alpha|\Vcal'|$, and $\Ccal''\cap\tilde\Ccal=\emptyset$.
\item $G_\Phi[\Ccal'']$ is connected, $\Vcal'\subseteq\bigcup_{C\in\Ccal''}\vbl(C)$, and $\vbl(\tilde C)\cap\Vcal'\neq\emptyset$ holds for all $\tilde C\in\tilde\Ccal$.
\end{itemize}
Let $s_1=|\Ccal''|$, $s_2=|\Vcal'|$, and $s_3=|\tilde\Ccal|$.
We now define the following events:
\begin{itemize}
\item $\Ecal(\Ccal'',\Vcal',\tilde\Ccal)$ is the event that ``$\Ccal'',\Vcal',\tilde\Ccal$ satisfy the conditions above''.
\item $\Ecal(\Ccal'')$ is the event that ``$G_\Phi[\Ccal'']$ is connected''.
\item $\Ecal(\Vcal',\tilde\Ccal)$ is the event that ``$\vbl(\tilde C)\cap\Vcal'\neq\emptyset$ holds for all $\tilde C\in\tilde\Ccal$''.
\end{itemize}
By union bounding over all $s_1^{s_1-2}$ labeled spanning trees over $\Ccal''$ and using \Cref{lem:prob_of_tree}, we have
$$
\Pr\sbra{\Ecal(\Ccal'')}\le s_1^{s_1-2}\pbra{\frac{k^2}n}^{s_1-1}.
$$
Since $\Ccal''\cap\tilde\Ccal=\emptyset$, by independence we have
$$
\Pr\sbra{\Ecal(\Vcal',\tilde\Ccal)\mid\Ecal(\Ccal'')}=\Pr\sbra{\Ecal(\Vcal',\tilde\Ccal)}\le\pbra{k\cdot\frac{s_2}n}^{s_3}.
$$
Hence 
$$
\Pr\sbra{\Ecal(\Ccal'',\Vcal',\tilde\Ccal)}\le\Pr\sbra{\Ecal(\Ccal'')\land\Ecal(\Vcal',\tilde\Ccal)}\le s_1^{s_1-2}\pbra{\frac{k^2}n}^{s_1-1}\pbra{\frac{k s_2}n}^{s_3}.
$$
Thus by union bound, we have
\begin{align*}
\Pr\sbra{\exists\text{ such }\Ecal(\Ccal'',\Vcal',\tilde\Ccal)}
&\le\sum_{s_1\ge\log(n)-1}\sum_{s_2\ge s_1}\sum_{s_3\ge2k^3\alpha\cdot s_2}\binom m{s_1}\binom {ks_1}{s_2}\binom m{s_3}\cdot s_1^{s_1-2}\pbra{\frac{k^2}n}^{s_1-1}\pbra{\frac{k s_2}n}^{s_3}
\tag{$\binom{ks_1}{s_2}$ comes from $\Vcal'\subseteq\bigcup_{C\in\Ccal''}\vbl(C)$}\\
&\le\sum_{s_1\ge\log(n)-1}\sum_{s_2\ge s_1}\sum_{s_3\ge2k^3\alpha\cdot s_2}\frac n{k^2s_1^2}\pbra{\Naturale k^2\alpha}^{s_1}\pbra{\frac{\Naturale ks_1}{s_2}}^{s_2}\pbra{\frac{\Naturale k\alpha s_2}{s_3}}^{s_3}
\tag{since $m=\alpha n$}\\
&\le\sum_{s_1\ge\log(n)-1}\sum_{s_2\ge s_1}\frac n{k^2s_1^2}\pbra{\Naturale k^2\alpha}^{s_1}(\Naturale k)^{s_2}\sum_{s_3\ge2k^3\alpha\cdot s_2}\pbra{\frac\Naturale{2k^2}}^{s_3}
\tag{since $s_2\ge s_1$}\\
&\le\sum_{s_1\ge\log(n)-1}\sum_{s_2\ge s_1}\frac{2n}{k^2s_1^2}\pbra{\Naturale k^2\alpha}^{s_1}(\Naturale k)^{s_2}\pbra{\frac\Naturale{2k^2}}^{2k^3\alpha\cdot s_2}
\tag{assume $k\ge2$}\\
&\le\sum_{s_1\ge\log(n)-1}\sum_{s_2\ge s_1}\frac{2n}{k^2s_1^2}\pbra{\Naturale k^2\alpha}^{s_1}\pbra{\Naturale k}^{s_2}\pbra{\frac\Naturale{2k^2}}^{k^3\alpha\cdot s_1}\pbra{\frac\Naturale{2k^2}}^{k^3\alpha\cdot s_2}
\tag{since $s_2\ge s_1$}\\
&=\sum_{s_1\ge\log(n)-1}\sum_{s_2\ge s_1}\frac{2n}{k^2s_1^2}\pbra{\frac{\Naturale k^2\alpha}{(2k^2/\Naturale)^{k^3\alpha}}}^{s_1}\pbra{\frac{\Naturale k}{(2k^2/\Naturale)^{k^3\alpha}}}^{s_2}
\\
&\le\sum_{s_1\ge\log(n)-1}\sum_{s_2\ge1}\frac{2n\cdot8^{-s_1-s_2}}{k^2s_1^2}
\tag{assume $\frac{\Naturale k^2\alpha}{(2k^2/\Naturale)^{k^3\alpha}}\le\frac18$ and $\frac{\Naturale k}{(2k^2/\Naturale)^{k^3\alpha}}\le\frac18$}\\
&\le\sum_{s_1\ge\log(n)-1}\sum_{s_2\ge1}\frac{2n\cdot8^{-s_1-s_2}}{(\log(n)-1)^2}
\le\sum_{s_1\ge\log(n)-1}\frac{32n\cdot8^{-s_1}}{(\log(n)-1)^2}\\
&\le\frac{4n}{n^2(\log(n)-1)^2}=o(1/n).
\end{align*}
Now we analyze the assumptions. Define $t=k^3\alpha$. Then the calculation above demands $k\ge2$ and
$$
\frac tk\pbra{\frac\Naturale{2k^2}}^t\le\frac1{8\Naturale}
\quad\text{and}\quad
k\pbra{\frac\Naturale{2k^2}}^t\le\frac1{8\Naturale}.
$$
Thus it suffices to assume $k^3\alpha=t\ge1$ and $k\ge30$.

Now we turn to the case $|\Vcal'|<\floorbra{k\log(n)}$.
If $|\tilde\Vcal|<\floorbra{k\log(n)}$, then we are done since $\alpha\ge1/k^3$.
Otherwise consider an arbitrary connected $\hat\Vcal\supset\Vcal'$ such that $|\hat\Vcal|=\floorbra{k\log(n)}$. Then by applying the previous argument on $\hat\Vcal$, we have
\begin{equation*}
|\tilde\Vcal|
\le\abs{\cbra{v\in\Vcal\mid v\in\hat\Vcal\text{ or $v$ is adjacent to }\hat\Vcal}}
\le3k^4\alpha|\hat\Vcal|
=3k^4\alpha\cdot\floorbra{k\log(n)}.
\tag*{\qedhere}
\end{equation*}
\end{proof}

\begin{proof}[Proof of \Cref{prop:maximum_degree}]
The degrees of the variables in $\Phi$ distribute as a balls-and-bins experiment with $km$ balls and $n$ bins. Let $D_1,\ldots,D_n\sim\Poi(k\alpha)$ be $n$ independent Poisson random variables with parameter $k\alpha$.
Then the degrees of the variables in $\Phi$ has the same distribution as $\cbra{D_1,\ldots,D_n}$ conditioned on the event $\Ecal$ that $\sum_{i=1}^nD_i=km$ \cite[Chapter 5.4]{mitzenmacher2017probability}. Note that $\sum_{i=1}^nD_i$ is a Poisson random variable with parameter $k\alpha n=km$. Thus
$$
\Pr\sbra{\Ecal}=\Naturale^{-km}\cdot\frac{(km)^{km}}{(km)!}\ge\frac1{\sqrt{2\pi km}}=\frac1{\sqrt{2\pi k\alpha n}}.
$$
Let $D=4k\alpha+6\log(n)$.
For any fixed $i\in[n]$, we have
\begin{align*}
\Pr\sbra{D_i\ge D}
&=\Pr\sbra{\Poi(k\alpha)\ge D}\le\frac{\Naturale^{-k\alpha}(\Naturale k\alpha)^D}{D^D}
\tag{by \cite[Theorem 5.4]{mitzenmacher2017probability}}\\
&\le\Naturale^{-k\alpha}(\Naturale/4)^D\le\Naturale^{-k\alpha}\cdot2^{-D/2}
\tag{since $D\ge4k\alpha$}\\
&\le \Naturale^{-k\alpha}\cdot n^{-3}.
\tag{since $D\ge6\log(n)$}
\end{align*}
Define $U=\cbra{i\in[n]\mid D_i\ge D}$. 
Then 
\begin{align*}
\Pr\sbra{\exists v\in\Vcal,\deg_\Ccal(v)\ge D}
&=\Pr\sbra{|U|\ge1\mid\Ecal}\le\frac{\Pr\sbra{|U|\ge 1}}{\Pr\sbra{\Ecal}}\\
&\le\sqrt{2\pi k\alpha n}\cdot n\cdot\Pr\sbra{D_i\ge D}
\tag{by Markov's inequality}\\
&\le\sqrt{2\pi k\alpha n}\cdot\Naturale^{-k\alpha}\cdot n^{-2}\\
&=O(1/n^{1.5})=o(1/n).
\tag*{\qedhere}
\end{align*}
\end{proof}

\begin{proof}[Proof of \Cref{prop:high-degree}]
The calculation is similar to the proof of \Cref{prop:maximum_degree}.

Let $D_1,\ldots,D_n\sim\Poi(k\alpha)$ be $n$ independent Poisson random variables with parameter $k\alpha$.
Then the degrees of the variables in $\Phi$ has the same distribution as $\cbra{D_1,\ldots,D_n}$ conditioned on the event $\Ecal$ that $\sum_{i=1}^nD_i=km$. Note that $\sum_{i=1}^nD_i$ is a Poisson random variable with parameter $k\alpha n=km$. Thus
$$
\Pr\sbra{\Ecal}=\Naturale^{-km}\cdot\frac{(km)^{km}}{(km)!}\ge\frac1{\sqrt{2\pi km}}=\frac1{\sqrt{2\pi k\alpha n}}.
$$
For any fixed $i\in[n]$, we have
\begin{align*}
\Pr\sbra{D_i\ge D}
&=\Pr\sbra{\Poi(k\alpha)\ge D}\le\frac{\Naturale^{-k\alpha}(\Naturale k\alpha)^D}{D^D}
\tag{by \cite[Theorem 5.4]{mitzenmacher2017probability}}\\
&\le\Naturale^{-k\alpha}(\Naturale/8)^D\le(\Naturale/8)^D
\tag{assume $D\ge8k\alpha$}\\
&\le2^{-4k-1}.
\tag{assume $D\ge8k$}
\end{align*}
Define $U=\cbra{i\in[n]\mid D_i\ge D}$. 
Then by Chernoff-Hoeffding bound, we have
$$
\Pr\sbra{|U|\ge n/2^{4k}}
\le\Pr\sbra{|U|-\E[|U|]\ge n/2^{4k+1}}
\le\Naturale^{-n/2^{4k+1}}.
$$
Thus
\begin{align*}
&\phantom{=}\Pr\sbra{\abs{\cbra{v\in\Vcal\mid\deg_\Ccal(v)\ge D}}\ge n/2^{4k}}\\
&=\Pr\sbra{|U|\ge n/2^{4k}\mid\Ecal}\le\frac{\Pr\sbra{|U|\ge n/2^{4k}}}{\Pr\sbra{\Ecal}}\\
&\le\sqrt{2\pi k\alpha n}\cdot\Naturale^{-n/2^{4k+1}}\\
&=o(1/n).
\tag{assume $n\ge2^{\Omega(k)}$ and $\alpha\le2^k$}
\end{align*}
Finally we note that if $k\ge2$, $\alpha\le2^k$, $D\ge8k(\alpha+1)$, and $n\ge2^{\Omega(k)}$, then all the assumptions above are satisfied.
\end{proof}

\begin{proof}[Proof of \Cref{prop:fraction_of_high-degrees}]
Let $\tilde\Vcal=\abs{\cbra{v\in\Vcal'\mid\deg_\Ccal(v)\ge D}}$ and $\tilde\Ccal=\cbra{C\in\Ccal\mid\vbl(C)\cap\tilde\Vcal\neq\emptyset}$.
To lower bound $|\tilde\Ccal|$, we perform a double counting for the size of $\cbra{(v,C)\mid v\in\tilde\Vcal,C\in\tilde\Ccal}$, which is lower bounded by $D\cdot|\tilde\Vcal|$ and upper bounded by $k\cdot|\tilde\Ccal|$.
Therefore we have $|\tilde\Ccal|\ge D|\tilde\Vcal|/k$.

By \Cref{prop:high-degree}, we have $|\tilde\Vcal|\le n/2^{4k}$ with probability $1-o(1/n)$. 
Since $D\le2^{2k}$, we have $|\tilde\Ccal|\le D|\tilde\Vcal|\le Dn/2^{4k}\le n/2^{2k/\log(k)}$.
By \Cref{itm:kvars_and_distinct_vars_2} of \Cref{prop:kvars_and_distinct_vars} with $\eta=1$, we have
$$
\abs{\bigcup_{C\in\Ccal:\vbl(C)\cap\Vcal'\neq\emptyset}\vbl(C)}\ge\abs{\bigcup_{C\in\tilde\Ccal}\vbl(C)}\ge\frac{k|\tilde\Ccal|}2\ge D|\tilde\Vcal|/2
$$
with probability $1-o(1/n)$.
On the other hand, by \Cref{prop:number_of_neighbors} we have
$$
\abs{\bigcup_{C\in\Ccal:\vbl(C)\cap\Vcal'\neq\emptyset}\vbl(C)}
\le3k^4\alpha\cdot\max\cbra{|\Vcal'|,k\log(n)}
\le3k^5\alpha\cdot|\Vcal'|
$$
with probability $1-o(1/n)$, where we use the bound $\max\cbra{|\Vcal'|,k\log(n)}\le k|\Vcal'|$ as $|\Vcal'|\ge\log(n)$.
Rearranging and assuming $D\ge6k^7\alpha$, we have
$$
|\tilde\Vcal|\le|\Vcal'|\cdot\frac{3k^5\alpha}{D/2}\le|\Vcal'|/k^2.
$$
Finally we note that if $k\ge2^{10}$, $6k^7(\alpha+1)\le D\le2^{2k}$, $1/k^3\le\alpha\le2^k$, and $n\ge2^{\Omega(k)}$, then all the assumptions used above are satisfied.
\end{proof}
 
\Cref{prop:peeling} is a simple union bound of the following lemma.

\begin{lemma}\label{lem:peeling_param}
Let $\eps=\eps(k,n)$ be a parameter satisfying $1/n\le\eps\le2^{-2.5k}$.
Assume $k\ge12$, $\alpha\le2^k$, and $n\ge2^{\Omega(k)}$.
Then with probability $1-o(1/n^3)$ over the random $\Phi$, the following holds:
Fix an arbitrary $\Ccal'\subseteq\Ccal$ with $|\Ccal'|\le\eps n$. 
Let $C_{i_1},\ldots,C_{i_\ell}\in\Ccal\setminus\Ccal'$ be clauses with distinct indices.
For each $s\in[\ell]$, define $\Vcal_s=\bigcup_{C\in\Ccal'}\vbl(C)\cup\bigcup_{j=1}^{s-1}\vbl(C_{i_j})$.
If $|\vbl(C_{i_s})\cap\Vcal_s|\ge6$ holds for all $s\in[\ell]$, then $\ell\le\eps n$.
\end{lemma}
\begin{proof}
Assume $\Ccal'$ and $C_{i_1},\ldots,C_{i_\ell}$ violates the statement.
By discarding redundant clauses from $C_{i_1},\ldots,C_{i_\ell}$, we assume $\ell=\floorbra{\eps n}+1$.
Now, as long as $|\Ccal'|<\floorbra{\eps n}$ and $\Ccal\setminus\cbra{\Ccal'\cup\cbra{C_{i_1},\ldots,C_{i_\ell}}}$ is not empty, we can enlarge $\Ccal'$ by including new clauses and the statement is still violated.
Therefore we assume $|\Ccal'|=\min\cbra{\floorbra{\eps n},m-\ell}=\min\cbra{\ell-1,m-\ell}$.

Note that the sets $Y=\bigcup_{j=1}^\ell\vbl(C_{i_j})\setminus\bigcup_{C\in\Ccal'}\vbl(C)$ and $\Ccal'$ have the following properties:
\begin{itemize}
\item $|Y|=\sum_{s=1}^\ell\abs{\vbl(C_{i_s})}-\abs{\vbl(C_{i_s})\cap\Vcal_s}\le(k-6)\ell$. 

This is because each $C_{i_s}$ intersects $\Vcal_s$ with at least $6$ variables.
\item There exists $\tilde\Ccal\subset\Ccal\setminus\Ccal'$ with $|\tilde\Ccal|=\ell$ such that $\vbl(\tilde C)\subseteq Y\cup\bigcup_{C\in\Ccal'}\vbl(C)$ for all $\tilde C\in\tilde\Ccal$.

This is because we can pick $\tilde\Ccal=\cbra{C_{i_1},\ldots,C_{i_\ell}}$.
\end{itemize}
Now for any fixed $\Ccal',Y,\tilde\Ccal$ satisfying $|\Ccal'|=\min\cbra{\ell-1,m-\ell}$, $|\tilde\Ccal|=\ell$, and $|Y|=t\le(k-6)\ell$. We define event $\Ecal(\Ccal',Y,\tilde\Ccal)$ to be ``$\vbl(\tilde C)\subseteq Y\cup\bigcup_{C\in\Ccal'}\vbl(C)$ for all $\tilde C\in\tilde\Ccal$''.
Then
$$
\Pr\sbra{\Ecal(\Ccal',Y,\tilde\Ccal)}
\le\pbra{\frac{k|\Ccal'|+|Y|}n}^{k|\tilde\Ccal|}
\le\pbra{\frac{k(\ell-1)+(k-6)\ell}n}^{k\ell}
\le(4k\eps)^{k\ell},
$$
where the last inequality is due to $\ell\le\eps n+1\le2\eps n$.
Therefore by union bound, we have
$$
\Pr\sbra{\exists\text{ such }\Ecal(\Ccal',Y,\tilde\Ccal)}
\le\sum_{t=0}^{(k-6)\ell}{\binom m{\min\cbra{\ell-1,m-\ell}}}^2\cdot\binom nt\cdot(4k\eps)^{k\ell}.
$$
Note that $(k-6)\ell\le(k-6)(\eps n+1)\le2k\eps n\le n/2$ assuming $\eps\le1/(4k)$. Thus $\binom nt\le\binom n{(k-6)\ell}\le\pbra{\frac{\Naturale n}{(k-6)\ell}}^{(k-6)\ell}$.
Also both $\binom m{\ell-1}$ and $\binom m{m-\ell}$ are upper bounded by $\pbra{\frac{\Naturale m}{\ell-1}}^\ell=\pbra{\frac{\Naturale\alpha n}{\ell-1}}^\ell$.
Then we have
\begin{align*}
\Pr\sbra{\exists\text{ such }\Ecal(\Ccal',Y,\tilde\Ccal)}
&\le n\cdot\pbra{\frac{\Naturale\alpha n}{\ell-1}}^{2\ell}\cdot\pbra{\frac{\Naturale n}{(k-6)\ell}}^{(k-6)\ell}\cdot(4k\eps)^{k\ell}\\
&\le
n\cdot\pbra{\frac{\Naturale^{k-4}\cdot2^{4k}\cdot n^{k-4}\cdot k^k\cdot\eps^k}{(\ell-1)^2\cdot\ell^{k-6}\cdot(k-6)^{k-6}}}^\ell
\tag{since $m=\alpha n$ and $\alpha\le2^k$}\\
&\le
n
\cdot\pbra{
\frac{\Naturale^{k-4}\cdot2^{4k}\cdot n^2\cdot \eps^6\cdot k^k}{(\ell-1)^2(k-6)^{k-6}}}^\ell
\tag{since $\ell=\floorbra{\eps n}+1\ge\eps n$}\\
&\le
n
\cdot\pbra{
\frac{\Naturale^{k-4}\cdot2^{4k+2}\cdot\eps^4\cdot k^k}{(k-6)^{k-6}}}^\ell
\tag{since $\ell-1=\floorbra{\eps n}\ge\eps n/2$}\\
&\le
n\cdot\pbra{
\Naturale^{k-4}\cdot2^{4k+14}\cdot\eps^4\cdot k^6}^\ell
\tag{since $\frac{k^k}{(k-6)^{k-6}}\le(4(k-6))^6\le(4k)^6$ for $k\ge12$}\\
&\le
n\cdot\pbra{2^{10k-1}\cdot\eps^4}^\ell
=:\tilde p.
\tag{since $k\ge12$}
\end{align*}
Now we have two cases:
\begin{itemize}
\item If $\eps n\ge5\log(n)$, then assuming $2^{10k-1}\cdot\eps^4\le1/2$, we have 
$$
\tilde p\le n\cdot(1/2)^\ell\le n\cdot(1/2)^{\eps n}=o(1/n^3).
$$
\item Otherwise $\eps\le5\log(n)/n$. Then assuming $n\ge2^{\Omega(k)}$, we have $2^{10k-1}\eps^4=o(1/n^3)$.
Now since $\eps n\ge1$, we have $\ell\ge2$ and $\tilde p\le
n\cdot o(1/n^3)^2=o(1/n^3)$.
\end{itemize}
Finally we note that if $k\ge12$, $\alpha\le2^k$, $1/n\le\eps\le2^{-2.5k}$, and $n\ge2^{\Omega(k)}$, then all the assumptions above are satisfied.
\end{proof}

Now we put explicit parameters into \Cref{lem:peeling_param} to prove \Cref{prop:peeling}.

\begin{proof}[Proof of \Cref{prop:peeling}]
For each $z\in[n/2^{4k}]$, let $\Ecal_z$ be the event that there exists some $\Ccal'\subset\Ccal$ with $|\Ccal'|=z$ that violates the desired property. 
Now we apply \Cref{lem:peeling_param} with $\eps=z/n$. 
Notice that if $k\ge12$ and $n\ge2^{\Omega(k)}$, then all the assumptions in \Cref{lem:peeling_param} are satisfied.
Thus $\Pr\sbra{\Ecal_z}=o(1/n^3)$.
Then the corollary follows immediately by union bound over all possible $z$ and assuming $n\ge2^{\Omega(k)}$.
\end{proof}

\end{document}